\documentclass[11pt]{article}
\usepackage[margin=1in]{geometry}

\usepackage{subfig}
\usepackage{tikz}
\usetikzlibrary{matrix,arrows,automata}

\usepackage{amsmath,amssymb,amsthm,url}
\usepackage{graphics}
\usepackage{float}
\usepackage{enumerate}
\usepackage{verbatim}
\usepackage{pgf}
\usepackage{comment}

\usepackage{array}

\newtheorem{theorem}{Theorem}[section]
\newtheorem{lemma}{Lemma}[section]
\newtheorem{proposition}{Proposition}[section]

\newtheorem{definition}{Definition}[section]
\newtheorem{corollary}{Corollary}[section]

\newtheorem{example}{Example}[section]

\newcommand{\be}{\begin{equation}}
\newcommand{\ee}{\end{equation}}
\newcommand{\req}[1]{(\ref{#1})}

\newcommand{\im}{{\rm im~  }}

\makeatother

\begin{document}

\title{Flows and Decompositions of Games: \\ Harmonic and Potential Games}

\author{Ozan Candogan, Ishai Menache, Asuman Ozdaglar and Pablo
  A. Parrilo\footnote{All authors are with the Laboratory for
    Information and Decision Systems (LIDS), Massachusetts Institute
    of Technology. E-mails: \texttt{\{candogan, ishai, asuman,
      parrilo\}@mit.edu}.  This research is supported in part by the
    National Science Foundation grants DMI-0545910 and ECCS-0621922,
    MURI AFOSR grant FA9550-06-1-0303, NSF FRG 0757207, 
    by the DARPA ITMANET program, and by a Marie Curie International
    Fellowship within the 7th European Community Framework
    Programme.}}

\date{}

\maketitle

\thispagestyle{empty}

\begin{abstract}
 
In this paper we introduce a novel flow representation for finite
games in strategic form. This representation allows us to develop a
canonical direct sum decomposition of an arbitrary game into three
components, which we refer to as the \emph{potential}, \emph{harmonic}
and \emph{nonstrategic} components. We analyze natural classes of
games that are induced by this decomposition, and
 in particular, focus on games with no harmonic component and games with no potential
component. We show that the first class corresponds to the well-known
\emph{potential games}.
We refer to the second class of games as \emph{harmonic games}, and study the structural and equilibrium properties of this new class of games.

Intuitively, the   potential component of a game captures interactions that can equivalently be represented as a common interest game, while the
harmonic part represents the conflicts between the interests of the players. 
We
make this intuition precise, by studying the properties of these two
classes, and show that indeed they have quite distinct and remarkable
characteristics. For instance, while finite potential games always
have pure Nash equilibria, harmonic games generically never
do. Moreover, we show that the nonstrategic component does not affect
the equilibria of a game, but plays a fundamental role in their
efficiency properties, thus decoupling the location of equilibria and
their payoff-related properties. Exploiting the properties of the
decomposition framework, we obtain explicit expressions for the
projections of games onto the subspaces of potential and harmonic
games.  This enables an extension of the properties of potential and
harmonic games to ``nearby'' games. We exemplify this point by showing
that the set of approximate equilibria of an arbitrary game can be
characterized through the equilibria of its projection onto the set of
potential games.

\end{abstract}

\vskip 1pc


\newpage

\setcounter{page}{1}

\section{Introduction} \label{se:intro}

Potential games play an important role in game-theoretic analysis due
to their desirable static properties (e.g., existence of a pure
strategy Nash equilibrium) and tractable dynamics (e.g.,
convergence of simple user dynamics to a Nash equilibrium); see
\cite{monderer1996pg, monderer1996fpp, Neyman:1997p1698}. However,
many multi-agent strategic interactions in economics and engineering
cannot be modeled as a potential game.

This paper provides a novel flow representation of the preference
structure in strategic-form finite games, which allows for  delineating
the fundamental characteristics in preferences that lead to potential
games. This representation enables us to develop a canonical orthogonal
decomposition of an arbitrary game into a potential component, a
harmonic component, and a nonstrategic component, each with its
distinct properties. The decomposition can be used to define the
``distance" of an arbitrary game to the set of potential games. We use this fact
to describe the approximate equilibria of the original game in terms
of the equilibria of the closest potential game.
 
The starting point is to associate to a given finite game a \emph{game
graph}, where the set of nodes corresponds to the strategy profiles and
the edges represent the ``comparable strategy profiles'' i.e.,
strategy profiles that differ in the strategy of a single player. The
utility differences for the deviating players along the edges define a
flow on the game graph. Although this graph contains strictly less
information than the original description of the game in terms of
utility functions, all relevant strategic aspects (e.g., equilibria)
are captured.

Our first result provides a \emph{canonical decomposition} of an
arbitrary game using tools from the study of flows on graphs (which
can be viewed as combinatorial analogues of the study of vector
fields). In particular, we use the \emph{Helmholtz decomposition
theorem} (e.g., \cite{jiang2008lrc}), which enables the decomposition
of a flow on a graph into three components: globally consistent,
locally consistent (but globally inconsistent), and locally
inconsistent component (see Theorem~\ref{thm:hodge}). The globally
consistent component represents a gradient flow while the locally
consistent flow corresponds to flows around global cycles.  The
locally inconsistent component represents local cycles (or
circulations) around 3-cliques of the graph.  

Our game decomposition has three components: \emph{nonstrategic},
\emph{potential} and \emph{harmonic}.  The first component represents
the ``nonstrategic interactions'' in a game. Consider two games in
which, given the strategies of the other players, each player's
utility function differs by an additive constant. These two games have
the same utility differences, and therefore they have the same flow
representation. Moreover, since equilibria are defined in terms of
utility differences, the two games have the same equilibrium set. We
refer to such games as \emph{strategically equivalent}. We normalize
the utilities, and refer to the utility differences between a game and
its normalization as the nonstrategic component of the game.  Our next
step is to remove the nonstrategic component and apply the Helmholtz
decomposition to the remainder. The flow representation of a game
defined in terms of utility functions (as opposed to preferences) does
not exhibit local cycles, therefore the Helmholtz decomposition yields
the two remaining components of a game: the \emph{potential component}
(gradient flow) and the \emph{harmonic component} (global cycles). The
decomposition result is particularly insightful for bimatrix games
(i.e., finite games with two players, see Section~\ref{se:bimatrix}),
where the potential component represents the ``team part'' of the
utilities (suitably perturbed to capture the utility matrix
differences), and the harmonic component corresponds to a zero-sum
game.

The canonical decomposition we introduce  is illustrated in the
following example.

\begin{example}[Road-sharing game]  
\label{ex:drivergame}
Consider a three-player game, where each player has to choose one of the
two roads $\{0,1\}$. We denote the players by $d_1$, $d_2$ and $s$.
The player $s$ tries to avoid sharing the road with other players: its
payoff decreases by $2$ with each player $d_1$ and $d_2$ who shares
the same road with it.  The player $d_1$ receives a payoff $-1$, if
$d_2$ shares the road with it and $0$ otherwise. The payoff of $d_2$
is equal to negative of the payoff of $d_1$, i.e.,
$u^{d_1}+u^{d_2}=0$.  Intuitively, player $d_1$ tries to avoid player
$d_2$, whereas player $d_2$ wants to use the same road with $d_1$.

In Figure~\ref{fig:drunk}a we present the flow representation for this
game (described in detail in Section~\ref{subsec:gamesNflows}), where
the nonstrategic component has been removed. Figures~\ref{fig:drunk}b
and \ref{fig:drunk}c show the decomposition of this flow into its
potential and harmonic components. In the figure, each tuple $(a,b,c)$
denotes a strategy profile, where player $s$ uses strategy $a$ and
players $d_1$ and $d_2$ use strategies $b$ and $c$ respectively.

\begin{figure}
\begin{center}
\subfloat[{Flow representation of the road-sharing game.}]{
\label{fig:3DGameA}
\scalebox{0.8}{
\begin{tikzpicture}[-,>=stealth',shorten >=1pt,auto,node distance=2.5cm]
  \tikzstyle{every state}=[fill=none,draw=black,text=black]
  \node                (000)                     {$(0,0,0)$};
  \node                (100) [above of=000]       {$(1,0,0)$};
  \node                (110) [above right of=100,node distance=2cm]       {$(1,1,0)$};
   \node                (010) [below of=110]       {$(0,1,0)$};
   \node                (111) [right of=110,node distance=5cm]       {$(1,1,1)$};
  \node                (101) [right of=100,node distance=5cm]       {$(1,0,1)$};
  \node                (011) [right of=010,node distance=5cm]       {$(0,1,1)$};
  \node                (001) [right of=000,node distance=5cm]       {$(0,0,1)$};
  
  \path
%first original graph
 (000) edge [-> ]            node {4} (100)
 (000) edge [-> ]            node {1} (010)
%  (000) edge [->,dashed]            node {1} (010)
 (100) edge [-> ]            node {1} (110)
 (110) edge [-> ]            node {1} (111)
 (111) edge [-> ]            node {1} (101)
 (101) edge [-> ]            node {1} (100)
 
  (111) edge [-> ]            node {4} (011)
  (011) edge [-> ]            node {1} (001)
  (001) edge [-> ]            node {1} (000)
  (010) edge [-> ]            node {1} (011) ;
%    (010) edge [->,dashed ]            node {1} (011) ;
  \end{tikzpicture}
  }
  }
\\
\subfloat[{Potential Component.}]{
\label{fig:3DGameB}
\scalebox{0.8}{
\begin{tikzpicture}[scale=.8,-,>=stealth',shorten >=1pt,auto,node distance=2.5cm]
  \tikzstyle{every state}=[fill=none,draw=black,text=black]
  \node                (000)                     {$(0,0,0)$};
  \node                (100) [above of=000]       {$(1,0,0)$};
  \node                (110) [above right of=100,node distance=2cm]       {$(1,1,0)$};
   \node                (010) [below of=110]       {$(0,1,0)$};
   \node                (111) [right of=110,node distance=5cm]       {$(1,1,1)$};
  \node                (101) [right of=100,node distance=5cm]       {$(1,0,1)$};
  \node                (011) [right of=010,node distance=5cm]       {$(0,1,1)$};
  \node                (001) [right of=000,node distance=5cm]       {$(0,0,1)$};
  
  \path
%first original graph
 (000) edge [-> ]            node {2} (100)
% (000) edge [->,dashed]            node {1} (010)
 (000) edge [->]            node {1} (010)
  (000) edge [->]            node {1} (001)
   (111) edge [-> ]            node {1} (110)
 (111) edge [->]            node {1} (101)
  (111) edge [->]            node {2} (011)
  
   (101) edge [-> ]            node {1} (100)
 (110) edge [-> ]            node {1} (100)
 (001) edge [-> ]            node {1} (011)
% (010) edge [->,dashed]            node {1} (011)
  (010) edge [->]            node {1} (011)
%  (111) edge [->]            node {2} (011)
  ;
  \end{tikzpicture}
  }  
  }
  \qquad
\subfloat[{Harmonic Component.}]{
  \scalebox{0.8}{
\label{fig:3DGameC}
\begin{tikzpicture}[scale=.8,-,>=stealth',shorten >=1pt,auto,node distance=2.5cm]
  \tikzstyle{every state}=[fill=none,draw=black,text=black]
  \node                (000)                     {$(0,0,0)$};
  \node                (100) [above of=000]       {$(1,0,0)$};
  \node                (110) [above right of=100,node distance=2cm]       {$(1,1,0)$};
   \node                (010) [below of=110]       {$(0,1,0)$};
   \node                (111) [right of=110,node distance=5cm]       {$(1,1,1)$};
  \node                (101) [right of=100,node distance=5cm]       {$(1,0,1)$};
  \node                (011) [right of=010,node distance=5cm]       {$(0,1,1)$};
  \node                (001) [right of=000,node distance=5cm]       {$(0,0,1)$};
  
  \path
%first original graph
 (000) edge [-> ]            node {2} (100)
  (100) edge [-> ]            node {2} (110)
  (110) edge [-> ]            node {2} (111)
    (111) edge [-> ]            node {2} (011)  
  (011) edge [-> ]            node {2} (001)
    (001) edge [-> ]            node {2} (000)  
 
  ;
  \end{tikzpicture}
  }
  }
\end{center}
\caption{Potential-harmonic decomposition of the road-sharing game. 
An arrow between two strategy profiles, indicates the improvement direction in the payoff of the player who changes its strategy,   and the associated number quantifies the improvement in its payoff. 
}
\label{fig:drunk}
\end{figure}
\end{example}

These components
induce a direct sum decomposition of the space of games into three respective
subspaces, which we refer to as the \emph{nonstrategic},
\emph{potential} and \emph{harmonic} subspaces, denoted by $\cal N$,
$\cal P$, and $\cal H$, respectively. 
We use these subspaces to define classes of games with distinct equilibrium properties.
We establish that the set of
potential games coincides with the direct sum of the subspaces $\cal
P$ and $\cal N$, i.e., potential games are those with no harmonic
component. 
 Similarly, we define a new class of  games in
which the potential component vanishes as \emph{harmonic
games}. 
The classical rock-paper-scissors and  matching pennies
games are examples of harmonic games.
The decomposition then has the following structure:
\[
\mathcal{P} \quad \: \oplus \: \quad \overbrace{ 
\mathcal{N} \makebox[0pt][r]{$\underbrace{\phantom{\mathcal{P} \quad \: \oplus \: \quad \mathcal{N}}}_{\text{Potential games}}$} 
\quad \: \oplus \: \quad \mathcal{H}}^{\text{Harmonic games}}.
\]  

Our second set of results establishes   properties of potential and harmonic games and examines how the nonstrategic component of a game affects the efficiency of equilibria.
Harmonic games can be characterized by the existence of
improvement cycles, i.e., cycles in the game graph, where at each step
the player that changes its action improves its payoffs.
We show that harmonic games generically do not have pure Nash
equilibria. Interestingly, for the special case when the number of
strategies of each player is the same, a harmonic game satisfies a
``multi-player zero-sum property'' (i.e., the sum of utilities of all
players is equal to zero at all strategy profiles).  
 We also study the mixed Nash and
correlated equilibria of harmonic games. We show that the uniformly
mixed strategy profile (see Definition \ref{def:unifMixed})  is always a mixed Nash equilibrium and if there are  two
players in the game, the set of mixed Nash equilibria generically coincides with
the set of correlated equilibria. We finally focus on the nonstrategic
component of a game. As discussed above, the nonstrategic component
does not affect the equilibrium set. Using this property, we show that
by changing the nonstrategic component of a game, it is possible to
make the set of Nash equilibria coincide with the set of Pareto
optimal strategy profiles in a game.

Our third set of results focuses on the \emph{projection of a game
  onto its respective components}. We first define a natural inner product
and show that under this inner product the components in our
decomposition are orthogonal. We further provide explicit expressions
for the closest potential and harmonic games to a game with respect to
the norm induced by the inner product. We use the distance of a game
to its closest potential game to characterize the approximate
equilibrium set in terms of the equilibria of the potential game.

 The decomposition framework in this paper leads to the identification of subspaces of games with distinct and tractable equilibrium properties. 
Understanding the structural properties of these subspaces and the classes of games they induce, provides new insights and tools   for analyzing the static and dynamical properties of general noncooperative games; further implications are outlined in Section \ref{se:conclusions}.

\paragraph{Related literature}
Besides the works already mentioned, our paper is also related to
several papers in the cooperative and noncooperative game theory
literature:
\begin{itemize}

\item The idea of decomposing a game
  (using different approaches) into simpler games which admit more
  tractable equilibrium analysis  appeared even in the early works in the cooperative
  game theory literature. In  \cite{VonNeumann:1947p3802},
  the authors propose to decompose games with large number of players
  into games with fewer players. In
  \cite{Marinacci:1996p1157,Gilboa:1995p5309,Shapley:1997p5242},   a
  different approach is followed: the authors identify cooperative games
  through the games' value functions (see \cite{VonNeumann:1947p3802})
  and obtain decompositions of the value function into simpler
  functions. By defining the component games using the simpler value
  functions, they obtain decompositions of games. In this approach,
  the set of players is not made smaller or larger by the
  decomposition but the component games have    simpler structure.
  Another method for decomposing the space of cooperative games
  appeared in
  \cite{Kleinberg:1985p1179,Kleinberg:1986p1188,KleinbergJeffrey:1985p4021}.
  In these papers, the algebraic properties of the space of games and
  the properties of the nullspace of the Shapley value operator (see
  \cite{Shapley:1997p5242}) and its orthogonal complement are
  exploited to decompose games. This approach does not necessarily
  simplify the analysis of games but it leads to an alternative
  expression for the Shapley value
  \cite{KleinbergJeffrey:1985p4021}. Our work is on decomposition of
  noncooperative games, and different from the above references since we
  explicitly exploit the properties of noncooperative games in our
  framework.

% authors present different decompositions of space of games.

\item In the context of noncooperative game theory, a decomposition
  for games in normal form appeared in \cite{sandholm2008decompositions}.
In this paper, the author  defines a component game for
 each subset of players and obtains a decomposition of normal form
 games with $M$ players to $2^M$ component games.  This method does
 not provide any insights about the properties of the component games,
 but  yields alternative tests to check whether a game is a
 potential game or not.  We note that our decomposition approach is
 different than this work in the properties of the component
 games. In particular, using the global preference structure in games,
 our approach yields decomposition of games to three components with
 distinct equilibrium properties, and these properties can be
 exploited to gain insights about the static and dynamic features of
 the original game.

\item Related ideas of representing finite strategic form games as
  graphs previously appeared in the literature to study different
  solution concepts in normal form games \cite{Goemans:2005p4142,
    christodoulou2006caa}. In these references, the authors focus on
  the restriction of the game graph to best-reply
  paths and analyze   the outcomes of  games using this subgraph.

\item In our work, the graph representation of games and the flows
  defined on this graph lead to a natural equivalence relation.
  Related notions of strategic equivalence are employed in the game
  theory literature to generalize the desirable  static and dynamic properties of games to
  their equivalence classes
  \cite{Moulin:1978p2968,Rosenthal:1974p3448,
    Morris:2004p2154,voorneveld2000brp,Germano:2006p3405,Hammond:2005p3676,Hofbauer:2005p4406,Kannan:2010p4099,Mertens:2004p4202,
    Anonymous:2001p3670}.  In \cite{Moulin:1978p2968}, the authors refer
  to games which have the same better-response correspondence as
  equivalent games and study the equilibrium properties of games which
  are equivalent to zero-sum games.  In
  \cite{Hammond:2005p3676,Hofbauer:2005p4406}, the dynamic and static
  properties of certain classes of bimatrix games are generalized to
  their equivalence classes.  Using the best-response correspondence
  instead of the better-response correspondence, the papers
  \cite{Rosenthal:1974p3448, Morris:2004p2154,voorneveld2000brp}
  define different equivalence classes of games.  We note that the
  notion of strategic equivalence used in our paper  implies 
  some of the equivalence notions mentioned above.  However, unlike
  these papers, our notion of strategic equivalence leads to a
  canonical decomposition of the space of games, which is then used to extend 
   the desirable
  properties of potential games to ``close'' games that are not strategically
  equivalent.

\item Despite the fact that harmonic games were not defined in the
  literature before (and thus, the term ``harmonic'' does not appear
  explicitly as such), specific instances of harmonic games were
  studied in different contexts. In \cite{Hofbauer:2000p2212}, the authors
  study dynamics in ``cyclic games'' and obtain results about a class
  of harmonic games which generalize the matching pennies game.  A
  parametrized version of Dawkins' battle of the sexes game, which is
  a harmonic game under certain conditions, is studied in
  \cite{MAYNARD:2010p2640}. Other examples of harmonic games have also
  appeared in the buyer/seller game of \cite{Friedman:1991p2328} and
the  crime deterrence game of \cite{Cressman:1998p2273}.

\end{itemize}

% % % % % % % % % % % % % % % % % % %

\paragraph{Structure of the paper} 
The remainder of this paper is organized as follows. In
Section~\ref{se:PotGames}, we present the relevant game theoretic
background and provide a representation of games in terms of graph
flows. In Section~\ref{se:Hodge}, we state the Helmholtz decomposition
theorem which provides the means of decomposing a flow into orthogonal
components. In Section~\ref{se:canonicalRep}, we use this machinery to
obtain a canonical decomposition of the space of games. We introduce  in Section~\ref{se:decProp}
natural classes of games, namely potential and harmonic games, which are
induced by this decomposition and describe
the equilibrium properties thereof. In
Section~\ref{se:projection}, we define an inner product for the space
of games, under which the components of  games turn out to be
orthogonal. Using this inner product and our decomposition framework
we propose a method for projecting a given game to the spaces of
potential and harmonic games. We then apply the projection to study the
equilibrium properties of ``near-potential'' games. We close in
Section~\ref{se:conclusions} with concluding remarks and directions
for future work.

\section{Game-Theoretic Background} 
\label{se:PotGames}

In this section, we describe the required game-theoretic
background. Notation and basic definitions are given in
Section~\ref{subsec:prelim}. In Section~\ref{subsec:gamesNflows}, we
provide an alternative representation of games in terms of flows on
graphs. This representation is used in the rest of the paper to
analyze finite games.

\subsection{Preliminaries} 
\label{subsec:prelim}

A (noncooperative) \emph{strategic-form finite game} consists of:

\begin{itemize}

\item A finite set of players, denoted ${\cal M}=\{1, \ldots, M\}$.

\item {Strategy spaces:} A finite set of strategies (or actions)
  $E^m$, for every $m\in {\cal M}$. The joint strategy space is
  denoted by $E=\prod_{m\in{\cal M}} E^m$.
\item {Utility functions:} $u^m:E\rightarrow {\mathbb R}$, $m\in {\cal M}$.
\end{itemize}
A (strategic-form) game instance is accordingly given by the tuple
$\langle {\cal M},\{E^m\}_{m\in{\cal M}},\{u^m \}_{m\in{\cal M}}
\rangle$, which for notational convenience  will often be abbreviated to
$\langle {\cal M},\{E^m\},\{u^m \}\rangle$.

We use the notation ${\bf p}^m\in E^m$ for a strategy of player $m$.
A collection of players' strategies is given by ${\bf p}=\{{\bf
  p}^m\}_{m\in{\cal M}}$ and is referred to as a {strategy profile}. A
collection of strategies for all players but the $m$-th one is denoted
by ${\bf p}^{-m} \in E^{-m}$.  We use $h_m=|E^m|$ for the cardinality
of the strategy space of player $m$, and $|E|=\prod_{m=1}^M h_m$ for
the overall cardinality of the strategy space. As an alternative
representation, we shall sometimes enumerate the actions of the
players, so that $E^m = \{1,\ldots, h_m\}$.

The basic solution concept in a noncooperative game is that of a
\emph{Nash Equilibrium} (NE).  A (pure) Nash equilibrium is a strategy
profile from which no player can unilaterally deviate and improve its
payoff. Formally, a strategy profile ${ \bf p} \triangleq \{ { \bf
  p}^1,\dots,{ \bf p }^M \}$ is a Nash equilibrium if
\begin{equation} 
\label{eq:nash_basic}
u^m({\bf p }^m,{ \bf p}^{-m}) \geq u^m({\bf q }^m,{ \bf p}^{-m}),
\quad \mbox{for every ${{ \bf q }^m \in E^m} $ and } m \in
\mathcal{M}.
\end{equation}

To address strategy profiles that are approximately a Nash
equilibrium, we introduce the concept of $\epsilon$-equilibrium. A
strategy profile ${ \bf p} \triangleq \{{ { \bf p}}^1,\dots,{{ \bf p}
}^M\}$ is an $\epsilon$-equilibrium if
\begin{equation} 
\label{eq:epsNash_basic}
u^m({\bf p }^m,{ \bf p}^{-m})\geq u^m({ \bf q}^m,{ { \bf
    p}}^{-m})-\epsilon \quad \mbox{for every ${ \bf q}^m \in E^m$ and
} m \in \mathcal{M}.
\end{equation}
Note that a Nash equilibrium is an $\epsilon$-equilibrium with $\epsilon=0$.

The next lemma shows that the $\epsilon$-equilibria of two games can
be related in terms of the differences in utilities.
\begin{lemma} 
\label{lemma:epsEqPre}
Consider two games ${\cal G}$ and $\hat{\cal G}$, which differ only in
their utility functions, i.e., ${\cal G}=\langle {\cal M},\{E^m\},\{u^m
\}\rangle$ and $\hat{\cal G}=\langle {\cal
  M},\{E^m\},\{\hat{u}^m\}\rangle$.  Assume that $|u^m({\bf
  p})-\hat{u}^m({\bf p})| \leq \epsilon_0$ for every $m\in {\cal M}$
and ${\bf p}\in E$.  Then, every $\epsilon_1$-equilibrium of
$\hat{\cal G}$ is an $\epsilon$-equilibrium of $\cal G$ for some
$\epsilon \leq 2\epsilon_0 +\epsilon_1$ (and viceversa).
\end{lemma}
\begin{proof}
Let ${\bf p}$ be an $\epsilon_1$-equilibrium of $\hat{\cal G}$ and let
${\bf q} \in E$ be a strategy profile with ${\bf q}^k \neq {\bf p}^k$
for some $k \in {\cal{M}}$, and ${\bf q}^m = {\bf p}^m$ for every $m
\in {\cal{M}} \setminus \{k\} $.  Then,
$${u}^{k}({\bf q}) -{u}^{k}({\bf p}) \leq {u}^{k}({\bf q})
-{u}^{k}({\bf p}) - (\hat{u}^{k}({\bf q}) - \hat{u}^{k}({\bf p}) ) +
\epsilon_1 \leq 2 \epsilon_0 + \epsilon_1,$$ where the first
inequality follows since $\bf p$ is an $\epsilon_1$-equilibrium of $
\hat{\cal G}$, hence $\hat{u}^{k}({\bf p})-\hat{u}^{k}({\bf q})\geq
-\epsilon_1$, and the second inequality follows by the lemma's
assumption.
\end{proof}

We turn now to describe a particular class of games that is central in
this paper, the class of potential games \cite{monderer1996pg}.
\begin{definition}[Potential Game] 
\label{def:ExactPot}
A potential game is a noncooperative game for which there exists a
function $\phi:E \rightarrow \mathbb{R}$ satisfying
\begin{equation} \label{eq:condForExact}
\phi({ \bf p}^m,{ \bf p}^{-m})-  \phi({ \bf q}^m,{ \bf p}^{-m})=
u^m({ \bf p}^m,{ \bf p}^{-m})-u^m({ \bf q}^m,{ \bf p}^{-m}),
\end{equation}
for every $m\in{\cal M}$, ${ \bf p}^m, { \bf q}^m \in E^m $, ${ \bf
  p}^{-m} \in E^{-m}$. The function $\phi$ is referred to as a
\emph{potential} function of the game.
\end{definition}

Potential games can be regarded as games in which the interests of the
players are aligned with a global potential function $\phi$.  Games
that obey condition \eqref{eq:condForExact} are also known in the
literature as \emph{exact} potential games, to distinguish them from
other classes of games that relate to a potential function (in a
different manner).  For simplicity of exposition, we will often write
`potential games' when referring to exact potential games.  Potential
games have desirable equilibrium and dynamic properties as summarized
in Section \ref{se:dynamicsNpotential}.

\subsection{Games and Flows on Graphs} 
\label{subsec:gamesNflows}

In noncooperative games, the utility functions capture the preferences
of agents at each strategy profile.  Specifically, the payoff
difference $[u^m({\bf p}^m,{\bf p}^{-m})-u^m({\bf q}^m, {\bf
    p}^{-m})]$ quantifies by how much player $m$ prefers strategy
${\bf p}^m$ over strategy ${\bf q}^m$ (given that others play ${\bf
  p}^{-m}$). Note that a Nash equilibrium is defined in terms of
payoff differences, suggesting that actual payoffs in the game are not
required for the identification of equilibria, as long as the payoff
differences are well defined.

A pair of strategy profiles that differ only in the strategy of a
single player will be henceforth referred to as \emph{comparable
  strategy profiles}. We denote the set (of pairs) of comparable
strategy profiles by $A \subset E\times E$, i.e., ${\bf p}, {\bf q}$
are comparable if and only if $({\bf p},{\bf q})\in A$.  A pair of
strategy profiles that differ only in the strategy of player $m$ is
called a pair of \emph{$m$-comparable strategy profiles}. The set of
pairs of $m$-comparable strategies is denoted by $A^m \subset E\times
E$. Clearly, $\cup_m A^m = A$, where $A^m \cap A^k= \emptyset$ for any
two different players $m$ and $k$.

For any given $m$-comparable strategy profiles $\bf p$ and $\bf q$,
the difference $[u^m({\bf p})-u^m({\bf q})]$ would be henceforth
identified as their \emph{pairwise comparison}.  For any game, we
define the \emph{pairwise comparison} function $X: E\times E
\rightarrow \mathbb{R}$ as follows
\begin{equation} 
X({\bf p},{\bf q})= \left\{
\begin{aligned}
u^m({\bf q})-u^m({\bf p}) & \quad\quad \mbox{  if $({\bf p},{\bf q})$ are
$m$-comparable for some $m\in {\cal M}$ } \\
0 & \quad\quad\mbox{ otherwise.}
\end{aligned} \right.
\label{eq:pairwiseCompsGame}
\end{equation}
In view of Definition~\ref{def:ExactPot}, a game is an exact potential
game if and only if there exists a function $\phi:E \rightarrow \mathbb{R}$
such that $\phi({\bf q})-\phi({\bf p})=X({\bf p},{\bf q})$ for any
comparable strategy profiles ${\bf p}$ and ${\bf q}$. Note that the
pairwise comparisons are uniquely defined for any given game. However,
the converse is not true in the sense that there are infinitely many
games that correspond to given pairwise comparisons. We exemplify this
below.

\begin{example}  
\label{ex:strVsNonStr}
Consider the payoff matrices of the two-player games in Tables
\ref{tab:matchingPenniesL} and \ref{tab:matchingPenniesR}. For a given
row and column, the first number denotes the payoff of the row player,
and the second number denotes the payoff of the column player.  The
game in Table~\ref{tab:matchingPenniesL} is the %``matching pennies"
``battle of the sexes" game, and the game in
\ref{tab:matchingPenniesR} is a variation in which the payoff of the
row player is increased by $1$ if the column player plays $O$.
\begin{table}[H]
\centering
\subfloat[Battle of the sexes]{ \label{tab:matchingPenniesL}
\begin{tabular}{ | c | c | c |}
\hline
 & O & F  \\ \hline
O & 3, 2 &  0, 0  \\ \hline
F &  0, 0 & 2, 3 \\ \hline
\end{tabular}
}
\quad\quad\quad
\subfloat[Modified battle of the sexes]{ \label{tab:matchingPenniesR}
\begin{tabular}{ | c | c | c |}
\hline
 & O & F  \\ \hline
O & 4, 2 &  0, 0  \\ \hline
F &  1, 0 & 2, 3 \\ \hline
\end{tabular}
}
\label{tab:MatchingPennies}
\end{table}
It is easy to see that these two games have the same pairwise
comparisons, which will lead to identical equilibria for the two
games: $(O,O)$ and $(F,F)$.  It is only the actual equilibrium payoffs
that would differ. In particular, in the equilibrium $(O,O)$, the
payoff of the row player is increased by $1$.
\end{example}
The usual solution concepts in games (e.g., Nash, mixed Nash,
correlated equilibria) are defined in terms of pairwise comparisons
only.  Games with identical pairwise comparisons share the same
equilibrium sets. Thus, we refer to games with identical pairwise
comparisons as \emph{strategically equivalent games}.

By employing the notion of pairwise comparisons, we can concisely
represent any strategic-form game in terms of a \emph{flow} in a
graph. We recall this notion next. Let $G=(N,L)$ be an undirected
graph, with set of nodes $N$ and set of links $L$.  An \emph{edge
  flow} (or just \emph{flow}) on this graph is a function $Y:N \times
N \rightarrow \mathbb{R}$ such that $Y({\bf p},{ \bf q})=-Y({\bf q},{
  \bf p})$ and $Y({\bf p},{ \bf q})=0$ for $({\bf p},{ \bf q})\notin
L$ \cite{jiang2008lrc,bertsimas1998ilo}.  Note that the flow conservation
equations are not enforced under this general definition.

Given a game $\mathcal{G}$, we define a graph where each node
corresponds to a strategy profile, and each edge connects two
comparable strategy profiles.  This undirected graph is referred to as
the \emph{game graph} and is denoted by $G(\mathcal{G}) \triangleq (E,A)$,
where $E$ and $A$ are the strategy profiles and pairs of comparable
strategy profiles defined above, respectively. Notice that, by
definition, the graph $G(\mathcal{G})$ has the structure of a direct
product of $M$ cliques (one per player), with clique $m$ having $h_m$
vertices. The pairwise comparison function $X:E \times E \rightarrow
\mathbb{R}$ defines a {flow} on $G(\mathcal{G})$, as it satisfies
$X({\bf p},{ \bf q})=-X({\bf q},{ \bf p})$ and $X({\bf p},{ \bf q})=0$
for $({\bf p},{ \bf q})\notin A$. This flow may thus serve as an
equivalent representation of any game (up to a ``non-strategic"
component). It follows directly from the statements above that two
games are strategically equivalent if and only if they have the same
flow representation and game graph.

Two examples of game graph representations are given below.
\begin{example}
\label{ex:ex1gamegraph}
Consider again the ``battle of the sexes'' game from
Example~\ref{ex:strVsNonStr}. The game graph has four
vertices, corresponding to the direct product of two 2-cliques, and is
presented in Figure~\ref{fig:bosex}.

\begin{figure}[H]
\begin{center}
\begin{tikzpicture}[-,>=stealth',shorten >=1pt,auto,node distance=2cm]
  \tikzstyle{every state}=[fill=none,draw=black,text=black]
  \node                (11)                     {$(O,O)$};
  \node                (12) [right of=11]       {$(O,F)$};
  \node                (21) [below of=11]       {$(F,O)$};
  \node                (22) [right of=21]       {$(F,F)$};
  \path
%first user 1's graph
 (11) edge [<- ]            node {3} (21)
 (12) edge [-> ]            node {2} (22)
%second user 2s graph
 (11) edge [<-]                node {2} (12)
 (21) edge [->]                node {3} (22);
\end{tikzpicture}
\caption{Flows on the game graph corresponding to ``battle of the
  sexes'' (Example~\ref{ex:ex1gamegraph}). }
\label{fig:bosex}
\end{center}
\end{figure}
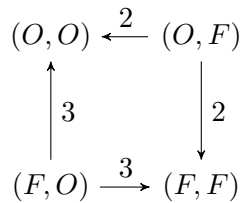
\end{example}

\begin{example}  
\label{ex:graph3game}
Consider a three-player game, where each player can choose between two
strategies $\{a,b\}$. We represent the strategic interactions among
the players by the directed graph in~Figure \ref{fig:3playGameA},
where the payoff of player $i$ is $-1$ if its strategy is identical to
the strategy of its successor (indexed $[i\mod 3 +1]$), and $1$
otherwise. Figure~\ref{fig:3playGameB} depicts the associated game
graph and pairwise comparisons of this game, where the arrow direction
corresponds to an increase in the utility by the deviating player.
The numerical values of the flow are omitted from the figure, and are
all equal to $2$; thus notice that flow conservation does not
hold. The highlighted cycle will play an important role later, after
we discuss potential games.

\begin{figure}[h]
\begin{center}
\subfloat[Player Interaction Graph]{
\label{fig:3playGameA}
\begin{tikzpicture}[->,>=stealth',shorten >=.8pt,auto,node distance=1.8cm,                    semithick]
 \tikzstyle{every state}=[fill=none,draw=black,text=black]
 \node[state]         (A)                    {$ 1$};
 \node[state]         (B) [below left of=A] {$2$};
 \node[state]         (C) [below right of=A] {$3$};
 \path
(A) edge              node {} (B)
(B) edge               node {} (C)
(C) edge              node {} (A)
;
\end{tikzpicture}
}
\quad\quad\quad
\subfloat[Flows on the game graph]{
\label{fig:3playGameB}
\scalebox{0.8}{
\begin{tikzpicture}[scale=.8,-,>=stealth',shorten >=1pt,auto,node distance=2.5cm]
  \tikzstyle{every state}=[fill=none,draw=black,text=black]
  \node                (000)                     {$(a,a,a)$};
  \node                (100) [above of=000]       {$(b,a,a)$};
  \node                (110) [above right of=100,node distance=2cm]       {$(b,b,a)$};
   \node                (010) [below of=110]       {$(a,b,a)$};
   \node                (111) [right of=110,node distance=5cm]       {$(b,b,b)$};
  \node                (101) [right of=100,node distance=5cm]       {$(b,a,b)$};
  \node                (011) [right of=010,node distance=5cm]       {$(a,b,b)$};
  \node                (001) [right of=000,node distance=5cm]       {$(a,a,b)$};
  
  \path
%first original graph
 (000) edge [-> ]              (100)
 (000) edge [-> ]            (010)
%  (000) edge [->,dashed]            (010)
  (000) edge [->]              (001)
   (111) edge [-> ]            (110)
 (111) edge [->]             (101)
  (111) edge [->]             (011)
  
   (101) edge [->,line width=2pt, color=red ]              (100)
 (100) edge [->,line width=2pt, color=red  ]             (110)
 (011) edge [->,line width=2pt, color=red  ]             (001)
 (010) edge [->, line width=2pt, color=red ]             (011)
 (110) edge [-> ,line width=2pt, color=red ]             (010)
 %(010) edge [->,dashed,line width=2pt, color=red ]             (011)
% (110) edge [->,dashed,line width=2pt, color=red ]             (010)
  (001) edge [->,line width=2pt, color=red  ]             (101)
  ;
  \end{tikzpicture}
  } 
}
\end{center}
\caption{A three-player game, and associated flow on its game graph. Each arrow designates an improvement in the payoff of the agent who unilaterally modifies its strategy.
The highlighted cycle implies that in this game, there can be an infinitely long sequence of profitable unilateral  deviations.}
\label{fig:3playGame}
\end{figure}
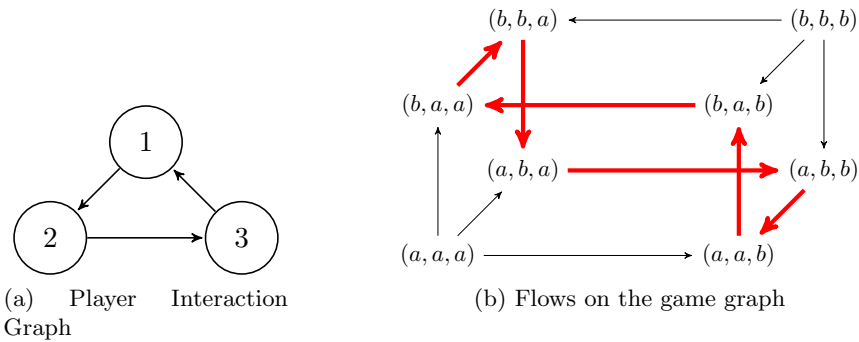
\end{example}

The representation of a game as a flow in a graph is natural and
useful for the understanding of its strategic interactions, as it
abstracts away the absolute utility values and allows for more direct
equilibrium-related interpretation. In more mathematical terms, it
considers the quotient of the utilities modulo 
the subspace of games that are ``equivalent''  to the trivial game (the game where all players receive zero payoff at all strategy profiles), and allows for the identification of ``equivalent'' games
as the same object, a point explored in more detail in later
sections. The game graph also contains much structural
information. For example, the highlighted sequence of arrows in
Figure~\ref{fig:3playGameB} forms a directed cycle, indicating that no
strategy profile within that cycle could be a pure Nash
equilibrium. Our goal in this paper is to use tools from the theory of
graph flows to decompose a game into components, each of which admits
tractable equilibrium characterization. The next section provides an
overview of the tools that are required for this objective.

\section{Flows and Helmholtz Decomposition} 
\label{se:Hodge}

The objective of this section is to provide a brief overview of the
notation and tools required for the analysis of flows on graphs.  The
basic high-level idea is that under certain conditions (e.g., for
graphs arising from games), it is possible to consider graphs as
natural topological spaces with nontrivial homological
properties. These topological features (e.g., the presence of
``holes'', due to the presence of different players) in turn enable
the possibility of interesting flow decompositions. In what follows,
we make these ideas precise. For simplicity and accessibility to a
wider audience, we describe the methods in relatively elementary
linear algebraic language, limiting the usage of algebraic topology
notions whenever possible. The main technical tool we use is the
Helmholtz decomposition theorem, a classical result from algebraic
topology with many applications in applied mathematics, including
among others electromagnetism, computational geometry and data
visualization; see e.g. \cite{PolthierPreuss,tong2003discrete}. In
particular, we mention the very interesting recent work by Jiang et
al. \cite{jiang2008lrc}, where the Helmholtz/Hodge decomposition is
applied to the problem of statistical ranking for sets of incomplete
data.

Consider an undirected graph $G=(E,A)$, where $E$ is the set of the
nodes, and $A$ is the set of edges of the graph\footnote{The results
discussed in this section apply to arbitrary graphs. We use the
notation introduced in Section~\ref{se:PotGames} since in the rest of
the paper we focus on the game graph introduced there.}.  Since the
graph is undirected $({\bf p},{\bf q})\in A$ if and only if $({\bf
q},{\bf p})\in A$. We denote the set of \emph{$3$-cliques} of the
graph by $T=\{({\bf p},{\bf q},{\bf r})|({\bf p},{\bf q}),({\bf
q},{\bf r}),({\bf p},{\bf r}) \in A \}$.
 
We denote by $C_0=\{f|~ f:E\rightarrow \mathbb{R}\}$ the set of
real-valued functions on the set of nodes.  Recall that the \emph{edge
  flows} $X:E\times E \rightarrow \mathbb{R}$ are functions which
satisfy
\begin{equation} 
\label{eq:pairwiseComps}
X({\bf p},{\bf q})= \left\{
\begin{aligned}
-X({\bf q}, {\bf p}) & \quad\quad \mbox{  if $({\bf p},{\bf q})\in A$ } \\
0 & \quad\quad\mbox{ otherwise.}
\end{aligned} \right.
\end{equation}
Similarly the \emph{triangular flows} $\Psi:E\times E \times E
\rightarrow \mathbb{R}$ are functions for which
\begin{equation}
\Psi({\bf p},{\bf q},{\bf r})=\Psi({\bf q},{\bf r},{\bf p})=\Psi({\bf r},{\bf p},{\bf q})=-\Psi({\bf q},{\bf p},{\bf r})=-\Psi({\bf p},{\bf r},{\bf q})=-\Psi({\bf r},{\bf q},{\bf p}),
\end{equation}
and $\Psi({\bf p},{\bf q},{\bf r})=0$ if $({\bf p},{\bf q},{\bf
  r})\notin T$.  Given a graph $G$, we denote the set of all possible
edge flows by $C_1$ and the set of triangular flows by $C_2$.  Notice
that both $C_1$ and $C_2$ are alternating functions of their
arguments. It follows from \eqref{eq:pairwiseComps} that $X({\bf p},
{\bf p})=0$ for all $X\in C_1$.

The sets $C_0$, $C_1$ and $C_2$ have a natural structure of vector
spaces, with the obvious operations of addition and scalar
multiplication.  In this paper, we use the following inner products:
\begin{equation}
\label{eq:innerProdonC1}
\begin{aligned}
\langle \phi_1, \phi_2 \rangle_0 &= \sum_{{\bf p}\in E} \phi_1({
\bf p}) \phi_2({\bf p}). \\
\langle X, Y \rangle_1 &= \frac{1}{2} \sum_{({\bf p},{\bf q})\in A}  X({\bf p},{\bf q})Y({\bf p},{\bf q}) \\
%\langle X, Y \rangle_1 &= \frac{1}{2}\sum_{({\bf p},{\bf q})\in E\times E} W({\bf p},{\bf q}) X({\bf p},{\bf q})Y({\bf p},{\bf q}) \\
\langle \Psi_1, \Psi_2 \rangle_2 &= \sum_{({\bf p},{\bf q},{\bf r})\in T}  \Psi_1({\bf p},{\bf q},{\bf r})\Psi_2({\bf p},{\bf q},{\bf r}).
\end{aligned}
\end{equation}
We shall frequently drop the subscript in the inner product notation, as
the respective space will  often be clear from the context.

We next define linear operators that relate the above defined objects.
To that end, let $W:E\times E \rightarrow \mathbb{R}$ be an indicator
function for the edges of the graph, namely
\begin{equation} 
\label{eq:defWeights}
W({\bf p},{\bf q})= \left\{
\begin{aligned}
1 & \quad\quad \mbox{  if $({\bf p},{\bf q})\in A$ } \\
0 & \quad\quad\mbox{ otherwise}.
\end{aligned} \right.
\end{equation}
Notice that $W({\bf p},{\bf q})$ can be simply interpreted as the
adjacency matrix of the graph $G$.
 
The first operator of interest is the \emph{combinatorial gradient
  operator} $\delta_0:C_0\rightarrow C_1$, given by
\begin{equation} 
\label{eq:defD0}
(\delta_0 \phi)({\bf p},{\bf q}) = W({\bf p}, {\bf q}) ( \phi({\bf q})-\phi({\bf p})), \quad {\bf p},{\bf q}\in E,
\end{equation}
for $\phi \in C_0$.  An operator which is used in the characterization
of ``circulations'' in edge flows is the \emph{curl operator}
$\delta_1: C_1\rightarrow C_2$, which is defined for all $X \in C_1$
and ${\bf p},{\bf q},{\bf r}\in E$ as
\begin{equation} 
\label{eq:curlExplicit}
(\delta_1 X)({\bf p},{\bf q},{\bf r})=\left\{
    \begin{aligned}
& X({\bf p},{\bf q})+X( {\bf q},{\bf r})+X({\bf r},{\bf p})  \quad &\mbox{ if $({\bf p},{\bf q},{\bf r})\in T $,}\\
&0 \quad\quad\quad &\mbox{otherwise.}
 \end{aligned}
 \right.
\end{equation}

We denote the adjoints of the operators $\delta_0$ and $\delta_1$ by
$\delta_0^*$ and $\delta_1^*$ respectively.  Recall that given
inner products $\langle \cdot, \cdot \rangle_k$ on $C_k$, the adjoint
of $\delta_k$, namely $\delta_k^* : C_{k+1} \rightarrow C_k$, is the
unique linear operator satisfying
\begin{equation} \label{eq:adjoint}
\langle \delta_k f_k, g_{k+1} \rangle_{k+1}= \langle  f_k, \delta_k^* g_{k+1} \rangle_{k},
\end{equation}
for all $f_k \in C_k$, $g_{k+1} \in C_{k+1}$.

Using the definitions in \req{eq:adjoint}, \eqref{eq:defD0} and
\req{eq:innerProdonC1}, it can be readily seen that the adjoint $\delta_0^*:C_1 \rightarrow C_0$ of the
combinatorial gradient $\delta_0$ satisfies
\begin{equation}\label{eq:delta0Explicit}
(\delta_0^* X)({\bf p})=-\sum_{{\bf q}|({\bf p},{\bf q})\in A} X({\bf p},{\bf q})=-\sum_{{\bf q}\in E} W({\bf p},{\bf q}) X({\bf p},{\bf q}).
\end{equation}
Note that $-(\delta_0^* X)({\bf p})$ represents the total flow ``leaving" $\bf p$. 
We shall sometimes refer to the operator $-\delta_0^*$ as the \emph{divergence} operator, due to its similarity to the divergence operator in Calculus.

The domains and codomains of the operators  $\delta_0, \delta_1,\delta_0^*, \delta_1^*$  are summarized below.
\begin{equation}
\label{eq:sumOp}
\begin{aligned}
C_0& \xrightarrow{\delta_0} C_1 \xrightarrow{\delta_1} C_2 \\
C_0& \xleftarrow{\delta_0^*} C_1 \xleftarrow{\delta_1^*} C_2.
\end{aligned}
\end{equation}

We next define the Laplacian operator, $\Delta_0: C_0 \rightarrow C_0$, given by
\begin{equation}
\Delta_0 \triangleq \delta_0^* \circ \delta_0,
\end{equation}
where $\circ$ represents operator composition.  To simplify the
notation, we henceforth omit $\circ$ and write $\Delta_0=\delta_0^*
\delta_0$.  Note that functions in $C_0$ can be represented by vectors
of length $|E|$ by indexing all nodes of the graph and constructing a
vector whose $i$th entry is the function evaluated at the $i$th
node. This allows us to easily represent these operators in terms of
matrices. In particular, the Laplacian can be expressed as a square
matrix of size $|E| \times |E|$; using the definitions for $\delta_0$
and $\delta_0^*$, it follows that
\begin{equation} \label{eq:LapExplicit}
[\Delta_0]_{{\bf p},{\bf q}}= \left\{
\begin{aligned}
\sum_{{\bf r}\in E} W({\bf p}, {\bf r})   & \quad\quad \mbox{  if ${\bf p}= {\bf q}$} \\
- 1 \quad  & \quad\quad \mbox{  if ${\bf p}\neq {\bf q}$ and $({\bf p}, {\bf q})\in A$} \\
0 \quad & \quad\quad \mbox{  otherwise,}
\end{aligned}
\right.
\end{equation}
where, with some abuse of the notation, $[\Delta_0]_{{\bf p},{\bf q}}$
denotes the entry of the matrix $\Delta_0$, with rows and columns
indexed by the nodes ${\bf p}$ and ${\bf q}$. The above matrix
naturally coincides with the definition of a Laplacian of an
undirected graph \cite{chung1997sgt}.

Since the entry of $\Delta_0 \phi$ corresponding to ${\bf p}$ equals
$\sum_{{\bf q}} W({\bf p}, {\bf q})\big (\phi({\bf p})-\phi({\bf
  q})\big )$, the Laplacian operator gives a measure of the aggregate
``value'' of a node over all its neighbors. A related operator is
\begin{equation}
\Delta_1 \triangleq \delta_1^* \circ \delta_1+\delta_0 \circ \delta_0^*,
\end{equation}
known in the literature as the vector Laplacian \cite{jiang2008lrc}.

We next provide additional flow-related terminology which will be used
in association with the above defined operators, and highlight some of
their basic properties. In analogy to the well-known identity in
vector calculus, $\mathrm{curl} \circ \mathrm{grad}=0$, we have that
$\delta_0$ is a closed form, i.e., $\delta_{1} \circ \delta_0=0$. An
edge flow $X \in C_1$ is said to be \emph{globally consistent} if $X$
corresponds to the combinatorial gradient of some $f \in C_0$, i.e.,
$X = \delta_0 f$; the function $f$ is referred to as the
\emph{potential function} corresponding to $X$. Equivalently, the set
of globally consistent edge flows can be represented as the image of
the gradient operator, namely $\im (\delta_0)$. By the closedness of
$\delta_0$, observe that $\delta_1 X=0$ for every globally consistent
edge flow $X$. We define \emph{locally consistent} edge flows as those
satisfying $(\delta_1 X)({\bf p}, {\bf q} ,{\bf r}) = X({\bf p}, {\bf
q})+X({\bf q}, {\bf r})+X({\bf r}, {\bf p}) = 0$ for all $({\bf p},
{\bf q}, {\bf r}) \in T$. Note that the kernel of the curl operator
$\ker (\delta_1)$ is the set of locally consistent edge flows. The
latter subset is generally not equivalent to $\im(\delta_0)$, as there
may exist edge flows that are globally inconsistent but locally
consistent (in fact, this will happen whenever the graph has a
nontrivial topology). We refer to such flows as \emph{harmonic flows}.
Note that the operators $\delta_0$, $\delta_1$ are linear operators,
thus their image spaces are orthogonal to the kernels of their
adjoints, i.e., $\im (\delta_0) \perp \ker (\delta_0^*)$ and $\im
(\delta_1) \perp \ker (\delta_1^*)$ [similarly, $\im (\delta_0^*)
\perp \ker (\delta_0)$ and $\im (\delta_1^*) \perp \ker (\delta_1)$ as
can be easily verified using \eqref{eq:adjoint}].

We state below a basic flow-decomposition theorem, known as the
Helmholtz Decomposition\footnote{The Helmholtz Decomposition can be
generalized to higher dimensions through the Hodge Decomposition
theorem (see \cite{jiang2008lrc}), however this generalization is not
required for our purposes.}, which will be used in our context of
noncooperative games.  The theorem implies that any graph flow can be
decomposed into three orthogonal flows.

\begin{theorem}[Helmholtz Decomposition] 
\label{thm:hodge}
The vector space of edge flows $C_1$ admits an orthogonal
decomposition
\begin{equation}
C_1=\im (\delta_{0}) \oplus \ker (\Delta_1) \oplus \im (\delta_1^*),
\end{equation}
where $\ker (\Delta_1)= \ker (\delta_1)\cap \ker (\delta^*_{0})$.
\end{theorem}

\begin{figure}[ht]
\centering
\includegraphics[width=7cm,height=10cm,angle=-90]{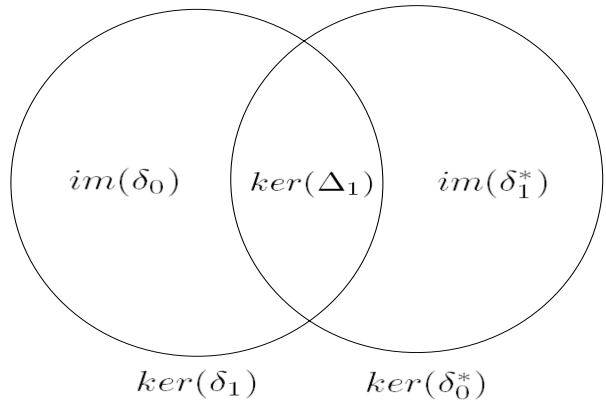}
\vskip -.8cm
\caption{Helmholtz decomposition of $C_1$}
\label{fig:Helmholtz}
\end{figure}

Below we summarize
the interpretation of each of
the components in the Helmholtz decomposition (see also
Figure~\ref{fig:Helmholtz}):
\begin{itemize}
\item $\im (\delta_{0})$ -- globally consistent flows.
\item $\ker (\Delta_1)= \ker (\delta_1)\cap \ker (\delta^*_{0})$ --
  harmonic flows, which are globally inconsistent but locally
  consistent.  Observe that $\ker (\delta_1)$ consists of locally
  consistent flows (that may or may not be globally consistent), while
  $\ker (\delta^*_{0})$ consists of globally inconsistent flows (that
  may or may not be locally consistent).
\item $\im (\delta_1^*)$ (or equivalently, the orthogonal complement
  of  $\ker (\delta_1)$ ) -- locally inconsistent flows.
\end{itemize}

We conclude this section with a brief remark on the decomposition and the flow conservation. For   $X \in C_1$, if  $\delta_0^* X=0$, i.e., if for every node,  the total flow leaving the node is zero, then  we say that $X$ satisfies the flow conservation condition.
Theorem \ref{thm:hodge} implies that $X$ satisfies this condition only when  $X\in \ker(\delta_0^*) =\im(\delta_0)^\perp =\ker(\Delta_1) \oplus \im(\delta_1^*)$. Thus, the  flow conservation condition is satisfied for harmonic flows and locally inconsistent flows but not for globally consistent flows.

\section{Canonical Decomposition of Games} 
\label{se:canonicalRep}

In this section we obtain a canonical decomposition of an arbitrary
game into basic components, by combining the game graph representation
introduced in Section~\ref{subsec:gamesNflows} with the Helmholtz
decomposition discussed above.

Section~\ref{se:prelim_operators_game} introduces the relevant
operators that are required for formulating the results. In
Section~\ref{se:orthDecomp} we provide the basic decomposition
theorem, which states that the space of games can be decomposed as a
direct sum of three subspaces, referred to as the \emph{potential},
\emph{harmonic} and \emph{nonstrategic} subspaces.  In
Section~\ref{se:bimatrix}, we focus on bimatrix games, and provide
explicit expressions for the decomposition.

\subsection{Preliminaries} 
\label{se:prelim_operators_game}

We consider a game $\mathcal{G}$ with set of players $\cal M$,
strategy profiles $E\triangleq E^1 \times \dots \times E^M$, and game
graph $G(\mathcal{G})=(E,A)$.  Using the notation of the previous
section, the utility functions of each player can be viewed as
elements of $C_0$, i.e., $u^m \in C_0$ for all $m\in {\cal M}$. For
given $\cal M$ and $E$, every game is uniquely defined by its set of
utility functions. Hence, the space of games with players $\cal M$ and
strategy profiles $E$ can be identified as ${\cal G}_{{\cal M},E}
\cong C_0^M$. In the rest of the paper we use the notations
$\{u^m\}_{m\in {\cal M}}$ and ${\cal G}=\langle {\cal
M},\{E^m\},\{u^m \} \rangle$
interchangeably when referring to games.

The pairwise comparison function $X(\mathbf{p},\mathbf{q})$ of a game,
defined in \req{eq:pairwiseCompsGame}, corresponds to a flow on the game graph,
and hence it belongs to $C_1$.  In general, the flows representing
games have some special structure. For example, the pairwise
comparison between any two comparable strategy profiles is associated
with the payoff of exactly a single player. It is therefore required
to introduce \emph{player-specific} operators and highlight some
important identities between them, as we elaborate below.

Let $W^m:E\times E \rightarrow \mathbb{R}$ be the indicator function
for $m$-comparable strategy profiles, namely
\begin{equation*}
W^m({\bf p},{\bf q})= \left\{
\begin{aligned}
1 & \quad\quad \mbox{  if ${\bf p},{\bf q}$ are $m$-comparable } \\
0 & \quad\quad\mbox{ otherwise.}
\end{aligned} \right.
\end{equation*}
%for all ${\bf p}, {\bf q}\in E$ and $m\in{\cal M}$.
Recalling that any pair of strategy profiles cannot be comparable by more than a single user, we have
\begin{equation} \label{eq:WmDisW}
W^m({\bf p},{\bf q}) W^k({\bf p},{\bf q})=0, \quad \mbox{for all $k\neq m$ and $\mathbf{p,q}\in E$},
\end{equation}
and
\begin{equation} \label{eq:WmVsW}
W=\sum_{m\in{\cal M}}W^m,
\end{equation}
%and $W^m({\bf p}, {\bf q}) W^k({\bf p}, {\bf q})=0$ for all ${\bf p}, {\bf q}\in E$ and $k,m\in{\cal M}$ such that $k\neq m$.
where $W$ is the indicator function of comparable strategy profiles
(edges of the game graph) defined in \eqref{eq:defWeights}. Note that
this can be interpreted as a decomposition of the adjacency matrix of
$G(\mathcal{G})$, where the different components correspond to the edges
associated with different players.

Given $\phi \in C_0$, we define  $D_m:C_0\rightarrow C_1$ such that
\begin{equation} \label{eq:defDm}
(D_m \phi)({\bf p},{\bf q}) =W^m({\bf p},{\bf q}) \left( \phi({\bf q})-\phi({\bf p}) \right).
\end{equation}
This operator quantifies the change in $\phi$ between strategy
profiles that are $m$-comparable.  Using this operator, we can
represent the pairwise differences $X$ of a game with payoffs
$\{u^m\}_{m\in{\cal M}}$ as follows:
\begin{equation} \label{eq:gameFlow}
X=\sum_{m\in{\cal M}} D_m u^m.
\end{equation}
We  define a relevant operator $D: C_0^M \rightarrow C_1$, such that  $D = [D_1\ldots,D_M]$. As can be seen from \eqref{eq:gameFlow}, for a game with  collection of utilities $u=[u^1; u^2 \dots ; u^M] \in C_0^M$, the pairwise differences can alternatively be represented by  $D u$.

Let
$\Lambda_m:C_1\rightarrow C_1$ be a scaling operator so that
\begin{equation*} 
(\Lambda_m X)({\bf p}, {\bf q})= W^m({\bf p}, {\bf q}) X({\bf p}, {\bf q})
\end{equation*}
for every $X\in C_1$, ${\bf p}, {\bf q} \in E$. 
From \eqref{eq:WmVsW}, it can be seen that for any $X\in C_1$, $\sum_m \Lambda_m X= X$. 
The definition of $\Lambda_m$ and 
\eqref{eq:WmDisW} imply that $\Lambda_m \Lambda_k =0$ for $k\neq m$.
Additionally, the definition of the inner product in $C_1$ implies that for $X,Y\in C_1$, it follows that $\langle \Lambda_m X, Y \rangle= \langle  X, \Lambda_m Y \rangle$, i.e., $\Lambda_m$ is self-adjoint.

This operator provides a convenient description for the operator $D_m$. From the definitions of $D_m$ and $\Lambda_m$, it immediately follows that
$D_m = \Lambda_m \delta_0$, and since $\sum_m \Lambda_m X=X$ for all $X\in C_1$, %$\delta_0= \sum_m \Lambda_m \delta_0= \sum_m D_m$.
%Similarly, it follows from \eqref{eq:defD0}, \req{eq:WmVsW} and %\eqref{eq:defDm} that the operators $\{D_m\}_{m\in {\cal M}}$ and the gradient %operator$\delta_0$ are related through
\begin{equation*} 
\delta_0= \sum_m \Lambda_m \delta_0= \sum_m D_m.
\end{equation*}
Since $\Lambda_m$ is self-adjoint, the adjoint of  $D_m$, which  is denoted by $D_m^*: C_1 \rightarrow C_0$, is given by:
\begin{equation*}
D_m^* =  \delta_0^* \Lambda_m.
\end{equation*}
Using \eqref{eq:delta0Explicit} and the above definitions, it follows that
\begin{equation} \label{eq:DmExplicit}
(D_m^* X) ({\bf p})= -  \sum_{ {\bf q}\in E  } W^m({\bf p}, {\bf q}) X({\bf p}, {\bf q}), \quad \mbox{ for all $X\in C_1$},
\end{equation}
and
\begin{equation} \label{eq:delta0Adj_DmAdj}
\delta_0^*=\sum_{m\in{\cal M}} D_m^*.
\end{equation}

Observe that $D_k^*D_m= \delta_0^* \Lambda_k \Lambda_m \delta_0=0$  for $k\neq
m$. This immediately implies that the image spaces of $\{D_m\}_{m\in
{\cal M}}$ are orthogonal, i.e., $D_k^* D_m =0$. 
Let  $D_m^\dagger$ denote the (Moore-Penrose) pseudoinverse of
$D_m$, with respect to the inner products introduced in Section \ref{se:Hodge}.
By the properties of the pseudoinverse, we have $\ker D_m^\dagger= (\im  D_m)^\perp$. Thus, orthogonality of the image spaces of $D_k$ operators imply that $D_k^\dagger D_m=0$ for $k\neq m$.

The orthogonality
leads to the following expression for the Laplacian operator,
\begin{equation} 
\label{eq:LaplacianVsLm}
\begin{aligned}
\Delta_0 &= \delta_0^* \delta_0 = 
\sum_{k\in{\cal M}} D_k^* \sum_{m\in{\cal M}} D_m =
\sum_{m\in{\cal M}} D_m^*D_m.
\end{aligned}
\end{equation}

In view of \req{eq:defDm} and \req{eq:DmExplicit}, $D_m$ and $-D_m^*$
are the gradient and divergence operators on the graph of
$m$-comparable strategy profiles $(E, A^m)$.  Therefore, the operator
$\Delta_{0,m} \triangleq D_m^* D_m$ is the Laplacian of the graph induced
by $m$-comparable strategies, and is referred to as the
\emph{Laplacian operator of the $m$-comparable strategy profiles}.  It
follows from \eqref{eq:LaplacianVsLm} that
\begin{equation*} 
\Delta_0=\sum_{m\in{\cal M}}\Delta_{0,m}.
\end{equation*}
The relation between the Laplacian operators $\Delta_0$ and $\Delta_{0,m}$ is illustrated in Figure \ref{fig:graphLap}.

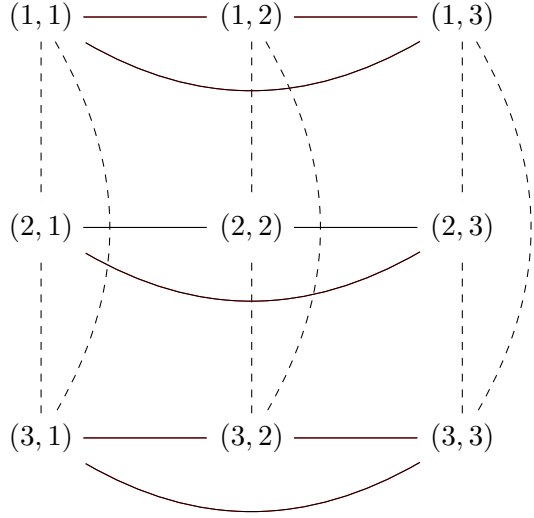
\begin{figure}
\begin{center}
\begin{tikzpicture}[-,>=stealth',shorten >=1pt,auto,node distance=2.8cm]
  \tikzstyle{every state}=[fill=none,draw=black,text=black]

  \node        (11)                    {$ (1,1)$};
  \node          (12) [right of=11]       {$(1,2)$};
  \node          (13) [ right of=12]      {$(1,3)$};
  \node          (21) [below of =11]          {$(2,1)$};
  \node         (22) [right of=21]       {$(2,2)$};
  \node         (23) [right of=22]       {$(2,3)$};
  \node         (31) [below of=21]       {$(3,1)$};
  \node         (32) [ right of=31]       {$(3,2)$};
  \node         (33) [ right of=32]       {$(3,3)$};

  \path
%first user 1's graph
 (11) edge  [dashed]            node {} (21)
 (11) edge [dashed,bend left]             node {} (31)
 (31) edge [dashed]              node {} (21)

 (12) edge  [dashed]            node {} (22)
 (12) edge [dashed,bend left]             node {} (32)
 (32) edge  [dashed]            node {} (22)

 (13) edge [dashed]              node {} (23)
 (13) edge [dashed,bend left]             node {} (33)
 (33) edge [dashed]             node {} (23)

%second user 2s graph

 (11) edge  [red]            node {} (12)
 (12) edge  [red]            node {} (13)
 (13) edge  [red,bend left]             node {} (11)

 (21) edge  [red]            node {} (22)
 (22) edge  [red]            node {} (23)
 (23) edge  [red,bend left]            node {} (21)

 (31) edge  [red]            node {} (32)
 (32) edge  [red]            node {} (33)
 (33) edge  [red, bend left]            node {} (31)
 (11) edge  []            node {} (12)
 (12) edge  []            node {} (13)
 (13) edge  [bend left]             node {} (11)

 (21) edge  []            node {} (22)
 (22) edge  []            node {} (23)
 (23) edge  [bend left]            node {} (21)

 (31) edge  []            node {} (32)
 (32) edge  []            node {} (33)
 (33) edge  [ bend left]            node {} (31)

;
\end{tikzpicture}
\caption{A game with two players, each of which has three
strategies. A node $(i, j)$ represents a strategy profile in which
player 1 and player 2 use strategies $i$ and $j$, respectively. The
Laplacian $\Delta_{0,1}$ ($\Delta_{0,2}$) is defined on the graph
whose edges are represented by dashed (solid) lines. The Laplacian
$\Delta_0$ is defined on the graph that includes all edges.}
\label{fig:graphLap}
\end{center}
\end{figure}

 Similarly, $\delta_1 \Lambda_m$ is the curl operator associated with the subgraph $(E, A^m)$. From the closedness of the curl ($\delta_1 \Lambda_m$) and gradient ($\Lambda_m \delta_0$) operators defined on this subgraph, we obtain $\delta_1 \Lambda_m^2 \delta_0=0$. Observing that $\Lambda_m^2 \delta_0= \Lambda_m \delta_0=D_m$, it follows that
\begin{equation} \label{eq:noCurl}
\delta_1 D_m=0.
\end{equation}
This result also implies that $\delta_1 D=0$, i.e., the pairwise comparisons of games belong to $\ker \delta_1$.
Thus, it follows from Theorem \ref{thm:hodge}  that the pairwise comparisons  do not have a locally inconsistent component.

Lastly, we introduce projection operators that will be useful in the subsequent analysis. Consider the operator, 
 \begin{equation*} \label{eq:projOp}
\Pi_m=D_m^\dagger D_m.
\end{equation*}
For any linear operator $L$, $L^\dagger L$ is a 
projection operator on  the orthogonal complement of the kernel of $L$ (see
\cite{golub1996mc}).  Since $D_m$ is a linear operator, $\Pi_m$ is a
projection operator to the orthogonal complement of the kernel of
$D_m$.
Using these operators, we define $\Pi: C_0^M \rightarrow C_0^M$ such that $\Pi= diag(\Pi_1, \dots , \Pi_M)$, i.e., for  $u=\{u_m\}_{m\in{\cal M}}\in C_0^M$,  we have $\Pi u= [\Pi_1 u^1; \dots \Pi_M u^M ] \in C_0^M$.
We extend the inner product in $C_0$ to $C_0^M$ (by defining the inner product as the sum of the inner products in all $C_0$ components),  and denote  by $D^\dagger$ the pseudoinverse of $D$ according to this inner product.
In Lemma \ref{lem:identity}, we will show that $\Pi$ is equivalent to the projection operator to the orthogonal complement of the kernel of $D$, i.e., $\Pi= D^\dagger D$.

For easy reference, Table~\ref{tab:SumNot} provides a summary of
notation. We next state some  basic facts about the operators we
introduced, which will be used in the subsequent analysis. The proofs
of these results can be found in Appendix~\ref{se:proofOfLapProj}.
\begin{table}[h]
\centering
\begin{tabular}{|c|| p{14cm} |}
\hline ${\cal G}$ & A game instance $\langle {\cal M},\{E^m\}_{m\in{\cal M}},\{u^m \}_{m\in{\cal M}} \rangle$.\\
\hline  ${\cal M}$ & Set of players, $\{1,\dots, M\}$.\\
\hline  $E^m$ & Set of actions for player $m$, $E^m=\{1,\dots, h_m\}$. \\
\hline $E$ & Joint action space $\prod_{m\in{\cal M}} E^m$. \\
\hline  $u^m$ & Utility function of player $m$. We have $u^m \in C_0$.\\
%\hline $u_{0}$ & The collection of utilities of all players, $u=\{u^m\}_{m\in{\cal M}} \in C_0^M $. USED IN MAIN TEXT? \\
\hline $W^m$ & Indicator function for $m$-comparable strategy profiles,
 %function indicating whether strategy profiles are $m$-comparable by player $m$ or not,
 $W^m: E\times E \rightarrow \{0,1\}$. \\
\hline $W$ & A function indicating whether strategy profiles are comparable, $W: E\times E \rightarrow \{0,1\}$. \\
\hline $C_0$ & Space of utilities, $C_0=\{u^m| u^m:E\rightarrow \mathbb{R} \}$. Note that $C_0 \cong \mathbb{R}^{|E|}$. \\
\hline $C_1$ & Space of pairwise comparison functions from $E\times E$ to $\mathbb{R}$. \\
\hline $\delta_0$ &  Gradient operator, $\delta_0:C_0\rightarrow C_1$, satisfying  $(\delta_0 \phi)({\bf p},{\bf q}) =W({\bf p}, {\bf q}) \left( \phi({\bf q})-\phi({\bf p}) \right)$.\\
\hline $D_m$ & $D_m:C_0\rightarrow C_1$, such that $(D_m \phi)({\bf p},{\bf q}) =W^m({\bf p},{\bf q}) \left( \phi({\bf q})-\phi({\bf p}) \right)$. \\
\hline $D$ & $D:C_0^M\rightarrow C_1$, such that $D(u^1;\ldots;u^M)=\sum_m D_m u^m$. \\
\hline $\delta_0^*$, $D_m^*$ &  $\delta_0^*, D_m^*:C_1\rightarrow C_0$ are the adjoints of the operators $\delta_0$ and $D_m$, respectively.\\
\hline $\Delta_0$ &  Laplacian for the game graph. $\Delta_0:C_0\rightarrow C_0$; satisfies $\Delta_0=\delta_0^*  \delta_0 = \sum_m \Delta_{0,m}.$ \\
\hline $\Delta_{0,m}$ &  Laplacian for the graph of $m$-comparable strategies, $\Delta_{0,m}:C_0\rightarrow C_0$; satisfies $\Delta_{0,m}=D_m^* D_m=D_m^* \delta_0$. \\
\hline $\Pi_m$ & Projection operator onto the orthogonal complement of kernel of $D_m$, $\Pi_{m}:C_0\rightarrow C_0$; satisfies $\Pi_{m}=D_m^\dagger D_m$.  \\
%No need to define norms as they are consistent with the definitions in the next section if matrix notation is adapted.
\hline
\end{tabular}
\caption{Notation summary}
\label{tab:SumNot}
\end{table}

\begin{lemma} 
\label{theo:projDiff}
The Laplacian of the graph induced by $m$-comparable strategies and
the projection operator $\Pi_m$ are related by $\Delta_{0,m}=h_m
\Pi_m$, where $h_m=|E^m|$ denotes the number of strategies of player
$m$.
\end{lemma}

\begin{lemma} \label{theo:projDiff2}
The kernels of operators $D_m$, $\Pi_m$ and $\Delta_{0,m}$ coincide,
 namely $\ker (D_m) = \ker (\Pi_m) = \ker
 (\Delta_{0,m})$. Furthermore, a basis for these kernels is given by a
 collection $\{\nu_{{\bf q}^{-m}} \}_{{\bf q}^{-m} \in E^{-m}} \in
 C_0$ such that
\begin{equation} \label{eq:defineNuQm}
\nu_{{\bf q}^{-m}}({\bf p})=\left\{
\begin{aligned}
& 1 \quad\quad \mbox{if ${\bf p}^{-m}={\bf q}^{-m}$} \\
& 0 \quad\quad \mbox{ otherwise}
\end{aligned}
\right.
\end{equation}

\end{lemma}

\begin{lemma} \label{lem:SpectraOfDelta0}
The Laplacian $\Delta_0$ of the game graph (the graph of comparable
strategy profiles) always has eigenvalue $0$ with multiplicity $1$,
corresponding to the constant eigenfunction (i.e., $f\in{C_0}$ such
that $f({\bf p})=1$ for all ${\bf p}\in E$).
\end{lemma}
\begin{lemma} \label{lem:identity}
The pseudoinverses of  operators $D_m$ and $D$ satisfy the following identities:
(i) $D_m^\dagger = \frac{1}{h_m} D_m^*$, (ii) $(\sum_i D_i)^\dagger  D_j =  (\sum_i D_i^* D_i)^\dagger D_j^* D_j$, (iii) $D^\dag = [D_1^\dag ; \ldots ; D_M^\dag]$, (iv) $\Pi=  D^\dagger D$  (v)  $D D^\dag \delta_0=\delta_0$.
\end{lemma}

\subsection{Decomposition of Games} \label{se:orthDecomp}
In this subsection we prove that the space of games ${\cal G}_{{\cal
M},E}$ is a direct sum of three subspaces -- potential, harmonic and
nonstrategic, each  with distinguishing properties.

We start our discussion by formalizing the notion of nonstrategic
information. Consider two games $\cal G$, $\hat{\cal G} \in {\cal
G}_{{\cal M}, E}$ with utilities $\{ u^m \}_{m\in {\cal M}}$ and $\{
\hat{u}^m \}_{m\in {\cal M}}$ respectively. Assume that the utility
functions $\{{u}^m\}_{m\in{\cal M}}$ satisfy $u^m({\bf p}^m, {\bf
p}^{-m})=\hat{u}^m({\bf p}^m, {\bf p}^{-m}) + \alpha({{\bf p}^{-m}})$
where $\alpha$ is an arbitrary function.  It can be readily seen that
these two games have exactly the same pairwise comparison functions,
hence they are strategically equivalent. To express the same idea in
words, whenever we add to the utility of one player an arbitrary
function of the actions of the others, this does not directly affect
the incentives of the player to choose among his/her possible
actions. Thus, pairwise comparisons of utilities (or equivalently, the game graph
representation) uniquely identify equivalent classes of games that
have identical properties in terms of, for instance, sets of
equilibria\footnote{We note, however, that payoff-specific information
such as efficiency notions are not necessarily preserved; see
Section~\ref{se:nonStrComp}.}. To fix a representative for each
strategically equivalent game, we introduce below a notion of games
where the nonstrategic information has been removed.

\begin{definition}[Normalized games] 
We say that a game with utility functions $\{u^m\}_{m\in{\cal M}}$ is
\emph{normalized} or does not contain {nonstrategic information} if
\be
\label{eq:nonstrategic_def} \sum_{{\bf p}^m} u^m({\bf p}^m, {\bf
p}^{-m})=0 \ee for all ${\bf p}^{-m} \in E^{-m}$ and all $m \in {\cal
M}$.
\end{definition}
Note that removing the nonstrategic information amounts to normalizing
the sum of the payoffs in the game. Normalization can be made with an
arbitrary constant. However, in order to simplify the subsequent
analysis we normalize the sum of the payoffs to zero.
Intuitively, this  suggests that given the strategies of  an agent's opponents, the average payoff its strategies yield, is  equal to zero.
The following lemma characterizes the set of normalized games in terms of the operators introduced in the previous section.
\begin{lemma} \label{lem:normalize}
Given  a game $\cal G$ with utilities $u=\{u^m\}_{m\in {\cal M}}$, the following are equivalent:
(i) $\cal G$ is normalized, (ii)  $ \Pi_m u^m = u^m $ for all $m $,  (iii)  $\Pi u =u$, (iv) $u \in (\ker D)^\perp$.
\end{lemma}
\begin{proof}
The equivalence of (iii) and (iv) is immediate since by  Lemma \ref{lem:identity}, $\Pi=D^\dagger D$ is a projection operator to the orthogonal complement of the kernel of $D$. The equivalence of (ii) and (iii) follows from the definition of $\Pi= diag(\Pi_1, \dots , \Pi_M)$.
 To complete the proof we prove (i) and (ii) are equivalent.

Observe that \eqref{eq:nonstrategic_def} holds if and only if $\langle u^m, \nu_{{\bf q}^{-m}}({\bf p})\rangle =0$ for all ${{\bf q}^{-m}} \in E^{-m}$, where $\nu_{{\bf q}^{-m}}$ is as defined in \eqref{eq:defineNuQm}.  Lemma \ref{theo:projDiff2} implies that  $\{\nu_{{\bf q}^{-m}}\}$ are basis vectors of $\ker D_m$. Thus, it follows that \eqref{eq:nonstrategic_def} holds if and only if $u^m$ is orthogonal to all of the basis vectors of $\ker D_m$, or equivalently when $u^m \in (\ker D_m) ^\perp$. Since $\Pi_m = D_m^\dagger D_m$ is a projection operator to $(\ker D_m)^\perp$, we have $u^m \in (\ker D_m) ^\perp$ if and only if $\Pi_m u^m=u^m$, and the claim follows.
\end{proof}

Using Lemma \ref{lem:normalize}, we next show below that for each game $\cal G$ there exists a unique
strategically equivalent game which is normalized (contains no
nonstrategic information).

\begin{lemma} \label{lem:uniqueNonstr}
 Let $\cal G$ be a game with utilities $\{u^m\}_{m\in {\cal M}}$. Then
 there exists a unique game~$\hat{\cal G}$ which (i) has the same
 pairwise comparison function as $\cal G$ and (ii) is
 normalized. Moreover the utilities $\hat{u}=\{\hat{u}^m\}_{m\in {\cal M}} $ of
 $\hat{\cal G}$ satisfy $\hat{u}^m= \Pi_m u^m$ for all $m$. 
\end{lemma}
\begin{proof}
To prove the claim, we show that given $u=\{u^m\}_{m\in {\cal M}}$, the game with the collection of utilities $D^\dagger D u = \Pi u$, is a normalized game with the same pairwise comparisons, and moreover there cannot be another normalized game which has the same pairwise comparisons.

Since $\Pi$ is a projection operator, it follows that $\Pi \Pi u = \Pi u$, and hence, Lemma \ref{lem:normalize} implies that $\Pi u$ is normalized. Additionally, using properties of the pseudoinverse we have $D \Pi u = D D^\dagger D u= D u$, thus $\Pi u$ and $u$ have the same pairwise comparison.

Let $v\in C_0^M$ denote the collection of payoff functions of a game which is normalized and has the same pairwise comparison as $u$.
It follows that $ D v = Du= D \Pi u$, and hence 
%$D(v- \Pi u)=0$, i.e., 
$v-\Pi u \in \ker D$. On the other hand, since both $v$ and $\Pi u$ are normalized, by Lemma \ref{lem:normalize}, we have $v, \Pi u \in   (\ker D)^\perp$, and thus $v - \Pi u \in   (\ker D)^\perp$. Therefore, it follows that $v - \Pi u =0$, hence $\Pi u$ is the collection of utility functions of the unique normalized game, which has the same pairwise comparison function as $\cal G$.
By Lemma \ref{lem:identity}, $\Pi u= \{\Pi_m u^m\}$, hence the claim follows.
%it follows that the unique normalized game  which has the same pairwise comparisons as the original game has utility functions $\{\Pi_m u^m\}$.
\end{proof}

We are now ready to define the subspaces of games that will appear in
our decomposition result.
\begin{definition} 
\label{def:subspaces}
The potential subspace $\cal P$, the harmonic subspace $\cal H$ and
the nonstrategic subspace $\cal N$ are defined as:
\begin{equation}
\label{eq:subspaceDef}
\begin{aligned}
{\cal P} & \triangleq \big \{ u \in C_0^M  ~| ~  u= \Pi u  ~ \mbox{and} ~ Du\in \im \delta_0 \big \} 
\\
{\cal H} & \triangleq\big \{ u\in C_0^M  ~|~  u= \Pi u   ~ \mbox{and} ~ Du\in \ker \delta_0^* \big \}\\
{\cal N} & \triangleq\big \{ u\in C_0^M  ~|~  u\in \ker D  \big \}. \\
\end{aligned}
\end{equation}
\end{definition}
Since the operators involved in the above definitions are linear, it follows that the sets ${\cal P}$, ${\cal H}$ and ${\cal N}$ are indeed subspaces.

Lemma \ref{lem:normalize} implies that the games in $\cal P$ and $\cal H$ are normalized (contain no
nonstrategic information).  The flows generated by the games
in these two subspaces are related to the flows induced by the
Helmholtz decomposition.  It follows from the definitions that the flows generated by a game
in $\cal P$ are in the image space of $\delta_0$ 
%$\sum_{m\in{\cal M}} D_m u^m= \sum_{m\in{\cal M}} D_m \phi = \delta_0 \phi$, 
and the flows generated by a game in $\cal H$ are in the kernel of
$\delta_0^*$. Thus, $\cal P$ corresponds to the set of normalized games, which have globally consistent pairwise comparisons. 
Due to \eqref{eq:noCurl}, the pairwise comparisons of games do not have locally inconsistent components, thus  Theorem \ref{thm:hodge} implies that $\cal H$ corresponds to the set of normalized games, which have globally inconsistent but locally consistent pairwise comparisons. 
% since $\delta_0^* \sum_{m\in{\cal M}} D_m u^m=0$. 
Hence,
from the perspective of the Helmholtz decomposition, the flows
generated by games in $\cal P$ and $\cal H$ are gradient 
and harmonic flows respectively.  
On the other hand the
flows generated by games in $\cal N$ are always zero,
since $Du=0$ in such games.

As discussed in the previous section the image spaces of $D_m$ are orthogonal.  Thus, since by definition $Du=\sum_{m\in{\cal M}} D_m u^m$, it follows that $u=\{u^m\}_{m\in{\cal M}} \in \ker D$ if and only if $u^m\in \ker D_m$ for all $m\in {\cal M}$.
Using these facts together with Lemma \ref{lem:normalize}, we obtain the following alternative description of the subspaces of games:
\begin{equation}
\label{eq:subspaceDef2}
\begin{aligned}
%{\cal P} &=\big \{ \{u^m\}_{m\in{\cal M}} \in {\cal G}_{{\cal M},E} ~|~ u^m=\Pi_m \phi \mbox{ for all $m\in {\cal M}$ and some $\phi \in C_0$}   \big \} \\
{\cal P} &=\big \{ \{u^m\}_{m\in{\cal M}}  ~|~ D_m u^m=D_m \phi \mbox{ and  $\Pi_m u^m = u^m$}\mbox{ for all $m\in {\cal M}$ and some $\phi \in C_0$}   \big \} \\
{\cal H} &=\big \{ \{u^m\}_{m\in{\cal M}}  ~|~ \delta_0^* \sum_{m\in{\cal M}} D_m u^m=0  \mbox{ and  $\Pi_m u^m = u^m$  for all $m\in {\cal M}$ }   \big \} \\
{\cal N} &=\big \{ \{u^m\}_{m\in{\cal M}}  ~|~ D_m u^m=0   \mbox{ for all $m\in {\cal M}$ }   \big \}. \\
\end{aligned}
\end{equation}

The main result of this section shows that not only these subspaces
have distinct properties in terms of the flows they generate, but in
fact they form a direct sum decomposition of the space of games.  We
exploit the Helmholtz decomposition (Theorem~\ref{thm:hodge}) for the
proof.
\begin{theorem} \label{theo:decompTheoSpace}
The space of games ${\cal G}_{{\cal M},E}$ is a direct sum of the
potential, harmonic and nonstrategic subspaces, i.e., ${\cal G}_{{\cal
M},E} = {\cal P} \oplus {\cal H} \oplus {\cal N}$. In particular,
given a game with utilities $u=\{u^m\}_{m\in{\cal M}}$, it can be
uniquely decomposed in three components:
\begin{itemize}
\item \noindent{\bf Potential Component:} $u_P  \triangleq D^\dag \delta_0 \delta_0^\dag D u$
\item \noindent{\bf Harmonic Component:}  $u_H  \triangleq D^\dag (I-\delta_0 \delta_0^\dag) D u$
\item \noindent{\bf Nonstrategic Component:} $u_N  \triangleq (I-D^\dag D) u$
\end{itemize}
where $u_P + u_H + u_N=u$, and $u_P \in \mathcal{P}$, $u_H \in
\mathcal{H}$, $u_N \in \mathcal{N}$. The potential function associated
with $u_P$ is $\phi \triangleq \delta_0^\dag D u$.
\end{theorem}
\begin{proof}
The decomposition of $\mathcal{G}_{\mathcal{M},E}$ described above
follows directly from pulling back the Helmholtz decomposition of
$C_1$ through the map $D$, and removing the kernel of $D$; see
Figure~\ref{fig:pullback}.

\begin{figure}[h]
\begin{center}
\begin{tikzpicture}
\matrix(a)[matrix of math nodes, row sep=3em, column sep=3em,
text height=2.5ex, text depth=0.35ex]
{{\cal G}_{{\cal M},E} \cong C_0^M 
\\ C_0 & C_1 & C_2\\};
%\path[->](a-2-1) edge node[left] {$R$} (a-1-1);
\path[->](a-1-1) edge node[above] {$D$} (a-2-2);
\path[->](a-2-1) edge node[below] {$\delta_0$} (a-2-2);
\path[->](a-2-2) edge node[below] {$\delta_1$} (a-2-3);
\end{tikzpicture}
\end{center}
\caption{The Helmholz decomposition of the space of flows ($C_1$) can
be pulled back through $D$ to a direct sum decomposition of the space
of games (${\cal G}_{{\cal M},E}$).}
\label{fig:pullback}
\end{figure}

The components of the decomposition clearly satisfy $u_{P} +
u_{H} + u_{N} = u$.  We verify the inclusion
properties, according to \eqref{eq:subspaceDef}. Both
$u_{P}$ and $u_{H}$ are orthogonal to $\mathcal{N} =
\ker D$, since they are in 
 the range of $D^\dag$.
\begin{itemize}
\item For the potential component, let $\phi \in C_0$ be such that  $\phi = \delta_0^\dag D
u$. Then, we have $D u_P \in \im(\delta_0)$, since
\[
D u_{P} = D D^\dag \delta_0 \delta_0^\dag D u = \delta_0 \delta_0^\dag
D u = \delta_0 \phi,
\]
where we used the definition of $u_{P}$,
%equation~\eqref{eq:ddd} 
the property (v) in Lemma \ref{lem:identity}
and the definition of $\phi$, respectively. This equality also implies that $\phi$ is the potential function associated with $u_P$.

\item For the harmonic component $u_{H}$, we have $D u_H \in \ker \delta_0^*$: 
\[
\delta_0^* D u_{H} =  
\delta_0^* D D^\dag (I-\delta_0 \delta_0^\dag) D u = \delta_0^* (I- \delta_0 \delta_0^\dag) D u = 0,
\]
as follows from the definition of $u_H$, 
the property (v) in Lemma \ref{lem:identity},
%the transpose of~\eqref{eq:ddd}, 
and properties of the pseudoinverse.

\item To check that $u_{N} \in \mathcal{N}$, we have
\[
D u_{N} = D (I-D^\dag D) u = (D- D D^\dag D) u =  0.
\]
\end{itemize}
In order to prove that the  direct sum decomposition property holds, we assume that  there exists $\hat{u}_{P} \in {\cal P}$, $\hat{u}_H \in {\cal H}$ and $\hat{u}_N\in {\cal N}$ such that $\hat{u}_{P} + \hat{u}_{H} + \hat{u}_{N} = 0$. Observe that $I-D^\dagger D$ is a projection operator to the kernel of  $D$. Thus, from the definition of   the subspaces $\cal P$,  $\cal H$ and $\cal N$, it follows that $(I-D^\dagger D) \hat{u}_N=\hat{u}_N$  and $(I-D^\dagger D) \hat{u}_P=(I-D^\dagger D) \hat{u}_H=0$. Similarly,   $\delta_0 \delta_0^\dagger$ is a projection operator to the image of $\delta_0$. Since by definition $D \hat{u}_P\in \im \delta_0$, and $D \hat{u}_H \in \ker \delta_0^*= (\im \delta_0)^\perp$, it follows that $\delta_0 \delta_0^\dagger D \hat{u}_P=D \hat{u}_P$ and $\delta_0 \delta_0^\dagger D \hat{u}_H=0$.

Using these identities, it follows that
\begin{align*}
  (D^\dag \delta_0 \delta_0^\dag D) (\hat{u}_{P} + \hat{u}_{H} + \hat{u}_{N}) = \hat{u}_{P} \\
  D^\dag (I - \delta_0 \delta_0^\dag) D (\hat{u}_{P} + \hat{u}_{H} + \hat{u}_{N}) = \hat{u}_{H} \\
 (I-D^\dag D) ( \hat{u}_{P} + \hat{u}_{H} + \hat{u}_{N}) =\hat{u}_{N}, 
\end{align*}
Since,  $\hat{u}_{P} + \hat{u}_{H} + \hat{u}_{N} = 0$ by our assumption, it follows that 
%Therefore, we obtain if  $\hat{u}_{P} + \hat{u}_{H} + \hat{u}_{N} = 0$ for some $\hat{u}_{P} \in {\cal P}$, $\hat{u}_H \in {\cal H}$ and $\hat{u}_N\in {\cal N}$, then 
$\hat{u}_{P} = \hat{u}_{H} = \hat{u}_{N} = 0$, and hence the direct sum decomposition property follows.
\end{proof}

The pseudoinverse  of a linear operator $L$, projects its argument to the image space  of $L$, and  then,  pulls the projection back to the the  domain of $L$.
%For a linear operator $L$, the pseudoinverse $L^\dagger f$, pulls the projection of   $f$ to the image space of $L$, back to the
Thus, intuitively,   the potential function $\phi=\delta_0^\dagger D u$, defined in the theorem, is such that the gradient flow associated with it
 ($\delta_0 \phi $) approximates the flow  in the original game  ($Du$),
in the best possible way.  The potential component of the game can be  identified by pulling back this gradient flow through $D$ to $C_0^M$. 
The harmonic component can similarly be obtained using the harmonic flow.

Since $\delta_0=\sum_m D_m$, it follows that $\phi=\delta_0^\dagger D u=(\sum_m D_m)^\dagger \sum_m D_m u^m$. Thus, Lemma \ref{lem:identity} (ii), and identities $\Delta_{0,m}=D_m^* D_m$ and $\Delta_0=\sum_m \Delta_{0,m}$ imply that 
$$\phi= \Delta_0^{\dagger}  \sum_{m\in{\cal M}}   \Delta_{0,m} u^m.$$
Additionally,  from  Lemma \ref{lem:identity} (iii) and (iv) it follows that $D^\dagger \delta_0= [D_1^\dagger D_1; \dots ; D_M^\dagger D_M]=[\Pi_1; \dots ; \Pi_M]$ and $D^\dagger D=\Pi=diag(\Pi_1, \dots , \Pi_M)$. Using these identities,  the utility functions of  components of a game
can  alternatively be expressed as follows:
\begin{itemize}
\item \noindent{\bf Potential Component:} 
%A game with collection of utilities $\{{u}_P^m\}_{m\in{\cal M}} \in {\cal P}$ such that 
${u}_P^m= \Pi_m \phi$, for all $m\in{\cal M}$,
\item \noindent{\bf Harmonic Component:} 
%A game with collection of utilities $\{{u}_H^m\}_{m\in{\cal M}} \in {\cal H}$ such that
$  u^m_H = \Pi_m u^m - \Pi_m \phi$, for all $m\in{\cal M}$,
\item \noindent{\bf  Nonstrategic Component:} 
%A game with collection of utilities $\{{u}_N^m\}_{m\in{\cal M}} \in {\cal N}$ such that 
${u}_N^m=(I-\Pi_m){u}^m$, for all $m\in{\cal M}$.
%${u}_H^m=u^m-{u}_N^m-{u}_P^m$.
\end{itemize}
%where  $\phi= \Delta_0^{\dagger}  \sum_{m\in{\cal M}}   \Delta_{0,m} u^m$ is the potential function of the potential component. 
It can be seen that the definitions of the subspaces do not rely on the inner product in $C_0^M$. Thus, the direct sum property implies  that the decomposition is canonical, i.e., it is independent of the inner product used in $C_0^M$. The above expressions provide  closed form solutions for the utility functions in the decomposition, without reference to this inner product. We show in Section \ref{se:projection} that our decomposition is indeed orthogonal with respect to a natural inner product in $C_0^M$.

Note that  $  \Delta_0 : C_0 \rightarrow C_0$, whereas $\delta_0: C_0 \rightarrow C_1$. Since $C_1$ and $C_0$ are associated with the edges and the nodes of the game graph respectively, in general $C_1$ is higher dimensional than $C_0$. Therefore, calculating $\Delta_0^\dagger$  is computationally more tractable than calculating $\delta_0^\dagger$. Hence, the alternative expressions for the components of a game and the potential function $\phi$, have computational benefits over using the results of Theorem \ref{theo:decompTheoSpace} directly.

We conclude this section by characterizing the dimensions of the
potential, harmonic and nonstrategic subspaces.
\begin{proposition} 
\label{theo:dimension}
The dimensions of the subspaces $\cal P$, $\cal H$ and $\cal N$ are:
\begin{enumerate}
\item \noindent{$\dim({\cal P})=\prod_{ m\in {\cal M}} h_m-1$},
\item \noindent{$\dim({\cal H})=(M-1)\prod_{m\in{\cal M}} h_m - \sum_{m\in{\cal M}} \prod_{k\neq m} h_k+1$}.
\item \noindent{$\dim({\cal N})=\sum_{m\in{\cal M}} \prod_{k\neq m} h_k$}.
\end{enumerate}
\end{proposition}

\begin{proof}
Lemma \ref{theo:projDiff2} provides a basis for kernel of $D_m$ and
$\dim (\ker (D_m)) = |E^{-m}|$, i.e., the cardinality of the basis is
equal to $|E^{-m}|$. By definition ${\cal N}= \ker D = \prod_{m\in{\cal M}}
\ker (D_m) $, hence
\begin{equation}
\dim({\cal N})= \sum_{m\in{\cal M}} \dim(\ker (D_m) )= \sum_{m\in{\cal
        M}} |E^{-m}|= \sum_{m\in{\cal M}} \prod_{k\neq m} h_k.
\end{equation}

Next consider the subspace $\cal P$ of normalized potential games. By
definition, the games in this set generate globally consistent
flows. Moreover, by Lemma~\ref{lem:uniqueNonstr} it follows that there is a unique
game in $\cal P$, which generates a given gradient flow. Thirdly, note that any
globally consistent flow can be obtained as $\delta_0 \phi $ for some
$\phi \in C_0$, and the game $\{\Pi_m \phi \}_{m\in{\cal M}}\in{\cal
P}$ generates the same flows as $\delta_0 \phi$. These three
facts imply that there is a linear bijective mapping between the games
in $\cal P$ and the globally consistent flows, and hence the
dimension of $\cal P$ is equal to the dimension of the globally
consistent flows.

On the other hand, the dimension of the globally consistent flows is
equivalent to $\dim ( \im (\delta_0) )$.  Since $\Delta_0 = \delta_0^*
\delta_0$ it follows that $\ker (\delta_0) \subset \ker (\Delta_0)$.
By Lemma \ref{lem:SpectraOfDelta0} it follows that $\ker (\Delta_0) =
\{f\in C_0~|   f({\bf p})=c\in \mathbb{R}, \mbox{for  all ${\bf p} \in E$ } \} $. It
follows from the definition of $\delta_0$ that $\delta_0 f=0$ for all
$f\in \ker (\Delta_0)$.  These facts imply that $\ker (\delta_0)= \ker
(\Delta_0)$ and hence $\dim ( \ker (\delta_0))=1$. Since $\delta_0$ is
a linear operator it follows that $\dim (\im (\delta_0) ) = \dim
(C_0)- \dim (\ker (\delta_0))= |E|-1= \prod_{m\in{\cal M}} h_m-1$.

Finally observe that $\dim( {\cal G}_{{\cal M}, E})= \dim(C_0^M)= M
|E|= M \prod_{m\in{\cal M}}h_m$. Theorem \ref{theo:decompTheoSpace}
implies that $\dim( {\cal G}_{{\cal M}, E})=\dim( {\cal P}) +
\dim({\cal H})+ \dim({\cal N})$.  Therefore, it follows that
$\dim({\cal H})=(M-1)\prod_{m\in{\cal M}} h_m - \sum_{m\in{\cal M}}
\prod_{k\neq m} h_k+1$.
\end{proof}

\subsection{Bimatrix Games} 
\label{se:bimatrix}

We conclude this section by providing an explicit decomposition result
for bimatrix games, i.e., finite games with two players.  Consider a
bimatrix game, where the payoff matrix of the row player is given by
$A$, and that of the column player is given by $B$; that is, when the
row player plays $i$ and the column player plays $j$, the row player's
payoff is equal to $A_{ij}$ and the column player's payoff is equal
to $B_{ij}$.

Assume that both the row player and the column player have the same
number $h$ of strategies. It immediately follows from Proposition
\ref{theo:dimension} that $\dim {\cal P} = h^2-1$, $\dim {\cal H} =
(h-1)^2$ and $\dim {\cal N} = 2h$.  For simplicity, we further assume
that the payoffs are normalized\footnote{Lemma~\ref{theo:projDiff2} and Lemma \ref{lem:uniqueNonstr}  imply that if the payoffs are not normalized, the normalized payoffs can be obtained as $(A-\frac{1}{h} {\bf 1}{\bf 1}^T A, ~B-\frac{1}{h}B  {\bf 1}{\bf 1}^T)$
.}.
Thus, the definition of normalized games implies that
 %In light of Lemma~\ref{theo:projDiff2} and Lemma~\ref{lem:normalize},
%%Theorem~\ref{theo:decompTheoSpace}, 
%this assumption translates to 
${\bf 1}^T A=B {\bf 1}=0$, where ${\bf 1}$
denotes the vector of ones.  Denote by ${A}_P$ (${B}_P$) and ${A}_H$
(${B}_H$) respectively, the payoff matrices of the row player (column player) in the potential and harmonic components of the game.  Using our decomposition result (Theorem
\ref{theo:decompTheoSpace}), it follows that
\begin{equation} \label{eq:harmGameDec2P}
(A_P,B_P) = (S+ \Gamma, S- \Gamma),  \qquad \qquad         
(A_H,B_H) = (D- \Gamma,-D+ \Gamma),
%       A_H &= \frac{A-B}{2}- \frac{A {\bf 1} {\bf 1}^T- {\bf 1} {\bf 1}^T B}{2h_0}  \\
%       A_P &= \frac{A+B}{2}+ \frac{A {\bf 1} {\bf 1}^T- {\bf 1} {\bf 1}^T B}{2h_0}  \\
%       B_H &= \frac{B-A}{2}- \frac{ {\bf 1} {\bf 1}^T B- A {\bf 1} {\bf 1}^T }{2h_0}  \\
%       B_P &= \frac{A+B}{2}+ \frac{ {\bf 1} {\bf 1}^T B- A {\bf 1} {\bf 1}^T }{2h_0},  \\
\end{equation}
where $S=\frac{1}{2}(A+B)$, $D=\frac{1}{2}(A-B)$, $\Gamma =\frac{1}{2h}(A {\bf 1}
  {\bf 1}^T- {\bf 1} {\bf 1}^T B)$. Interestingly, the potential
component of the game relates to the average of the payoffs in the
original game and the harmonic component relates to the difference in
payoffs of players. The $\Gamma$ term ensures that the potential and
harmonic components do not contain nonstrategic information.  We use
the above characterization in the next example for obtaining explicit
payoff matrices for each of the game components.

\begin{example}[Generalized Rock-Paper-Scissors] \label{ex:rps}
The payoff matrix of the generalized Rock-Paper-Scissors (RPS) game
is given in Table \ref{tab:rpsGen}.  Tables
\ref{tab:rpsNon}, \ref{tab:rpsProj} and \ref{tab:rpsRes} include the
nonstrategic, potential and the harmonic components of the game.  The
special case where $x=y=z=\frac{1}{3}$ corresponds to the celebrated
RPS game. Note that in this case, the potential component of the game
is equal to zero.

\begin{table}[ht]
\centering
{\small
\subfloat[Generalized RPS Game]{ \label{tab:rpsGen}
\begin{tabular}{ | c | c | c | c |}
\hline
  & R & P & S \\ \hline
R & $0, 0$ & $-3x, 3x$ & $3y, -3y$  \\ \hline
P &  $3x, -3x$ & 0, 0 & $-3z, 3z$ \\ \hline
S &  $-3y, 3y$ & $3z, -3z $& $0, 0$ \\ \hline
\end{tabular}
}
\quad\quad
\subfloat[Nonstrategic Component]{ \label{tab:rpsNon}
\begin{tabular}{ | c | c | c | c |}
\hline
  & R & P & S \\ \hline
R & $(x-y), (x-y)$  & $(z-x), (x-y)$ & $(y-z), (x-y)$  \\ \hline
P & $(x-y), (z-x)$  & $(z-x), (z-x)$ & $(y-z), (z-x)$ \\ \hline
S & $(x-y), (y-z)$  & $(z-x), (y-z)$ & $(y-z), (y-z)$ \\ \hline
\end{tabular}
}

\subfloat[Potential Component]{ \label{tab:rpsProj}
\begin{tabular}{ | c | c | c | c |}
\hline
  & R & P & S \\ \hline
R & $(y-x), (y-x)$ & $(y-x), (x-z)$ & $(y-x), (z-y)$  \\ \hline
P & $(x-z), (y-x)$ & $(x-z), (x-z)$ & $(x-z), (z-y)$ \\ \hline
S & $(z-y), (y-x)$ & $(z-y), (x-z)$ & $(z-y), (z-y)$ \\ \hline
\end{tabular}
}

\subfloat[Harmonic Component]{ \label{tab:rpsRes}
\begin{tabular}{ | c | c | c | c |}
\hline
   & R & P & S \\ \hline
R & 0, 0 & $-(x+y+z), (x+y+z)$ & $(x+y+z), -(x+y+z)$  \\ \hline
P &  $(x+y+z), -(x+y+z)$ & 0, 0 & $-(x+y+z), (x+y+z)$ \\ \hline
S &  $-(x+y+z), (x+y+z)$ & $(x+y+z), -(x+y+z)$ & 0, 0 \\ \hline
\end{tabular}
}
}
\caption{Generalized RPS game and its components.}
\label{tab:rps}
\end{table}
\end{example}

\section{Properties of the Components} 
\label{se:decProp}

In this section we study the classes of games that are naturally
motivated by our decomposition.  In particular, we focus on two
classes of games: (i) Games with no harmonic component, (ii) Games
with no potential component.  We show that the first class is
equivalent to the well-known class of \emph{potential games}. We refer
to the games in the second class as \emph{harmonic
games}. Pictorially, we have
\[
\mathcal{P} \quad \: \oplus \: \quad \overbrace{ 
\mathcal{N} \makebox[0pt][r]{$\underbrace{\phantom{\mathcal{P} \quad \: \oplus \: \quad \mathcal{N}}}_{\text{Potential games}}$} 
\quad \: \oplus \: \quad \mathcal{H}}^{\text{Harmonic games}}.
\]
In Sections~\ref{se:dynamicsNpotential} and~\ref{se:HarmGames}, we
explain these facts, and develop and discuss several properties of
these classes of games, with particular emphasis on their
equilibria. Since potential games have been extensively studied in the   literature, our 
main focus is on harmonic games.
In Section~\ref{se:nonStrComp}, we elaborate on the effect
of the nonstrategic component.  Potential and harmonic games are
related to other well-known classes of games, such as the zero-sum
games and identical interest games.  In Section \ref{se:ZS}, we
discuss this relation, in the context of bimatrix games.  As a
preview, in Table \ref{tab:HarmNew}, we summarize some of the
properties of potential and harmonic games that we obtain in the
subsequent sections.

\begin{table}
\centering
\begin{tabular}{ | m{2.5cm} | m{3.5cm} | m{9cm} | }
\hline
 				 & Potential Games & Harmonic Games \\ \hline
{Subspaces} 	& ${\cal P} \oplus {\cal N}$ & ${\cal H} \oplus{\cal N}$ \\ \hline
{Flows}	& Globally consistent  & Locally consistent but globally inconsistent  \\ \hline
Pure NE 		&  Always Exists  &  Generically does not exist   \\  \hline
Mixed NE 		&  Always Exists  & -Uniformly mixed strategy  is always  a mixed NE\\  
	&							& -Players do not strictly prefer their equilibrium strategies. \\ \hline
Special Cases	&			 { \center{\qquad \qquad \quad --} }				& -(two players) Set of mixed Nash equilibria coincides with the set of correlated equilibria \\
	&							& -(two players \& equal number of strategies) Uniformly mixed strategy  is the unique mixed NE \\ \hline
\end{tabular}
 \caption{Properties of potential and harmonic games.} 
 \label{tab:HarmNew}
\end{table}

\subsection{Potential Games} 
\label{se:dynamicsNpotential}
Since the seminal paper of Monderer and Shapley \cite{monderer1996pg},
potential games have been an active research topic.  The desirable
equilibrium properties and structure of these games played a key role
in this. In this section we explain the relation of the potential
games to the decomposition in Section~\ref{se:canonicalRep} and
briefly discuss their properties.

Recall from Definition~\ref{def:ExactPot} that a game is a
potential game if and only if there exists some $\phi \in C_0$ such
that $D u= \delta_0 \phi $. This condition implies that a game is
potential if and only if the associated flow is globally
consistent.
Thus, it can be seen from the definition of the subspaces and Theorem~\ref{theo:decompTheoSpace} that the set of potential games is actually equivalent    to ${\cal P}\oplus{\cal N}$. 
For future reference, we summarize this result in the following theorem.
 
 \begin{theorem} 
\label{theo:setOfPot}
The set of potential games is equal to the subspace ${\cal P} \oplus {\cal N}$.
\end{theorem}

 Theorem~\ref{theo:setOfPot} implies that potential games are games
which only have potential and nonstrategic components. Since this set
is a subspace, one can consider \emph{projections} onto the set of
potential games, i.e., it is possible to find the closest potential
game to a given game. We pursue the idea of projection in
Section~\ref{se:projection}.  Using the previous theorem we next find
the dimension of the subspace of potential games.
\begin{corollary} \label{cor:potDimFin}
The subspace of potential games, ${\cal P} \oplus {\cal N}$, has
dimension $\prod_{ m\in {\cal M}} h_m + \sum_{m\in{\cal M}}
\prod_{k\neq m} h_k -1$.
\end{corollary}
\begin{proof}
The result immediately follows from Theorem~\ref{theo:setOfPot} and
Proposition~\ref{theo:dimension}.
\end{proof}

We next provide a brief discussion of the equilibrium properties of potential games.

\begin{theorem}[\cite{monderer1996pg}]
    Let ${\cal G}=\langle {\cal M},\{E^m\},\{u^m \} \rangle$ be a potential game and $\phi$ be a corresponding potential function.
    \begin{enumerate}
        \item The equilibrium set of $\cal G$ coincides with the
        equilibrium set of $\mathcal{G}_\phi \triangleq \langle {\cal
        M},\{E^m\},\{\phi \} \rangle$.
        \item $\cal G$ has a pure Nash equilibrium.
%       \item $\cal G$ is isomorphic to a congestion game.
    \end{enumerate}
\end{theorem}
The first result follows from the fact that the games $\mathcal{G}$
and $\mathcal{G}_\phi$ are strategically equivalent. Alternatively,
the preferences in $\mathcal{G}$ are aligned with the global objective
denoted by the potential function $\phi$. The second result is implied
by the first one since in finite games the potential function $\phi$
necessarily has a maximum, and the maximum is a Nash equilibrium of
${\cal G}_\phi$. These results indicate that potential games can be
analyzed by an equivalent game where each player has the same utility
function $\phi$. The second game is easy to analyze since when agents
have the same objective,  the game is similar to an
optimization problem with objective function $\phi$.

Another desirable property of potential games relates to their
dynamical properties.  An important question in game theory is how a
game reaches an equilibrium. This question is usually answered by
theoretical models of player dynamics.  For general games, ``natural''
player dynamics do not necessarily converge to an equilibrium and
various counterexamples are provided in the literature
\cite{fudenberg1998tlg, Jordan:1993p1604}. However, it is known that
some of the well-known dynamics such as fictitious play, best-response
dynamics (and their variants) converges in potential games
\cite{monderer1996pg, marden2005jsf, young2004sla,Candogan2009Pricing,
marden2008rll,fudenberg1998tlg, hofbauer2002gcs, Shamma:2004p1450}.
The results for convergence in potential games can be extended to
``near-potential'' games using our decomposition framework and these
results are discussed in \cite{Candogan2009}.

\subsection{Harmonic Games} 
\label{se:HarmGames}
 
In this section, we focus on games in which the potential component is
zero, hence the strategic interactions are governed only by the
harmonic component. We refer to such games as \emph{harmonic games},
i.e., a game $\cal G$ is a harmonic game if ${\cal G}\in {\cal H}
\oplus {\cal N}$.

This section studies the properties of equilibria of harmonic games.
We first characterize the Nash equilibria of such games, and show that
generically they do not have a pure Nash equilibrium. We further
consider mixed Nash and correlated equilibria, and show how the
properties of harmonic games restrict the 
possible set of equilibria.

\subsubsection{Pure Equilibria}
In this section, we focus on pure Nash equilibria in harmonic
games. Additionally, we characterize the dimension of the space of harmonic
games, ${\cal H} \oplus {\cal N}$.

We first show that at a pure Nash equilibrium of a harmonic game, all
players are indifferent between \emph{all} of their strategies.

\begin{lemma} \label{theo:pureNashHarm}
Let $\mathcal{G} = \langle {\cal  M},\{E^m\},\{u^m\} \rangle$ be a harmonic game and
 $\bf p$ be a pure Nash equilibrium. Then,
\begin{equation} \label{eq:indiffEq}
    u^m({\bf p}^m, {\bf p}^{-m})=   u^m({\bf q}^m, {\bf p}^{-m}) \quad {\mbox{for all $m\in {\cal M}$ and ${\bf q}^m \in E^m.$}}
\end{equation}
\end{lemma}
\begin{proof}
By definition, in harmonic games the utility functions $u=\{u^m\}$ satisfy the condition $\delta_0^* D u =0$.
By \eqref{eq:delta0Explicit}
and \req{eq:defDm}, $\delta_0^* D u$ evaluated at $\bf p$ can be expressed as,
\begin{equation} \label{eq:strictNE}
\sum_{m\in{\cal M}}\sum_{{\bf q} | ({\bf p}, {\bf q})\in A^m}  \left(u^m({\bf p})- u^m ({ \bf q})\right) = 0.
\end{equation}
Since $\bf p $ is a Nash equilibrium it follows that $u^m({\bf p})-
u^m({\bf q}) \geq 0$ for all $ ({\bf p}, {\bf q} ) \in A^m$ and
$m\in{\cal M}$. Combining this with \eqref{eq:strictNE} it follows
that $u^m({\bf p})- u^m({\bf q}) = 0$ for all $({\bf p},{\bf q}) \in
A^m$ and $m\in{\cal M}$.  Observing that $({\bf p}, {\bf q})\in A^m$
if and only if ${\bf q}=({\bf q}^m, {\bf p}^{-m})$ , the result
follows.
\end{proof}

Using this result we next prove that harmonic games generically do not
have pure Nash equilibria. By ``generically'', we mean that it is true for
almost all harmonic games, except possibly for a set of measure zero
(for instance, the trivial game where all utilities are zero is
harmonic, and clearly has pure Nash equilibria).

\begin{proposition}
Harmonic games generically do not have pure Nash equilibria.
\end{proposition}
\begin{proof}
Define ${\cal G}_{\bf p} \subset {\cal H} \oplus {\cal N}$ as the set
of harmonic games for which $\bf p$ is a pure Nash
equilibrium. Observe that $\cup_{{\bf p}\in E} {\cal G}_{\bf p}$ is
the set of all harmonic games which have a pure Nash equilibria. We
show that ${\cal G}_{\bf p}$ is a lower dimensional subspace of the
space of harmonic games for each ${\bf p}\in E$. Since the set of
harmonic games with pure Nash equilibrium is a finite union of lower
dimensional subspaces it follows that generically harmonic games do
not have pure Nash equilibria.

By Lemma \ref{theo:pureNashHarm} it follows that $${\cal G}_{\bf p}
=({\cal H} \oplus {\cal N}) \cap \{ \{u^m\}_{m\in{\cal M}} | u^m({\bf
p})= u^m({\bf q}), \mbox{ for all $\bf q$ such that $({\bf p}, {\bf
q})\in A^m$ and $m\in{\cal M}$ } \}. $$ Hence ${\cal G}_{\bf p}$ is a
subspace contained in ${\cal H} \oplus {\cal N}$.  It immediately
follows that ${\cal G}_{\bf p}$ is a lower dimensional subspace if we
can show that there exists harmonic games which are not in ${\cal
G}_{\bf p }$, i.e., in which $\bf p$ is not a pure Nash equilibrium.

Assume that $\bf p$ is a pure Nash equilibrium in all harmonic
games. Since $\bf p$ is arbitrary this holds only if all strategy
profiles are pure Nash equilibria in harmonic games. If all strategy
profiles are Nash equilibria, by Lemma \ref{theo:pureNashHarm} it
follows that the pairwise ranking function is equal to zero in
harmonic games, hence ${\cal H} \oplus {\cal N} \subset {\cal N}$. We
reach a contradiction since dimension of $\cal H$ is larger than zero.

Therefore, ${\cal G}_{\bf p}$ is a strict subspace of the space of
harmonic games, and thus harmonic games generically do not have pure
Nash equilibria.
\end{proof}

We conclude this section by
 a dimension result that is analogous to the result obtained for potential games.
\begin{theorem}
The set of harmonic games,  ${\cal H} \oplus {\cal N}$, has dimension $(M-1)\prod_{m\in{\cal M}} h_m +1$.
\end{theorem}
\begin{proof}
The result immediately follows from Theorem~\ref{theo:decompTheoSpace}
and Proposition~\ref{theo:dimension}.
\end{proof}

\subsubsection{Mixed Nash and Correlated Equilibria in Harmonic Games}

In the previous section we showed that harmonic games generically do
not have pure Nash equilibria.  In this section, we study their mixed
Nash and correlated equilibria. 
In particular, we show that in harmonic games, the mixed strategy profile, in which players uniformly randomize over their strategies is always a mixed Nash equilibrium. Additionally, in the case of two-player harmonic games mixed Nash and correlated equilibria coincide, and if players have equal number of strategies the uniformly mixed strategy profile is the unique correlated equilibrium of the game.
Before we discuss the details of these results, we next provide some preliminaries
and notation.

We denote the set of probability distributions on $E$ by $\Delta
E$. Given $x\in \Delta E$, $x({\bf p})$ denotes the probability
assigned to ${\bf p}\in  E$. Observe that for all $x\in \Delta
{E}$, $\sum_{{\bf p}\in E} x({\bf p})=1$, and $x({\bf p}) \geq 0$.
Similarly for each player $m\in{\cal M}$, $\Delta E^m$ denotes the set
of probability distributions on $E^m$ and for $x^m \in \Delta E^m$,
$x^m({{\bf p}^m})$ is the probability assigned to strategy ${\bf p}^m
\in E^m$. As before all $x^m \in \Delta E^m$ satisfies $\sum_{{\bf
p}^m\in{E^m}}x^m({\bf p}^m)=1$ and $x^m({{\bf p}^m}) \geq 0$.  We
refer to the distribution $x^m\in \Delta E^m$ as a mixed strategy of
player $m\in {\cal M}$ and the collection $x=\{x^m\}_m$ as a mixed
strategy profile. Note that $\{x^m\}_m \in \prod_m \Delta E^m \subset
\Delta E$. Mixed strategies of all players but the $m$th one is
denoted by $x^{-m}$.

With some  abuse of the notation, we define the mixed extensions of the utility functions
% such that for any  $x\in \prod_m \Delta E^m$,
$u^m:  \prod_m \Delta E^m \rightarrow \mathbb{R}$ such that for any  $x\in \prod_m \Delta E^m$,
\begin{equation}
u^m(x)=\sum_{{\bf p}\in E} u^m({\bf p}) \prod_{k\in{\cal M}} x^k({{\bf p}^k}).
\end{equation}
Similarly, if player $m$ uses pure strategy ${\bf q}^m$ and the other
players use the mixed strategies $x^{-m}$ we denote the payoff of
player $m$ by,
\begin{equation}
u^m({\bf q}^m,x^{-m})=\sum_{{\bf p}^{-m}\in E^{-m}} u^m({\bf q}^m, {\bf p}^{-m}) \prod_{k\in{\cal M}, k\neq m} x^k({{\bf p}^k}).
\end{equation}
Using this notation we can define the solution concepts.
\begin{definition}[Mixed Nash / Correlated Equilibrium]  \label{def:mixed}
    Consider the game  $\langle {\cal M}, \{E^m\}, \{u^m\} \rangle $.
    \begin{enumerate}

        \item A mixed strategy profile $x=\{x^m\}_m \in \prod_m \Delta
        E^m$ is a \emph{mixed Nash equilibrium} if for all $m\in{ \cal M}$
        and ${\bf p}^m \in E^m$, $u^m({x}^m, {x}^{-m})\geq u^m({\bf
        p}^m, x^{-m})$.

        \item A probability distribution $x\in \Delta E$ is a
\emph{correlated equilibrium} if for all $m\in{ \cal M}$ and ${\bf p}^m,
{\bf q}^m \in E^m$, $ \sum_{{\bf p}^{-m}} \left( u^m({\bf p}^m,{\bf
p}^{-m})-u^m({\bf q}^m,{\bf p}^{-m})\right) x({\bf p}^m, {\bf p}^{-m})
\geq 0. $

    \end{enumerate}
\end{definition}

From these definitions it can be seen that every mixed Nash
equilibrium is a correlated equilibrium where the corresponding
distribution $x\in \prod_m \Delta E^m \subset \Delta E$ is a product
distribution, i.e., it satisfies $x(\textbf{p}) = \prod_m
x^m(\textbf{p}^m)$

These definitions also imply that  similar to  Nash equilibrium, the conditions for mixed Nash and correlated equilibria can be expressed only in terms of  pairwise comparisons. Therefore,  these equilibrium sets are   independent of the nonstrategic components of games.

We next obtain an alternative characterization of correlated equilibria in normalized harmonic games. This characterization will be  more convenient when studying the equilibrium properties of harmonic games, as it is expressed in terms of equalities, instead of  inequalities.

\begin{proposition} \label{prop:corEqCharNew}
Consider a normalized harmonic game, ${\cal G}=\langle {\cal M}, \{u^m\}, {E^m} \rangle $ and a probability distribution $x\in  \Delta E$. The following are equivalent:
\begin{itemize}
\item[(i)] $x$ is a correlated equilibrium.
\item[(ii)]  For all ${\bf p}^m$, ${\bf q}^m$ and $m \in {\cal M}$,
      \begin{equation} \label{eq:harmCorr0}
            \sum_{{\bf p}^{-m}} \left( u^m({\bf p}^m,{\bf p}^{-m})-u^m({\bf q}^m,{\bf p}^{-m}) \right) x({\bf p}^m, {\bf p}^{-m})  = 0.
        \end{equation}
       
 \item[(iii)]  For all ${\bf p}^m$, ${\bf q}^m$ and $m \in {\cal M}$,
 \begin{equation} \label{eq:harmCorr1}
            \sum_{{\bf p}^{-m}} u^m({\bf q}^m,{\bf p}^{-m})  x({\bf p}^m, {\bf p}^{-m})  = 0.
        \end{equation}
        \end{itemize}
\end{proposition}
\begin{proof} 
We prove the claim, by  first showing  (i) and (ii) are equivalent and then establishing the equivalence (ii) and (iii). 

    By the definition of correlated equilibrium,  \eqref{eq:harmCorr0} implies that $x$ is a correlated equilibrium.
    To see that  any correlated equilibrium of $\cal G$ satisfies \eqref{eq:harmCorr0}, assume 
 $x\in \Delta E$ is a correlated equilibrium.
Since the game is a harmonic game, by definition, the utility functions $u=\{u^m\}$ satisfy the condition $\delta_0^* D u =0$.
Using \eqref{eq:delta0Explicit} and \req{eq:defDm}, this condition can equivalently be expressed as
%For harmonic games the divergence is zero at each strategy profile, i.e. for all ${\bf p}\in E$,
\begin{equation}
\sum_{m\in{\cal M}} \sum_{{\bf q}^m \in E^m} u^m({\bf q}^m, {\bf p}^{-m})- u^m({\bf p}^m,  {\bf p}^{-m})=0  \quad \mbox{for all ${\bf p}\in E$}.
\end{equation}
Thus, it follows that 
%mixed strategy $x$, let $x_{\bf p}= \prod_{k\in{\cal M}} x_{{\bf p}^k}$ and $x_{{\bf p}^{-m}}= \prod_{k\in{\cal M}, k\neq m} x_{{\bf p}^k}$. The divergence condition implies the desired result:
\begin{equation} \label{eq:harmCorr2}
\begin{aligned}
0&=\sum_{{\bf p}\in E} x({\bf p}) \sum_{m\in{\cal M}} \sum_{{\bf q}^m \in E^m} u^m({\bf q}^m, {\bf p}^{-m})- u^m({\bf p}^m,  {\bf p}^{-m}) \\
&=      \sum_{m\in {\cal M}} \sum_{{\bf q}^m \in E^m} \sum_{{\bf p}^m\in E^m} \sum_{{\bf p}^{-m}\in E^{-m}} x({\bf p}^m, {\bf p}^{-m}) \left( u^m({\bf q}^m, {\bf p}^{-m})- u^m({\bf p}^m,  {\bf p}^{-m})\right).
%\sum_{m\in{\cal M}} \sum_{{\bf q}^m \in E^m} \sum_{{\bf p}\in E} x({\bf p}) (u^m({\bf q}^m, {\bf p}^{-m})- u^m({\bf p}^m,  {\bf p}^{-m})).
\end{aligned}
\end{equation} 
Since $x$ is a correlated equilibrium, $\sum_{{\bf p}^{-m}\in E^{-m}} x({\bf p}^m, {\bf p}^{-m}) \left( u^m({\bf q}^m, {\bf p}^{-m})- u^m({\bf p}^m,  {\bf p}^{-m}) \right) \leq 0$ for all ${\bf p}^m$, ${\bf q}^m$ and $m \in {\cal M}$. Hence, \eqref{eq:harmCorr2} implies that
$$
 \sum_{{\bf p}^{-m}\in E^{-m}} x({\bf p}^m, {\bf p}^{-m}) \left( u^m({\bf q}^m, {\bf p}^{-m})- u^m({\bf p}^m,  {\bf p}^{-m})\right)=0
$$
for all ${\bf p}^m$, ${\bf q}^m$ and $m \in {\cal M}$. Thus, we conclude (i) and (ii) are equivalent.

To see the equivalence of (ii) and (iii), observe that  (iii)  immediately implies (ii). Assume (ii) holds, then writing  \eqref{eq:harmCorr0} 
for two strategies  $\mathbf{r}^m, \mathbf{q}^m \in E^m$, and subtracting these equations from each other, it follows that
      \begin{equation} \label{eq:harmCorr3}
            \sum_{{\bf p}^{-m}} \left( u^m({\bf r}^m,{\bf p}^{-m})-u^m({\bf q}^m,{\bf p}^{-m}) \right) x({\bf p}^m, {\bf p}^{-m})  = 0.
        \end{equation}
        Since $\mathbf{r}^m$ and  $\mathbf{q}^m$ are arbitrary it follows that for all $\mathbf{q}^m\in E^m$
  \begin{equation} \label{eq:harmCorr4}
            \sum_{{\bf p}^{-m}}  u^m({\bf q}^m,{\bf p}^{-m})  x({\bf p}^m, {\bf p}^{-m})  = c_{\mathbf{p}^m},
        \end{equation}
        for some $ c_{\mathbf{p}^m} \in \mathbb{R}$. Since the game is normalized, we have $\sum_{{\bf q}^{m}}  u^m({\bf q}^m,{\bf p}^{-m})  =0$.  Thus summing \eqref{eq:harmCorr4} over $\mathbf{q}^m$ it follows that $ c_{\mathbf{p}^m}=0$, and hence 
(ii)  implies (iii). 

Therefore we conclude that (i), (ii) and (iii) are equivalent for normalized harmonic games.
\end{proof}

Note that in the above proof, we used the assumption that the game is normalized, only when establishing the equivalence of (ii) and (iii). Therefore, it can be seen that (i) and (ii) are equivalent for all harmonic games.

The above proposition implies that the correlated equilibria of
harmonic games correspond to the intersection of the probability
simplex with a subspace defined by the utilities in the game. Using
this result, we obtain the following characterization of mixed Nash
equilibria of harmonic games.

\begin{corollary}\label{cor:harmGameEqCond}
Let   ${\cal G}=\langle {\cal M}, \{u^m\}, \{E^m\} \rangle $ be a harmonic game. The mixed strategy profile  $x\in \prod_m \Delta E^m$ is a mixed Nash equilibrium if and only if,
\begin{equation} \label{eq:mixedEqCond}
u^m(x^m, x^{-m})=u^m({\bf p}^m, x^{-m}) \quad \quad \mbox{for all  ${\bf p}^m \in E^m$ and $m\in \cal{M}$.}
\end{equation}
\end{corollary}
\begin{proof}
Assume that \eqref{eq:mixedEqCond} holds, then clearly all players are
indifferent between all their mixed strategies, hence it follows that
$x$ is a mixed Nash equilibrium of the game.

Let $x$ be a mixed Nash equilibrium. Since each mixed Nash equilibrium
is also a correlated equilibrium, from equivalence of (i) and (ii) of
Proposition~\ref{prop:corEqCharNew}  for all harmonic games, it follows that for all ${\bf p}^m$,
${\bf q}^m$ and $m \in {\cal M}$,
    \begin{equation}
    \begin{aligned}
        0 &=\sum_{{\bf p}^{-m}} \left( u^m({\bf p}^m,{\bf p}^{-m})-u^m({\bf q}^m,{\bf p}^{-m}) \right) x({\bf p}^m, {\bf p}^{-m})   \\
        &=x^m({\bf p}^m) \sum_{{\bf p}^{-m}} \left( u^m({\bf p}^m,{\bf p}^{-m})-u^m({\bf q}^m,{\bf p}^{-m}) \right)  \prod_{k\neq m} x^k({\bf p}^{k})  \\
        &=x^m({\bf p}^m)   \left( u^m({\bf p}^m,{ x}^{-m})-u^m({\bf q}^m,{ x}^{-m}) \right).
    \end{aligned}
    \end{equation}
    Since by definition of probability distributions, there exists $\textbf{p}^m$ such that $x^m({\bf p}^m)>0$ it follows that
 $u^m({\bf p}^m,{ x}^{-m})=u^m({\bf q}^m,{ x}^{-m})$ for all   ${\bf q}^m \in E^m$.  Thus,   $u^m({x}^m,{ x}^{-m})=u^m({\bf q}^m,{ x}^{-m})$ for all ${\bf q}^m \in E^m$. Since $m$ is arbitrary, the claim follows.
\end{proof}

It is well-known that in mixed Nash equilibria of games, players are
indifferent between all the pure strategies in the support of their
mixed strategy (see \cite{fudenberg1991gt}), i.e., if $x\in \prod_m
\Delta E^m$ is a mixed Nash equilibrium then
 \begin{equation}
  u^m(x^m, x^{-m}) \left\{
  \begin{aligned}
  &=u^m({\bf p}^m, x^{-m}) \quad\quad \mbox{ for all ${\bf p}^m$ such that  $x^m({\textbf{p}^m}) \geq 0$} \\
  &\geq u^m({\bf p}^m, x^{-m}) \quad\quad \mbox{ for all ${\bf p}^m$ such that  $x^m({\textbf{p}^m}) = 0$}.
  \end{aligned}\right.
  \end{equation}
The above corollary implies that at a mixed equilibrium of a harmonic
game, each player is indifferent between all its pure strategies,
including those which are not in the support of its mixed strategy.

We next define a particular mixed strategy profile, and show that it is an equilibrium in all harmonic games. 
%All harmonic games, have a  mixed equilibrium, which we define next.
\begin{definition} [Uniformly Mixed Strategy Profile] \label{def:unifMixed}
The \emph{uniformly mixed strategy} of  player $m$ is a mixed strategy where player $m$ uses $x_{{\bf
q}^m}=\frac{1}{h_m}$ for all ${\bf q}^m \in E^m$. 
Respectively, we
define the \emph{uniformly mixed strategy profile} as the one in which all
players use uniformly mixed strategies.
\end{definition}
Recall that rock-paper-scissors and matching pennies are examples of harmonic games, in which the uniformly mixed strategy profile is a mixed Nash equilibrium.
The next theorem shows that this is a general property of harmonic games
and the uniformly mixed strategy profile is always a
Nash equilibrium.

\begin{theorem}\label{theo:harmUniformlyMixedNE}
In harmonic games, the   uniformly mixed strategy profile is always a Nash equilibrium.
\end{theorem}
\begin{proof}
Let   ${\cal G}=\langle {\cal M}, \{u^m\}, \{E^m\} \rangle $ be a harmonic game, and $x$ be the uniformly mixed strategy profile. 
In order to prove the claim, we first state the following useful identity (see Appendix for a proof),  on the utility functions of harmonic games.
\begin{lemma} \label{lem:harmGameFlux}
Let   ${\cal G}=\langle {\cal M}, \{u^m\}, \{E^m\} \rangle $ be a harmonic game.
 Then {for all ${\bf q}^{m}, {\bf r}^{m}\in E^m$, $m\in {\cal M},$}
$
\sum_{{\bf p}^{-m}\in E^{-m}}{u}^m({\bf r}^{m},{\bf p}^{-m})-{u}^m({\bf q}^{m},{\bf p}^{-m})=0$.
\end{lemma}

Using this lemma, and the definition of the uniformly mixed strategy, it follows that
\begin{equation}\label{eq:harmGameUniformlyMixed}
\begin{aligned}
u^m({\bf q}^m, x^{-m})-u^m({\bf p}^m, x^{-m}) &= \sum_{{\bf p}^{-m}} c_m \left( u^m({\bf q}^m, {\bf p}^{-m})-u^m({\bf p}^m, {\bf p}^{-m}) \right)\\
&=c_m \sum_{{\bf p}^{-m}}  \left( u^m({\bf q}^m, {\bf p}^{-m})-u^m({\bf p}^m, {\bf p}^{-m}) \right)  \\
&= 0,
\end{aligned}
\end{equation}
where $c_m=\prod_{k\neq m} x^k({{\bf p}^k})=\prod_{k\neq m} \frac{1}{h_k}$. % and the last equality follows from  Lemma \ref{lem:harmGameFlux}.
Since $\mathbf{p}^m$ and $\mathbf{q}^m$ are arbitrary,  \eqref{eq:harmGameUniformlyMixed} implies that
\begin{equation}
u^m({ x}^m, x^{-m})=u^m({\bf p}^m, x^{-m})
\end{equation}
for all ${\bf p}^m \in E^m$, and by
Corollary~\ref{cor:harmGameEqCond}, $x$ is a mixed strategy Nash
equilibrium.
\end{proof}
In the sequel, we identify a basis for two-player normalized harmonic games, and through a simple dimension argument,  show that this Nash equilibrium is not
unique, for general harmonic games.  
In order to simplify the derivation of the basis result,  we first  provide a simple characterization of normalized harmonic games,  in terms of the  utility functions in the game.
\begin{theorem} \label{theo:zeroSumTheoNew}
The game ${\cal G}$ with utilities $u=\{u^m\}_{m\in{\cal M}}$ is a normalized harmonic game, i.e., it belongs to ${\cal H} $ if and only if 
$\sum_{m\in{\cal M}} h_m u^m=0$  and $\Pi_m u^m=u^m$ for all ${\bf m}\in {\cal M}$, where $h_m=|E^m|$.
\end{theorem}
\begin{proof} 
    By Definition \ref{def:subspaces},  ${\cal G}\in {\cal H}$ if and only if $\Pi u =u$ and $\delta_0^* D u =0$. Using the definitions of the operators, these conditions can alternatively be expressed as  $\Pi_m u^m=u^m$ and $\delta_0^* \sum_{m\in{\cal M}} D_m u^m=0$.
    By \eqref{eq:delta0Adj_DmAdj} and  the orthogonality of image spaces of operators $D_m$, the latter equality implies that $ \sum_{m\in{\cal M}} D_m^*D_m u^m=\sum_{m\in{\cal M}} \Delta_{0,m} u^m=0$. Using Lemma  \ref{theo:projDiff},  $ \Delta_{0,m}=h_m \Pi_m$, 
   and hence it follows that ${\cal G} \in {\cal H}$, if and only if %hence the last equality implies that
    \begin{equation}
     \sum_{m\in{\cal M}} h_m \Pi_m u^m=0 \qquad \mbox{ and, } \qquad \Pi_m u^m=u^m \mbox{ for all $m$}.
    \end{equation}
The claim follows by   replacing $ \Pi_m u^m$ in the summation with $u^m$.
\end{proof}
The above theorem implies that normalized harmonic games, where players have equal number of strategies, are zero-sum games, i.e., in such games the  payoffs of players add up to zero at all strategy profiles. We explore the further relations between zero-sum games and harmonic games in Section \ref{se:ZS}.

In the following theorem, we present a basis for two-player  normalized harmonic games. 
%It can be seen that for at  basis elements only two strategies of each  player are different than the rest of their strategies, in terms of the payoffs they yield. 
The idea behind our construction is to obtain a collection of games, in which both players have ``effectively'' two strategies (the payoffs are equal to zero, if other strategies are played), and ensure that they are linearly independent normalized harmonic games.
\begin{theorem} \label{theo:harmBasis}
Consider the set of two-player   games where the first player has $h_1$ strategies and the second player has $h_2$ strategies. 
For any $i\in \{1,\dots, h_1-1\}$ and $j\in \{1,\dots, h_2-1\}$, define bimatrix games ${\cal G}^{ij}$, with payoff matrices $(h_2 A^{ij}, - h_1 A^{ij} )$, where $A^{ij}\in \mathbb{R}^{h_1 \times h_2}$ is such that
\begin{equation} \label{eq:defineAUtil}
A^{ij}_{kl}= \left\{
\begin{aligned}
1 & \qquad \mbox{if $(k,l)=(i,j)$ or $(k,l)=(i+1,j+1)$,}\\
-1 & \qquad \mbox{if $(k,l)=(i+1,j)=(k,l)$ or $(k,l)=(i,j+1)$,}\\
0 & \qquad \mbox{otherwise.}\\
\end{aligned}
\right.
\end{equation}
The collection  $\{{\cal G}^{ij}\}$ provides a basis of $\cal H$.
\end{theorem}
\begin{proof}
It can be seen that each ${\cal G}^{ij}$ is normalized, since row and column sums of $A^{ij}$ is equal to zero. 
By Theorem~\ref{theo:zeroSumTheoNew} and \eqref{eq:defineAUtil}, it also follows that ${\cal G}^{ij}$ belongs to $\cal H$.
%By definition, a game belongs to ${\cal H}\oplus {\cal N}$,  only if $\delta_0^* \sum_m D_m u^m=\sum_{m}{\Delta_{0,m}}u^m=0$.  Note that for a normalized game $\Pi_m u^m =u^m$ and  $\Pi_k u^k=u^k$. Thus, Lemma \ref{theo:projDiff} implies that a normalized game belongs to $\cal H$  only if $\sum_{m} h_m u^m=0$.   Using this it follows that, ${\cal G}^{ij}$ is a harmonic game. 
 It can be seen from Proposition~\ref{theo:dimension}  that $\dim {\cal H}=(h_1- 1)(  h_2 -1)$, is equal to the cardinality of the collection $\{{\cal G}^{ij}\}$. Thus, in order to prove the claim, it is sufficient to prove that
\begin{equation} \label{eq:linDepA}
\sum_{i\in \{1,\dots, h_1-1\} } \sum_{j\in \{1,\dots, h_2-1\}} \alpha_{ij} A^{ij}=0,
\end{equation} 
only if $\alpha_{ij}=0$ for all $i,j$. 

Note that $A^{11}$ is the only matrix which has a nonzero entry in the first column and the first row. Thus, \eqref{eq:linDepA} implies that $\alpha_{11}=0$. Similarly it can be seen that $A^{11}$ and $A^{12}$ are the only matrices which have nonzero entries in the first row and the second column, thus $\alpha_{12}=0$. Proceeding iteratively it follows that  if \eqref{eq:linDepA}  holds, then  $\alpha_{ij}=0$ for all $i,j$ and the claim follows.
\end{proof}

The next example uses the basis introduced above, to show that in harmonic games, the uniformly mixed strategy profile is not necessarily the unique mixed Nash equilibrium.

\begin{example}
In this example we consider two-player harmonic games, where $E^1=
\{x,y\}$ and $E^2= \{a,b, c\}$.  Using Theorem \ref{theo:harmBasis},  a basis for normalized two-player harmonic games is   given in
Tables~\ref{tab:HarmBase1} and \ref{tab:HarmBase2}. Thus, any harmonic
game with these strategy sets, can be expressed as in Table
\ref{tab:HarmBase3}.  Consider some fixed $\alpha$ and $\beta$. As can be seen from Definition \ref{def:mixed}, the
mixed equilibria for this game are given by
\[
(\textstyle{\frac{1}{2}}, \textstyle{\frac{1}{2}}) \times (\theta_1, \theta_2, \theta_3)
\]
where $\theta_1$, $\theta_2$ and $\theta_3$ are scalars that satisfy
$\theta_1+\theta_2+\theta_3=1 $, $\theta_1, \theta_2, \theta_3 \geq 0$
and $\theta_1 (6\alpha ) +\theta_2 (-6 \alpha + 6\beta ) +\theta_3
(-6\beta )=0$. Note that since there are two linear equations in three
variables, this system has a continuum of solutions. Moreover, since
$(\theta_1, \theta_2, \theta_3)=(\frac{1}{3},\frac{1}{3},\frac{1}{3})$
is a solution, it follows that there is a continuum of solutions for
which $\theta_1, \theta_2, \theta_3 \geq 0$.

Since this is true for any $\alpha$, $\beta$,
% and $\cal H$ is spanned by the game in Table \ref{tab:HarmBase3} by choosing $\alpha, \beta \in {\mathbb{R}}$, 
we  conclude that all games in $\cal H$
have uncountably many mixed equilibria. Additionally, since the nonstrategic component does not affect the equilibrium properties of a game it follows that all harmonic games on $E^1 \times E^2$ (all games in ${\cal H}\oplus {\cal N}$) have uncountably many mixed Nash equilibria.
        \begin{table}[ht]
\centering
\subfloat[Basis element $1$]{ 
\label{tab:HarmBase1}
\begin{tabular}{ | c | c | c | c |}
\hline
 & $a$ & $b$ & $c$  \\ \hline
$x$ & 3,  -2 &   -3, 2  &  0, 0    \\ \hline
$y$ &  -3,  2 &  3,  -2 &  0, 0   \\ \hline
\end{tabular}
}
\quad\quad\quad
\subfloat[Basis element $2$]{ \label{tab:HarmBase2}
    \begin{tabular}{ | c | c | c | c |}
\hline
 & $a$ & $b$ & $c$  \\ \hline
$x$ & 0, 0 &3,  -2 &   -3, 2       \\ \hline
$y$ &  0, 0& -3,  2 &  3,  -2      \\ \hline
\end{tabular}
}

\subfloat[A game in $\cal H$]{ \label{tab:HarmBase3}
\begin{tabular}{ | c | c | c | c |}
\hline
 & $a$ & $b$ & $c$  \\ \hline
$x$ &  $3\alpha $ , $-2\alpha $ & $-3\alpha+3\beta$, $2\alpha-2\beta$ &  $ -3\beta$, $ 2\beta$   \\ \hline
$y$ & $-3\alpha  $, $2\alpha $ & $3\alpha-3\beta$, $-2\alpha+2\beta$  & $ 3\beta$, $-2\beta$  \\ \hline
\end{tabular}
}
\caption{Basis of $\cal H$}
\label{tab:HarmBase}
\end{table}
\end{example}

Using this basis,  we characterize in the following theorem, the
correlated equilibria in two-player harmonic games. Interestingly, our
results suggest that  in two-player harmonic games, the set
of mixed Nash equilibria and correlated equilibria generically coincide.

\begin{theorem} \label{theo:genHarmEq}
Consider the set of two-player harmonic games where the first player
has $h_1$ strategies and the second player has $h_2$ strategies.
Without loss of generality assume that $h_1 \geq h_2$. Generically,
    \begin{itemize}
        \item[(i)] Every correlated equilibrium is a mixed Nash equilibrium, where the player with minimum number of strategies uses the uniformly mixed strategy.
        \item[(ii)] The dimension of the set of correlated equilibria is $h_1 - h_2$
    \end{itemize}
\end{theorem}
\begin{proof}
As discussed earlier,  nonstrategic components of games do not affect the equilibrium sets. Thus, to prove that (i) and (ii) are generically true for harmonic games, it is sufficient to prove that they generically hold for normalized harmonic games.

Consider a two-player normalized harmonic game with payoff matrices
$(A,B)$, where $A,B \in \mathbb{R}^{h_1 \times h_2}$.  By Theorem
\ref{theo:zeroSumTheoNew}, it follows that $A=- \frac{h_2}{h_1}B$.
Denote by $e_1$ (similarly $e_2$), the $h_1$ (similarly $h_2$)
dimensional vector, all entries of which are identically equal to
$1$. Since the game is normalized, it follows that $e_1^T A =0 $ and
$B e_2 = - \frac{h_1}{h_2} A e_2 =0$.

Let $x$ be a correlated equilibrium of this game.
For each $\mathbf{p}^1 \in E^1$, denote by
$x({\textbf{p}^1}, ~{\cdot}) \in \mathbb{R}^{h_2}$   the vector of probabilities $[x({\textbf{p}^1}, \textbf{p}^2)]_{{\bf p}^2} $.
    By Proposition \ref{prop:corEqCharNew} (iii), it follows that these vectors satisfy the condition
    \begin{equation} \label{eq:corrMatrixDef}
     A x(\mathbf{p}^1, \cdot) =0.
    \end{equation}
Note that we need to characterize the kernel of the payoff matrix $A$, to identify the correlated equilibria. For that reason, we state the following technical lemma:
\begin{lemma} \label{lem:genericMat}
Consider the set of normalized harmonic games  in Theorem \ref{theo:genHarmEq}.
Generically, the payoff matrices of players have their row and column ranks equal to $h_2-1$.
\end{lemma}
\begin{proof}
The payoff matrices of the players satisfy $A=- \frac{h_2}{h_1}B$, so they have the same row and column rank.
It follows from Theorem \ref{theo:harmBasis} that the collection of matrices $\{A^{ij} \}$ span the payoff matrices of harmonic games. It can be seen that the matrices in the span of this collection generically have row and column rank equal to  $h_2-1$, and the claim follows.
\end{proof}
Using this lemma, it follows that generically the kernel of $A$ is 1 dimensional. As shown earlier, $e_2$ is in the kernel of $A$, thus, \eqref{eq:corrMatrixDef}, implies that generically $x({\textbf{p}^1}, ~{\cdot})$ has the form
 $ x({\textbf{p}^1}, ~{\cdot})=c_{\mathbf{p}^1}  e_2$, for some $c_{\mathbf{p}^1} \in \mathbb{R}$. 
Since $x$ is a probability distribution, 
 the definition of $x({\textbf{p}^1}, ~{\cdot})$ implies that $x({\textbf{p}^1}, \textbf{p}^2)= c_{\mathbf{p}^1}
 \geq 0$, and  $\sum_{\mathbf{p}^1, \mathbf{p}^2} x({\textbf{p}^1}, \textbf{p}^2)=h_2 \sum_{\mathbf{p}^1} c_{\mathbf{p}^1} =1$.  Thus, it follows that $x({\textbf{p}^1}, \textbf{p}^2)= c_{\mathbf{p}^1}=\frac{\alpha_{\mathbf{p}^1}}{h_2}$, for some $\alpha_{\mathbf{p}^1}\geq0$ such that $\sum_{\mathbf{p}^1} \alpha_{\mathbf{p}^1} =1$. It can be seen from this description that generically, the correlated equilibria are mixed equilibria where the first player uses the probability distribution $x^1=\alpha \triangleq [\alpha_{\mathbf{p}^1}]_{\mathbf{p}^1} \in \Delta E^1 $ and the second player uses the distribution 
 $x^2=\left[ \frac{1}{h_2} \right]_{\mathbf{p}^2}$. 
 
Since the correlated equilibria have this form, it can be seen using Proposition \ref{prop:corEqCharNew} (iii) for the second player that 
\begin{equation} \label{eq:CorrEqCondNew}
            \sum_{{\bf p}^{1}} u^2({\bf p}^{1},{\bf q}^2)  x({\bf p}^1, {\bf p}^{2})  =           \frac{1}{h_2} \sum_{{\bf p}^{1}} u^2({\bf p}^{1},{\bf q}^2)  \alpha_{\mathbf{p}_1} =0,
\end{equation}
where $\alpha \in \Delta E^1$.
The above  condition  can be restated using the payoff matrices as follows:
\begin{equation} \label{eq:condOnA}
\alpha^T B = -  \frac{h_1}{h_2} \alpha^T A=0,
\end{equation}
where $\alpha \in \Delta E^1$.
 Since, the row rank of $A$ is $h_2-1$, the dimension of $\alpha$ that satisfies  \eqref{eq:condOnA}  is $h_1-h_2+1$. Note that since $\alpha$ is a probability distribution, it also satisfies the condition $\alpha^T e_1=1 $. Note that since $e_1^T A=0$, this condition is orthogonal to the ones in \eqref{eq:condOnA}. Hence, it follows that the dimension of $\alpha$ which satisfies the correlated equilibrium conditions in \eqref{eq:condOnA} (other than the positivity) is $h_1 -h_2$.
 On the other hand, $\alpha =\frac{1}{h_1} e_1$ gives a correlated equilibrium (by Theorem \ref{theo:harmUniformlyMixedNE}), thus the positivity condition does not change the dimension of the set of correlated equilibria, and   the dimension  is generically $h_1 -h_2$.
\end{proof}

An immediate implication of this theorem is the following:
\begin{corollary}
 In two-player harmonic games where players have equal number of
 strategies, the uniformly mixed strategy is generically the unique correlated
 equilibrium.
\end{corollary}

Note that Theorem~\ref{theo:genHarmEq} implies that in two-player
harmonic games, generically there are no correlated equilibria that
are not mixed equilibria. This statement fails, when the number of
players is more than two, as shown in the following theorem.
\begin{theorem}
Consider a $M$-player harmonic game, where $M > 2$, and in which each
player has $h$ strategies such that $h^M>M(h^2-1)+1$. The set of
correlated equilibria is strictly larger than the set of mixed Nash
equilibria: The set of correlated equilibria has dimension at least
$h^M-1- Mh(h-1) $, and the set of mixed equilibria has dimension at
most $M (h -1)$.
\end{theorem}
\begin{proof}
Since each player has $h$ strategies, the set of mixed strategies has
dimension $M (h -1)$, and this is a trivial upper bound on the
dimension of the set of mixed equilibria. The set of correlated
equilibria, on the other hand, is defined by the equalities in
Proposition  \ref{prop:corEqCharNew}. Note that there are
$Mh(h-1)$ such equalities and the dimension of $\Delta E$ is $h^M-1$,
hence the dimension of the correlated equilibria is at least $h^M-1-
Mh(h-1) $ (by ignoring possible dependence of the equalities).

The difference in the dimensions   implies that the set of correlated equilibria is strictly larger than the set of mixed equilibria.
%follows from the difference of the bounds on the dimensions of the set of correlated equilibria and  mixed equilibria.
\end{proof}
Note that this theorem can be easily generalized to the case when
players have different number of strategies. An interesting problem is
to find the exact dimensions of the set of mixed Nash and correlated
equilibria when there are more than two players. However,
due to complicated dependence relations of the correlated equilibrium
conditions in Proposition~\ref{prop:corEqCharNew}, we do not pursue this
question in this paper, and leave it as a future problem.

\subsection{Nonstrategic Component and Efficiency in Games} 
\label{se:nonStrComp}

We first consider games for which the potential and harmonic components are equal to zero. In such games all pairwise comparisons are equal to zero, hence each player is indifferent between any of his strategies given any strategies  of other players. It is thus immediate that all strategy profiles are Nash equilibria in such games.

More generally, from the definition of the nonstrategic component it can be seen that in any game, the pairwise comparisons   are functions of only the potential and harmonic components of the game.
Thus, the nonstrategic component has \emph{no effect whatsoever} on the equilibrium properties of games.
However, the nonstrategic component is of interest mainly through its effect on the efficiency properties of games, as  discussed in the rest of this section. The efficiency measure we focus on  is Pareto optimality.

\begin{definition}[Pareto Optimality] \label{def:PoAPoS}
    A strategy profile $\bf p$ is Pareto optimal if and only if there does not exist another strategy profile $\bf q$ such that all players weakly increase their payoffs and one player strictly increases its payoff, i.e,
    \begin{equation}
    \begin{aligned}
        u ^m ({\bf q})& \geq u ^m({\bf p}), \qquad \qquad \mbox{for all $m\in {\cal M}$}  \\
        u ^k ({\bf q})& > u ^k({\bf p}), \qquad \qquad \mbox{for some $k\in {\cal M}$}.
    \end{aligned}
    \end{equation}

\end{definition}

We first state a preliminary lemma, which will be useful in the subsequent analysis.
\begin{lemma} \label{lem:NSHarm0}
    Let  $\cal G$ be a game with utilities $\{u^m\} $. There exists a  game $\hat{\cal G}$ with utilities $\{\hat{u}^m\} $  such that (i) the potential and harmonic components of $\hat{\cal G}$ are identical to  these of $\cal G$ and (ii) in $\hat{\cal G}$ all players get zero payoff at all strategy profiles that are pure Nash equilibria of ${\cal G}$.
\end{lemma}
\begin{proof}
    Let ${N}_{\cal G}$ be the set of pure  Nash equilibria of $\cal G$.
    First observe that if there are $m$-comparable equilibria in $\cal G$ player $m$ receives the same payoff in these equilibria, i.e., if ${\bf p},\textbf{q}\in {N}_{\cal G}$  and ${\bf p}=(\textbf{p}^m, \textbf{p}^{-m})$, ${\bf q}=(\textbf{q}^m, \textbf{p}^{-m})$ for some $m$, then $u^m(\textbf{p}^m, \textbf{p}^{-m})=u^m(\textbf{q}^m, \textbf{p}^{-m})$.
        This equality holds since otherwise, player $m$ would have incentive to improve its payoff at $\bf p$ or  $\bf q$ by switching to a strategy profile with better payoff, and this contradicts with $\textbf{p}$ and $\textbf{q}$ being Nash equilibria of ${\cal G}$.

         Define the  game $\hat{\cal G}$ with utilities $\{\hat{u}^m\}_{m\in{\cal M}}$ such that
        \begin{equation} \label{eq:defUhat_pareto1}
            \hat{u}^m(\textbf{p})= \left\{
                \begin{aligned}
                    0 \qquad \qquad &\mbox{ ${\bf p}\in{N}_{\cal G}$} \\
                    u^m({\bf p})-u^m({\bf q}) \qquad \qquad &\mbox{if there exists a ${\bf q}\in{N}_{\cal G}$ which is $m$-comparable with $\bf p$} \\
                    u^m({\bf p}) \qquad \qquad &\mbox{otherwise.}
                \end{aligned}
                \right.
        \end{equation}
        for all $m\in{\cal M}$, $\textbf{p}\in E$. Note that  $\hat{u}^m$ is well defined since in $\cal G$  player $m$ gets the same payoff in all $\textbf{p}\in{N}_{\cal G}$ that are $m$-comparable. Note that in $\hat{\cal G}$ all players receive zero payoff at  all strategy profiles $\textbf{p}\in {N}_{\cal G}$.
To prove the claim, it suffices to  show that    $\cal G$ and $\hat{\cal G}$   have the same potential and harmonic components,  or equivalently  the game with utilities $\{u^m-\hat{u}^m\}_{m\in {\cal M}}$ is nonstrategic, i.e., belongs to $\cal N$.
%Since this only holds when $\cal G$ and $\hat{\cal G}$   have the same potential and harmonic components, the claim follows if $\{u^m-\hat{u}^m\}_{m\in {\cal M}}\in {\cal N}$.

        In order to prove that the difference is nonstrategic, we first show that the pairwise comparisons of games with utilities $\{u^m\}_{m\in {\cal M}}$ and $\{\hat{u}^m\}_{m\in {\cal M}}$ are the same. Note that by \eqref{eq:defUhat_pareto1} given $m$-comparable $\bf p$ and $\bf q$, $u^m(\textbf{p})-u^m(\textbf{q})=\hat{u}^m(\textbf{p})-\hat{u}^m(\textbf{q})$, if there is no $\textbf{r} \in {N}_{\cal G}$ that is $m$-comparable with $\bf p$ or $\bf q$.         
        If there exists $\textbf{r}\in{N}_{\cal G}$   that is $m$-comparable with $\bf p$, then it is also $m$-comparable with $\bf q$, hence
    it follows by \eqref{eq:defUhat_pareto1} that  $\hat{u}^m(\textbf{p})-\hat{u}^m(\textbf{q})=u^m(\textbf{p})-u^m(\textbf{r})-u^m(\textbf{q})+u^m(\textbf{r})=u^m(\textbf{p})-u^m(\textbf{q})$. Note that these equalities hold even if $\bf p$ or $\bf q$ is in ${N}_{\cal G}$.

        Thus, for any $m$-comparable  $\bf p$ and $\bf q$ it follows that
        \begin{equation*}
            \left( u^m(\textbf{p})-\hat{u}^m(\textbf{p}) \right) - \left( u^m(\textbf{q})-\hat{u}^m(\textbf{q}) \right)=0,
        \end{equation*}
         hence  the game with utilities $\{u^m-\hat{u}^m\}_{m\in {\cal M}}$ is nonstrategic and the claim follows.
 \end{proof}

Note that if two games differ only in their nonstrategic components, the pairwise comparisons, and hence the equilibria  of these games are identical.
Therefore, an immediate implication of the above lemma is that for a given game there exists another game with same potential and harmonic components such that the payoffs at all Nash equilibria are equal to zero. We use this to prove the following Pareto optimality result.

\begin{theorem}\label{theo:allParetoNE}
Let  $\cal G$ be a game with utilities $\{u^m\}$. There exists a  game $\bar{\cal G}$ with utilities $\{\bar{u}^m\}$  such that (i) the potential and harmonic components of $\bar{\cal G}$ are identical to  these of $\cal G$ and (ii) in $\bar{\cal G}$ the set of pure NE coincides with the set of Pareto optimal strategy profiles.
%all players get zero payoff at all pure Nash equilibria of ${\cal G}$.
%       In a game $\cal G$, with utilities $\{u^m\}_{m\in{\cal M}}$, the nonstrategic component can be changed in such a way that all  NE are Pareto  optimal and all Pareto optimal strategy profiles are NE.
\end{theorem}
\begin{proof}
Games that differ only in nonstrategic components have identical pairwise comparisons, hence the set of Nash equilibria (NE) is the same  for such games.     Let ${N}_{\cal G}$ denote the set of pure NE of $\cal G$, or equivalently the set of pure NE of a game which differs from $\cal G$ only by its nonstrategic component.

    By Lemma \ref{lem:NSHarm0}, it follows that for any game $\cal G$ there exists a game such that the two games differ only in their nonstrategic components and all players receive zero payoffs at all pure NE (strategy profiles in ${N}_{\cal G}$).
     Therefore, without loss of generality, we let $\cal G$ be a game  in which all players receive zero payoffs at all NE. Given such a game, let $\alpha =1+ \max_{m,\textbf{p}}u^m(\textbf{p})$.   Consider the game $\bar{\cal G}$ with utilities $\{\bar{u}^m\}_{m\in{\cal M}}$ such that
        \begin{equation*}
            \bar{u}^m(\textbf{p})= \left\{
                \begin{aligned}
                u^m({\bf p})  \qquad &\mbox{if $\textbf{p}\in {N}_{\cal G}$ or if there exists a ${\bf q}\in{N}_{\cal G}$ which is $m$-comparable with $\bf p$} \\
                u^m({\bf p})-\alpha  \qquad &\mbox{otherwise.}
                \end{aligned}
                \right.
    \end{equation*}
        for all $m\in{\cal M}$, $\textbf{p}\in E$.

        Consider $m$-comparable strategy profiles  $\bf p$ and $\bf q$. Observe that if there exists a strategy profile $\textbf{r}$ that is $m$-comparable with $\textbf{p}$, it is also ${m}$-comparable with $\textbf{q}$ since by definition of $m$-comparable strategy profiles $\textbf{p}^{-m}= \textbf{r}^{-m}=\textbf{q}^{-m}$.

         Assume that
         there is a NE that is $m$-comparable with $\bf p$ or $\bf q$, then by definition of $\bar{u}^m$ it follows that $u^m(\textbf{p})-u^m(\textbf{q})=\bar{u}^m(\textbf{p})-\bar{u}^m(\textbf{q})$. On the contrary if there is no NE that is $m$-comparable with $\textbf{p}$ or $\textbf{q}$ then $\bar{u}^m(\textbf{p})-\bar{u}^m(\textbf{q})={u}^m(\textbf{p})-\alpha-{u}^m(\textbf{q})+\alpha={u}^m(\textbf{p})-{u}^m(\textbf{q})$.
          Hence  $\cal G$ and $\bar{\cal G}$ have identical pairwise comparisons, and thus
         the game with utilities $\{u^m-\bar{u}^m\}_{m\in {\cal M}}$ is nonstrategic.

         We prove the claim, by showing that  at all strategy profiles that are not an equilibrium in $\bar{\cal G}$ (equivalently in $\cal G$), the players receive nonpositive payoffs and at least one player receives negative   payoff and at all NE all players receive zero payoff. This immediately implies that strategy profiles, that are not NE cannot be Pareto optimal, as deviation to a NE increases the payoff of at least one player and the payoff of other players do not decrease by such a deviation. Additionally, it implies that all NE are Pareto optimal, since at all NE all players receive the same payoff, and deviation to a strategy profile that is not a NE strictly decreases the payoff of at least a single player.

         By construction it follows that at all NE all players receive zero payoff. Let $\textbf{p}$ be a strategy profile that is not a NE. If there is some $m$ for which $\textbf{p}$ is not $m$-comparable to a NE, then it follows that $\bar{u}^m(\textbf{p})={u}^m(\textbf{p})-\alpha \leq -1$. If on the other hand, $\textbf{p}$ is $m$-comparable to a NE, then $\bar{u}^m(\textbf{p})\leq 0$, since payoffs are equal to zero at NE. Thus, at any strategy profile, $\bf p$, that is not a NE players receive nonpositive payoffs, and additionally if for some player $m$, $\bf p$ is not  $m$-comparable to a NE, player $m$ receives strictly negative payoff.

         To finish the proof we need to show that if $\bf p$ is $m$-comparable to a NE for all $m\in{\cal M}$, then it still follows that  $\bar{u}^m(\textbf{p}) < 0$ for some $m\in{\cal M}$. Assume that this is not true and $\bar{u}^m(\textbf{p}) = 0$ for all $m\in{\cal M}$. Since $\bf p$ is not a NE, there is at least  one player, say $m$, who can get strictly positive payoff by deviating to a different strategy profile. Therefore this player has strictly positive payoff after its deviation. However, as argued earlier payoffs are nonpositive at strategy profiles that are not NE, and zero at NE.
         Thus. we reach a contradiction and  $\bar{u}^m(\textbf{p}) < 0$ for some $m\in{\cal M}$.

         Therefore, it follows that all players have zero payoffs at all NE, and  at any other strategy profile all players have nonpositive payoffs and at least one player has strictly negative payoff.
\end{proof}

Note that it is possible to obtain similar results for other efficiency measures using similar arguments to those given in this section. This direction will not be pursued in this paper.
The above theorem suggests that the difference in the nonstrategic component of games that are otherwise identical may cause the efficiency properties of these game to be very different. In particular, in one of the games all equilibria may be Pareto optimal when this is not the case for the other game. Therefore, although the nonstrategic component does not change the pairwise comparisons and equilibrium properties in a game it plays a key role in Pareto optimality of equilibria.

\subsection{Zero-Sum Games and Identical Interest Games} \label{se:ZS}

In this section we present a different decomposition of the space of games, and discuss its relation to our decomposition. To simplify the presentation, we focus on bimatrix games, where each player has $h$ strategies. 
Before introducing the   decomposition, we define zero-sum games and identical interest games.
\begin{definition}[Zero-sum and Identical Interest Games]
Let $\cal G$ denote the bimatrix game with payoff matrices $(A,B)$.  $\cal G$ is  a \emph{zero-sum game}, if $A+B=0$, and  $\cal G$ is an \emph{identical interest game}, if $A=B$.
% and if $A=B$, $\cal G$ is  an identical interest game.
\end{definition}

We denote the set of zero-sum games by $\cal Z$, and the set of identical interest games by $\cal I$. 
Since these sets are defined by equality constraints on the payoff matrices, it follows that  they are subspaces. 

The idea of decomposing a game to an identical interest game and a
zero-sum game was previously mentioned in the literature for
two-player games, \cite{Basar:1974p1629}.  The following lemma implies
that $\cal Z$ and $\cal I$ decomposition of the set of games, has the
direct sum property.
\begin{lemma}
The space of two-player games ${\cal G}_{{\cal M},E}$ is a direct sum
of subspaces of zero-sum and identical interest games, i.e., ${\cal
G}_{{\cal M},E}={\cal Z} \oplus {\cal I}$.
\end{lemma}
\begin{proof}
Consider a bimatrix game with utilities $(u^1,u^2)$. 
Observe that this game can be decomposed to the games with payoff functions $(\frac{u^1-u^2}{2}, \frac{u^2-u^1}{2})$ and $(\frac{u^1+u^2}{2},\frac{u^1+u^2}{2})$. Clearly the former game is a zero-sum game, where the latter is an identical interest game. Since the initial game was arbitrary, it follows that any game can be decomposed to a zero-sum game and an identical interest game. The direct sum property follows, since  for two-player zero-sum and identical interest  games, with  utility functions  $(u,-u)$ and $(v,v)$ respectively,  if 
$(u+v, u-v)=(0,0)$, then $u=v=0$. 
\end{proof}

Note that Theorem~\ref{theo:zeroSumTheoNew} suggests that two-player
normalized harmonic games, where players have equal number of
strategies are zero-sum.\footnote{In addition, if the definition of
zero-sum is generalized to include multiplayer games where payoffs of
all players add up to zero, then it can be seen that normalized
harmonic games where players have equal number of strategies are still zero-sum games.}
%This observation suggests that zero-sum games are closely related to harmonic games. 
Also, it immediately follows by checking the definitions that  identical interest games are potential games.
This intuitively suggests that  the  zero-sum and identical interest game decomposition closely relates to  our decomposition.
 In the following theorem, we establish this relation by characterizing the dimensions of   the intersections of the subspaces $\cal Z$ and $\cal I$, with the sets of potential and harmonic games. We provide a proof in the Appendix.
\begin{theorem} \label{theo:ZSPH_dim}
Consider two-player  games, in which each player has $h$ strategies. The dimensions of intersections of the subspaces  of zero-sum and identical interest games ($\cal Z$ and  $\cal I$) with the subspaces of potential  and harmonic games (${\cal P}\oplus{\cal N}$ and ${\cal H}\oplus{\cal N}$) are  as in the following table.
\begin{table}[h]
\centering
\begin{tabular}{|c|c|c|c|}
\hline
	     	  & ${\cal Z}$ & $\cal I$ & ${\cal Z} \oplus {\cal I}$ \\ \hline
${\cal P}\oplus{\cal N}$ & $2h-1$		& $h^2$ & $h^2+2h-1$ \\ \hline
${\cal H}\oplus{\cal N}$ & $h^2-2h+2$ &1 & $h^2+1$	\\ \hline
${\cal P}\oplus{\cal H} \oplus{\cal N}$ & $h^2$ & $h^2$ & $2h^2$\\ 
\hline
\end{tabular} 
\label{tab:dimensions}
\caption{Dimensions of subspaces of games and their intersections}
\end{table}
\end{theorem}

The above theorem suggests that the dimensions  of harmonic games and zero-sum games  (and similarly identical interest games and potential games) are close to the dimension of their intersections. Thus, zero-sum games are in general closely related to harmonic games, and identical interest games are related to potential games. On the other hand, it is possible to find instances of zero-sum games that are potential games, and  not harmonic games (see Table \ref{tab:zSExample}).

\begin{table}[ht] 
\centering
\subfloat[Payoffs ]{
\begin{tabular}{ | c | c | c | }
\hline
 & $a$ & $b$  \\ \hline
$x$ &  0, 0 & 1,-1  \\ \hline
$y$ & -1, 1 & 0, 0   \\ \hline
\end{tabular}
}
\qquad
\subfloat[Potential function ]{
\begin{tabular}{ | c | c | c | }
\hline
 & $a$ & $b$  \\ \hline
$x$ &  ~~2~~ & ~~1~~  \\ \hline
$y$ & ~~1~~ & ~~0~~   \\ \hline
\end{tabular}
}
\caption{A zero-sum potential game}
\label{tab:zSExample}
\end{table}

\begin{comment}
\begin{table}[ht] 
\centering
\subfloat[Payoffs ]{
\begin{tabular}{ | c | c | c | }
\hline
 & $a$ & $b$  \\ \hline
$x$ &  1, -1 & 1,-1  \\ \hline
$y$ & -1, 1 & -1, 1   \\ \hline
\end{tabular}
}
\qquad
\subfloat[Potential function ]{
\begin{tabular}{ | c | c | c | }
\hline
 & $a$ & $b$  \\ \hline
$x$ &  ~~2~~ & ~~2~~  \\ \hline
$y$ & ~~0~~ & ~~0~~   \\ \hline
\end{tabular}
}
\caption{A zero-sum potential game}
\label{tab:zSExample}
\end{table}
\end{comment}

In general, the identical interest component is a potential game, and
it can be used to approximate a given game with a potential game.
However, as illustrated in Table~\ref{tab:zSExample2}, this approximation need not yield the closest potential game to a given game. In this example, despite the fact that the original game is a potential game, the zero-sum and identical interest game  decomposition may lead to a potential game
which is much farther than the closest potential game

\begin{table}[ht] 
\centering
\subfloat[Payoffs in  $\cal G$ ]{
\begin{tabular}{ | c | c | c | }
\hline
 & $a$ & $b$  \\ \hline
$x$ &  1, 1 & 1,-1  \\ \hline
$y$ &  -1, 1 & -1, -1   \\ \hline
\end{tabular}
}
\qquad
\subfloat[Potential function of $\cal G$ ]{
\begin{tabular}{ | c | c | c | }
\hline
 & $a$ & $b$  \\ \hline
$x$ &  ~~4~~ & ~~2~~  \\ \hline
$y$ & ~~2~~ & ~~0~~   \\ \hline
\end{tabular}
}
\qquad
\subfloat[Payoffs in  ${\cal G}_Z$ ]{
\begin{tabular}{ | c | c | c | }
\hline
 & $a$ & $b$  \\ \hline
$x$ &  0, 0 & 1,-1  \\ \hline
$y$ &  -1,  1 & 0, 0   \\ \hline
\end{tabular}
}
\qquad
\subfloat[Payoffs in ${\cal G}_I$ ]{
\begin{tabular}{ | c | c | c | }
\hline
 & $a$ & $b$  \\ \hline
$x$ &  1, 1 & 0,0  \\ \hline
$y$ &  0,  0 & -1, -1   \\ \hline
\end{tabular}
}
\caption{A potential game $\cal G$ and its zero-sum (${\cal G}_{Z}$) and identical interest components (${\cal G}_I$).}
\label{tab:zSExample2}
\end{table}
We believe that the decomposition presented in
Section~\ref{se:canonicalRep} is more natural than the zero-sum
identical interest game decomposition, as it clearly separates the
strategic (${\cal P} \oplus {\cal H}$) and nonstrategic (${\cal N}$)
components of games and further identifies components, such as
potential and harmonic components, with distinct strategic properties.
In addition, it is invariant under trivial manipulations that do not
change the strategic interactions, i.e., changes in the nonstrategic
component.

\section{Projections onto Potential and Harmonic Games} 
\label{se:projection}

In this section, we discuss projections of games onto the subspaces of
potential and harmonic games.
In Section~\ref{se:orthDecomp}, we defined the subspaces $\mathcal{P},
\mathcal{H}, \mathcal{N}$ of potential, harmonic, and nonstrategic
components, respectively. We also proved that they provide a direct
sum decomposition of the space of all games. In this section, we show
that under an appropriately defined inner product in
$\mathcal{G}_{\mathcal{M},E}$, the harmonic, potential and
nonstrategic subspaces become \emph{orthogonal}. We use our
decomposition result together with this inner product to obtain
projections of games to these subspaces, i.e., for an arbitrary game,
we present closed-form expressions for the ``closest'' potential and
harmonic games with respect to this inner product.

Let ${\cal G}, \hat{\cal G}  $ be two 
games in ${\cal G}_{{\cal M},E}$. We define the inner product on ${\cal G}_{{\cal M},E}$ as
 \begin{equation} \label{eq:innerProdGames}
    \langle {\cal G}, \hat{\cal G} \rangle_{{\cal M},E}  \triangleq \sum_{m\in{\cal M}} h_m \langle u^m , \hat{u}^m \rangle,
 \end{equation}
where the inner product in the right hand side is the inner product of
$C_0$ as defined in \eqref{eq:innerProdonC1}, i.e., it is the inner
product of the space of functions defined on $E$. Note that it can be
easily checked that \eqref{eq:innerProdGames} is an inner product, by
observing that it is a weighted version of the standard inner product
in $C_0^M$. The given inner product also induces a norm which will
help us quantify the distance between games. We define the norm on
${\cal G}_{{\cal M},E}$ as follows:
\begin{equation} \label{eq:normGames}
 || {\cal G} ||_{{\cal M},E}^2  = \langle {\cal G}, {\cal G} \rangle_{{\cal M},E}.
 \end{equation}
Note that this norm also corresponds to a weighted $l_2$ norm defined on the space  $C_0^M$.
% captures the $2$ norm change in the utilities if every player has same utilities

 Next we prove that the potential, harmonic and nonstrategic subspaces are orthogonal under this inner product.
 \begin{theorem} \label{theo:orthOfDecomposition}
Under the inner product introduced in \eqref{eq:innerProdGames}, we
    have ${\cal P} \perp {\cal H} \perp {\cal N}$, i.e., the
    potential, harmonic and nonstrategic subspaces are orthogonal.
 \end{theorem}
 \begin{proof}
    Let $\{u^m_P\}_{m\in{\cal M}}={\cal G}_P\in{\cal P}$, $\{{u}_H^m\}_{m\in{\cal M}}={\cal G}_H \in {\cal H}$ and  $\{{u}_N^m\}_{m\in{\cal M}}={\cal G}_N \in {\cal N}$ be arbitrary games in $\cal P$, $\cal H$ and $\cal N$ respectively. In order to prove the claim we will first prove ${\cal G}_N \perp {\cal G}_H$  and ${\cal G}_N \perp  {\cal G}_P$. Secondly we prove  ${\cal G}_P \perp {\cal G}_H$. Since the games are arbitrary the first part will imply that ${\cal N} \perp {\cal H}$ and ${\cal N} \perp {\cal P}$ and the second part will imply that ${\cal P} \perp {\cal H}$ proving the claim.

    Note that by definition $u^m_N \in \ker (D_m)$ for all $m\in{\cal M}$ and $u^m_P, u^m_N $ are in the orthogonal complement of $\ker (D_m)$ since $\Pi_m u^m_P=u^m_P$, $\Pi_m u^m_H=u^m_H$ and $\Pi_m$ is the projection operator to the orthogonal complement of  $\ker (D_m)$. This implies that $\langle u^m_P, u^m_N\rangle = \langle u^m_H, u^m_N\rangle= 0 $ for all $m\in{\cal M}$ and hence using the inner product introduced in \eqref{eq:innerProdGames} it follows that ${\cal G}_N \perp {\cal G}_H$  and ${\cal G}_N \perp  {\cal G}_P$.

    Next observe that for all $m\in{\cal M}$,
    \begin{equation*}
    \begin{aligned}
        \langle u^m_P, u^m_H\rangle &= \langle D_m^\dagger D_m u^m_P, u^m_H\rangle 
        &= \frac{1}{h_m} \langle D_m^* D_m \phi, u^m_H\rangle 
        &= \frac{1}{h_m} \langle \phi,  D_m^* D_m  u^m_H\rangle,
    \end{aligned}
    \end{equation*}
    where the first equality follows from $\Pi_m u^m_P= u^m_P$, and the second equality follows from Lemma \ref{theo:projDiff} and the fact that $D_m u^m_P = D_m \phi$. The third equality uses the properties of the   operators $D_m$ and $D_m^*$.
    % for $C_0$ and $C_1$ inner products.
    Therefore,
    \begin{equation*}
    \begin{aligned}
        \langle {\cal G}_P ,{\cal G}_H \rangle_{{\cal M},E} &=  \sum_{m\in{\cal M}} \langle \phi,  D_m^* D_m  u^m_H\rangle 
        &= \langle \phi, \sum_{m\in{\cal M}} D_m^* D_m  u^m_H\rangle 
        &= \langle \phi, \delta_0^* \sum_{m\in{\cal M}}  D_m  u^m_H\rangle =0.
    \end{aligned}
    \end{equation*}
    Since $\delta_0^*\sum_{m\in{\cal M}}  D_m  u^m_H=0$ by the definition of $\cal H$.
    Here the last equality follows using $\delta_0^* = \sum_m D_m ^*$ and orthogonality of the image spaces of $D_m$ for $m\in {\cal M}$. Therefore,  ${\cal G}_H \perp  {\cal G}_P$ as claimed and the result follows.
\end{proof}

The next theorem provides closed form expressions for the closest potential and harmonic games with respect to the norm in \eqref{eq:normGames}.
\begin{theorem} \label{theo:closestGames}
Let ${\cal G} \in {\cal G}_{{\cal M},E}$ be a game with utilities
$\{u^m\}_{m\in{\cal M}}$, and let $\phi= \delta_0^\dag Du$. With respect to the norm in
\eqref{eq:normGames},
    \begin{enumerate}
        \item The closest potential game to ${\cal G}$ has utilities $\Pi_m \phi+ (I-\Pi_m){u}^m$ for all $m\in{\cal M},$
        \item The closest harmonic game to ${\cal G}$ has utilities $ u^m - \Pi_m \phi $ for all $m\in{\cal M}.$
    \end{enumerate}
\end{theorem}
\begin{proof}
    By Theorem \ref{theo:orthOfDecomposition}, the harmonic component of $\cal G$ is orthogonal to the space of potential games ${\cal P}\oplus {\cal N}$. Thus, the closest potential game to $\cal G$ has utilities $u^m-u^m_H$, where $\{u^m_H\}_{m\in{\cal M}}$ is the harmonic component of $\cal G$. Similarly, the potential component of $\cal G$ is orthogonal to the space of harmonic games ${\cal H}\oplus {\cal N}$ and thus the closest harmonic game to $\cal G$ has utilities $u^m-u^m_P$, where $\{u^m_P\}_{m\in{\cal M}}$ is the potential component of $\cal G$. Using the closed form expressions for $u^m_P$ and $u^m_H$ from Theorem \ref{theo:decompTheoSpace}, the claim follows.
\end{proof}
Note that the utilities in the closest potential game consist of two parts: the term $\Pi_m \phi$ expresses the preferences  that are captured by the potential function $\phi$, and $(I-\Pi_m) u^m$ corresponds to the nonstrategic component of the original game.  
Similarly, the closest harmonic game differs from the original game by its  potential component, and hence has the same nonstrategic and harmonic components with the original game.
This implies that the projection decomposes the
flows generated by a game to its consistent and inconsistent
components and is closely related to the decomposition of flows to the
orthogonal subspaces of the space of flows provided in the Helmholtz
decomposition.

Analyzing the projection of a game to the space of potential games may
provide useful insights for the original game; see
Section~\ref{se:conclusions} for a description of ongoing and future
work on this direction. We conclude this section by relating the
approximate equilibria of a game to the equilibria of the closest
potential game.

\begin{theorem} \label{cor:epsEq}
Let $\cal G$ be a game, and $\hat{\cal G}$ be its closest potential
game. Assume that $h_m$ denotes the number of
strategies of player $m$, and define   $\alpha\triangleq ||{\cal G}-
\hat{\cal G}||_{{\cal M},E}$.    Then, every $\epsilon_1$-equilibrium of
$\hat{\cal G}$ is an $\epsilon$-equilibrium of $\cal G$ for some
$\epsilon \leq \max_m \frac{2\alpha}{\sqrt{h_m}} +\epsilon_1$ (and viceversa).
\end{theorem}
\begin{proof}
By the definition of the norm, it follows that
\[
|u^k({\bf p})-\hat{u}^k({\bf p})| \leq 
 \frac{1}{\sqrt{h_k}} ||{\cal G}- \hat{\cal G}||_{{\cal M},E}
\leq \max_m \frac{\alpha}{\sqrt{h_m}},
\]
for all $k\in{\cal M}$, ${\bf p} \in E$. Using Lemma
\ref{lemma:epsEqPre}, the result follows.
%it follows that every $\epsilon_0$ equilibrium of
\end{proof}

This result implies that the study and characterization of the
structure of approximate equilibria in an arbitrary game can be
facilitated by making use of the connection between its
$\epsilon$-equilibrium set and the equilibria of its closest potential
game.

\section{Conclusions} 
\label{se:conclusions}

We have introduced a novel and natural direct sum decomposition of the
space of games into potential, harmonic and nonstrategic subspaces. We
studied the equilibrium properties of the subclasses of games induced
by this decomposition, and showed that the potential and harmonic
components of games have quite distinct and appealing equilibrium
properties. In particular, there is a sharp contrast between potential
games, that always have pure Nash equilibria, and harmonic games, that
generically never do. Moreover, we have shown that while the
nonstrategic component does not effect the equilibrium set of games,
it can drastically affect their efficiency properties. Using the
decomposition framework, we obtained closed-form expressions for the
projections of games to their corresponding components, enabling the
approximation of arbitrary games in terms of potential and harmonic
games.  This provides a systematic method for characterizing the set
of $\epsilon$-equilibria by relating it to the equilibria of the
closest potential game.

The framework provided in this paper opens up a number of interesting
research directions, several of which we are currently
investigating. Among them, we mention the following:

\paragraph{Decomposition and dynamics} 
One immediate and promising direction is to use the projection
techniques to analyze natural player dynamics through the convergence
properties of their potential component. It is well-known that in
potential games many natural dynamics, such as best-response and
fictitious play, converge to an equilibrium
\cite{young2004sla,marden2005jsf}. In our companion paper
\cite{Candogan2009}, we show that such dynamics converge to a
neighborhood of equilibria in near-potential games, where the size of
the neighborhood depends on the distance of the original game to its
closest potential game.

\paragraph{Dynamics in harmonic games} 
While the behavior of player dynamics in potential games is reasonably
well-understood, there seems to be a number of interesting research
questions regarding their harmonic counterpart. In
\cite{Candogan2009}, we made some partial progress in this direction,
by showing that in harmonic games, the uniformly mixed strategy
profile is the unique equilibrium of the continuous time
fictitious-play dynamics and this equilibrium point is locally
stable. Moreover, in two-player games where each player has equal
number of strategies, this equilibrium is globally stable. Global
stability of the equilibrium under more general settings and
convergence of different dynamics in harmonic games are open future
questions.

\paragraph{Game approximation}
The idea of analyzing an arbitrary game through a ``nearby'' game with
tractable equilibrium properties seems to be a useful approach. In
\cite{Candogan2009Pricing}, we have developed this methodology
(``near-potential'' games), and applied it in the context of pricing
in a networking application. 
%We studied convex optimization formulations for approximating a game with generalizations of potential games (weighted and ordinal potential games) in .
We believe that these techniques can be
extended to other special classes of games.
An extension to generalizations of potential games (weighted and ordinal potential games) was considered in \cite{candogan2010_weighted}.
%, we introduced a   convex optimization framework to approximate a given game by.
 Another interesting direction
is to study the proximity of an arbitrary game to supermodular games
\cite{topkis1998supermodularity} and analyze how properties of
supermodular games are inherited in ``near-supermodular'' games.

\paragraph{Alternative projections} 
In this work, the projections onto the spaces of potential and
harmonic games are obtained using a weighted $l_2$ norm. This norm
leads to closed-form expressions for the components, but some problems
may require or benefit from projections using different metrics. For
instance, finding the closest potential game, by perturbing each of
the pairwise comparisons in a minimal way requires a projection using
a suitably defined $\infty$-norm, and this projection leads to better
error bounds when analyzing approximate equilibria and
dynamics. Projections under different norms and their properties are
left for future research.

\paragraph{Additional restrictions} 
Another interesting extension is to study projections to subsets of
potential games with additional properties. For example, in the
current projection framework, if we require the potential function to
be concave (see \cite{Ui:2008p1753} for a definition of discrete
concavity), it may be possible to project a given game to a set of
potential games with a unique Nash equilibrium.

\vspace{1cm}
{\small
\noindent\textbf{Acknowledgements:} We thank Prof. Tamer Ba{\c{s}}ar, Prof. Sergiu Hart, Prof. Dov Monderer and Prof. Jeff
Shamma for their useful comments and suggestions.
}

\bibliographystyle{plain}
\bibliography{projectionBibtex}

\newpage

\appendix

\section{Additional Proofs} 
\label{se:proofOfLapProj}

In this section we provide proofs to some of the results from Sections \ref{se:canonicalRep} and \ref{se:decProp}.

\begin{proof}[Proof of Lemmas \ref{theo:projDiff} and \ref{theo:projDiff2}]
The proof relies on the fact that $D_m^* D_m=\Delta_{0,m}$ is a Laplacian operator defined on the graph of $m$-comparable strategy profiles. We show that
the kernels of $D_m$ and $\Delta_{0,m}$ coincide, and using the spectral properties of the Laplacian and projection matrices we obtain the desired result.

For a fixed $m$, it can be seen that strategy profile ${\bf p}=({\bf p}^m,{\bf p}^{-m})$ is comparable to strategy profiles $({\bf q}^m,{\bf p}^{-m})$ for all ${\bf q}^m \in E^m$, ${\bf q}^m\neq{\bf p}^m$ but to none of the strategy profiles $({\bf q}^m,{\bf q}^{-m})$ for ${\bf q}^{-m} \neq {\bf p}^{-m}$. This implies that the graph over which $\Delta_{0,m}$ is defined has $|E^{-m}|= \prod_{k\neq m} h_k$ components (each  ${\bf p}^{-m} \in E^{-m}$ creates a different component), each of which has $|E^m|=h_m$ elements. Note that all strategy profiles in a component are $m$-comparable, thus the underlying graph consists of $|E^{-m}|$ components, each of which is a complete graph with $|E^m|$ nodes.

The Laplacian of an unweighted complete graph with $n$ nodes has eigenvalues $0$ and ${n}$, where the multiplicity of nonzero eigenvalues is $n-1$ \cite{chung1997sgt}. Each component of  $\Delta_{0,m}$ leads to  eigenvalues $0$ and ${h_m}$ with multiplicities $1$ and $h_m-1$ respectively.  Therefore, $\Delta_{0,m}$ has eigenvalues $0$ and ${h_m}$ where the multiplicity of nonzero eigenvalues is $(h_m-1) \prod_{k\neq m} h_k= \prod_{k} h_k -\prod_{k\neq m} h_k$. This suggests that the dimension of the kernel of $\Delta_{0,m}$ is $\prod_{k\neq m} h_k$.

Observe that the kernel of $\Delta_{0,m}=D_m^*D_m$ contains the kernel of $D_m$. For every ${\bf q}^{-m}\in E^{-m}$ define $\nu_{{\bf q}^{-m}} \in C_0$ such that
\begin{equation}
\nu_{{\bf q}^{-m}}({\bf p})=\left\{
\begin{aligned}
& 1 \quad\quad \mbox{if ${\bf p}^{-m}={\bf q}^{-m}$} \\
& 0 \quad\quad \mbox{ otherwise}
\end{aligned}
\right.
\end{equation}
It is easy to see that  $\nu_{{\bf p}^{-m}} \perp \nu_{{\bf q}^{-m}}$  for ${\bf p}^{-m} \neq {\bf q}^{-m}$ and $D_m \nu_{{\bf p}^{-m}}=0$ for all ${{\bf p}^{-m}} \in E^{-m}$. Thus, for all ${\bf q}^{-m}$, $\nu_{{\bf q}^{-m}}$ belongs to the kernel of $D_m$ and by mutual orthogonality of these functions, the kernel of $D_m$ has dimension at least $|E^{-m}|=\prod_{k\neq m} h_k$. As the dimension of the kernel of $\Delta_{0,m}$ is $\prod_{k\neq m} h_k$ and it contains kernel of $D_m$, this implies that the kernels of $D_m$ and $\Delta_{0,m}$ coincide.

Thus $\Delta_{0,m}$ maps any $\nu\in C_0$  in the kernel of $D_m$ to zero and scales the $\nu$ in the orthogonal complement of the kernel by ${h_m}$. On the other hand $D_m^\dagger D_m$
is a projection operator and it has eigenvalue $0$ for all functions in the kernel of $D_m$ and $1$ for the functions in the orthogonal complement of kernel of $D_m$. This implies that
\begin{equation}
\Delta_{0,m}={h_m} D_m^\dagger D_m,
\end{equation}
and the kernels of $\Pi_m$, $D_m$ and $\Delta_{0,m}$ coincide
as the claim suggests.
\end{proof}

\begin{proof}[Proof of Lemma \ref{lem:SpectraOfDelta0}]
For a game, the  graph of comparable  strategy profiles is connected as can be seen from the definition of the comparable strategy profiles.
% and \eqref{eq:WandWm}.
It is known that for a connected graph, the Laplacian operator  has multiplicity $1$ for eigenvalue $0$ \cite{chung1997sgt}.  By \eqref{eq:LapExplicit} it follows that the function $f\in C_0$ satisfying $f({\bf p})=1$ for all ${\bf p} \in E$, is an eigenfunction of $\Delta_0$ with eigenvalue $0$, implying the result.
\end{proof}

\begin{proof}[Proof of Lemma~\ref{lem:identity}]
For the proof of this lemma, we use the following property of the pseudoinverse
\begin{equation} \label{eq:pinvIdentity}
A^\dag = (A^* A)^\dag A^*,
\end{equation}
and the orthogonality properties of the $D_m$ operators: $D_m^* D_k=0$ and $D_m^\dagger D_k =0$ if $m\neq k$.

\paragraph{(i)} Using \eqref{eq:pinvIdentity}, with $A=D_m$ implies that $D_m^\dagger= (D_m^* D_m)^\dagger D_m^*$. Since for any linear operator $L$, $(L^\dagger)^*=(L^*)^\dagger$, it follows that $D_m^\dagger=(D_m (D_m^* D_m)^\dagger )^*=(D_m (\Delta_{0,m})^\dagger )^*$. Hence, using Lemma \ref{theo:projDiff} we obtain $D_m^\dagger=h_m (D_m (\Pi_m)^\dagger )^*$. Since $\Pi_m$ is a projection operator to the orthogonal complement of the kernel of $D_m$, we have $\Pi_m^\dagger=\Pi_m$,  and $D_m \Pi_m = D_m$. Hence, it follows that $D_m^\dagger=h_m (D_m \Pi_m )^*= h_m D_m ^*$ as claimed.

\paragraph{(ii)}
The identity in \eqref{eq:pinvIdentity}, implies that
\[
(\sum_i D_i)^\dagger= \left( (\sum_i D_i)^* (\sum_i D_i) \right)^\dagger (\sum_i D_i)^*.
\]
By the orthogonality of the image spaces of $D_i$, it follows that $(\sum_i D_i)^* (\sum_i D_i) = \sum_i D_i^* D_i$, and hence
\[
(\sum_i D_i)^\dagger= \left( \sum_i D_i^* D_i \right)^\dagger (\sum_i D_i)^*.
\]
Right-multiplying the above equation by $D_j$ and using the
orthogonality of the image spaces of $D_i$s it follows that
\[
(\sum_i D_i)^\dagger  D_j = \left( \sum_i D_i^* D_i \right)^\dagger (\sum_i D_i)^* D_j= 
\left( \sum_i D_i^* D_i \right)^\dagger D_j^* D_j.
\]
\paragraph{(iii)} 
From the definition of pseudoinverse, it is sufficient to show the
following $4$ properties to prove the claim: a) $(D D^\dagger )^*=D
D^\dagger$, b) $(D^\dagger D )^*= D^\dagger D$, c) $D D^\dagger D= D$,
d) $ D^\dagger D D^\dagger= D^\dagger$.

Using the identity $D_m^\dagger D_k=0$ for $k\neq m$, it follows that
$D D^\dagger=\sum_m D_m D_m^\dagger$, and $ D^\dagger D= diag\left(
D_1^\dagger D_1, \dots D_M^\dagger D_M \right)$.  The pseudoinverse of
$D_m$ satisfies the properties $D_m^\dagger D_m=(D_m^\dagger D_m)^*$
and $D_m D_m^\dagger=(D_m D_m^\dagger)^*$, and the requirements a) and
b) follow immediately using these properties.  The identity
$D_m^\dagger D_k=0$ also implies that $D D^\dagger D =[
 D_1 D_1^\dagger D_1, \dots,   D_M D_M^\dagger D_M] $,
and $D^\dagger D D^\dagger =[D_1^\dagger D_1 D_1^\dagger ; \dots ;
D_M^\dagger D_M D_M^\dagger] $. Since the pseudoinverse of $D_m$ also
satisfies $D_m^\dagger D_m D_m^\dagger=D_m^\dagger$, and $D_m
D_m^\dagger D_m=D_m$, the requirements c) and d) are satisfied and the
claim follows.

\begin{comment}
Using \eqref{eq:pinvIdentity} with $A=D$, it follows that $D^\dagger=
(D^* D)^\dagger D^*$. Observe that $D=[D_1, \dots, D_M]$, and
$D^*=[D_1^*; \dots; D_M^*]$. The orthogonality of the image spaces of
$D_k$ operators and (i) imply that $D^* D$ is a block diagonal
operator such that
$$D^*D=diag(D_1^* D_1, \dots, D_M^*D_M )= diag(h_1 D_1^\dag D_1, \dots, h_M D_M^\dag D_M ).$$ Since for each $m$, $D_m^\dag D_m$ is a projection operator, it follows that $(D^*D)^\dag=  diag(\frac{1}{h_1} D_1^\dag D_1, \dots \frac{1}{h_M} D_M^\dag D_M )$. Therefore, we obtain
\begin{equation} \label{eq:DIdentity}
D^\dagger= (D^* D)^\dagger D^*= \left[\frac{1}{h_1} D_1^\dag D_1 D_1^*, \dots, \frac{1}{h_M} D_M^\dag D_M D_M^* \right]
\end{equation}
Using the identity $ D_m^\dag D_m D_m^\dag = D_m$ of the pseudoinverse and $D_m^*=h_m D_m^\dag$, the claim follows from \eqref{eq:DIdentity}.
\end{comment}

\paragraph{(iv)}  Since $\Pi= diag\left( \Pi_1,  \dots   \Pi_M   \right)$, and $\Pi_m = D_m^\dagger D_m$, it follows that
$ D^\dagger D=  diag\left( D_1^\dagger D_1,  \dots    D_M^\dagger D_M   \right)= diag\left( \Pi_1,  \dots   \Pi_M   \right)=\Pi$.
\paragraph{(v)} 
Using the identities $D_m^\dagger D_k=0$ for $k\neq m$,  $\delta_0=\sum_m D_m$, it follows that
\begin{equation*}
D D^\dagger \delta_0 =D D^\dagger \sum_{m \in \mathcal{M}}   D_m= \sum_{m \in \mathcal{M}}  D_m D_m^\dagger D_m= \sum_{m \in \mathcal{M}}  D_m =\delta_0.
\end{equation*}

\begin{comment}
From (iii), it follows that
\begin{equation}
D D^\dag \delta_0= [D_1,\ldots,D_M] [D_1^\dag; \ldots; D_M^\dag] \sum_{m \in \mathcal{M}} D_m  
= \sum_{m \in \mathcal{M}}  \sum_{k\in {\cal M}} \left(D_k D_k^\dag\right) D_m.
\label{eq:ddd}
\end{equation}
Since $D_k^\dag = \frac{1}{h_k} D_k^*$ and image spaces of $D_k$
operators are orthogonal, we obtain $\sum_{k\in {\cal M}} \left(D_k
D_k^\dag\right) D_m= D_m D_m^\dag D_m=D_m$, where the last equality is
a property of the pseudoinverse.  Thus, the claim follows from
\eqref{eq:ddd}, by noting that $\delta_0=\sum_m D_m$
\end{comment}
\end{proof}

\begin{proof}[Proof of Lemma \ref{lem:harmGameFlux}]
Let $X=Du$ denote the pairwise comparison function of the harmonic game. By definition,   $(\delta_0^* X)({\bf p})=0$  for all ${\bf p}\in E$.
Thus, for all ${\bf r}^m\in E^m$, it follows that
\begin{equation} \label{eq:fluxLastStep}
\begin{aligned}
0=\sum_{{\bf p}^{-m}\in E^{-m}}(\delta_0^* X)({\bf r}^m,{\bf p}^{-m}) %&=-\sum_{{\bf p}^{-m}\in E^{-m}}\sum_{{\bf r}|( {\bf p},{\bf r})\in A^m  }X({\bf p},{\bf r}) \\
&=  \sum_{{\bf p}\in S}(\delta_0^* X)({\bf p}) 
\end{aligned}
\end{equation}
where  $S=\{({\bf r}^{m},{\bf p}^{-m})| {\bf p}^{-m}    \in E^{-m} \}$.
To complete the proof we require the following identity related to the pairwise comparison functions.
\begin{lemma} \label{lem:flux}
For all $\hat{X}\in C_1$ and set of strategy profiles $\hat{S}\subset E$, $\sum_{{\bf p}\in \hat{S}} (\delta_0^* \hat{X})({\bf p})=-\sum_{{\bf p}\in \hat{S}}\sum_{{\bf q}\in \hat{S}^c} \hat{X}({\bf p},{\bf q})$.
\end{lemma}
\begin{proof}
It follows from the definition of $\delta_0^*$ that
\begin{equation}
\begin{aligned}
\sum_{{\bf p}\in \hat{S}} (\delta_0^* \hat{X})({\bf p}) &=-\sum_{{\bf p}\in \hat{S}}\sum_{{\bf q}\in E} \hat{X}({\bf p},{\bf q})\\
&=-\sum_{{\bf p}\in \hat{S}}\sum_{{\bf q}\in \hat{S}^c} \hat{X}({\bf p},{\bf q})-\sum_{{\bf p}\in \hat{S}}\sum_{{\bf q}\in \hat{S}} \hat{X}({\bf p},{\bf q}) \\
&=-\sum_{{\bf p}\in \hat{S}}\sum_{{\bf q}\in \hat{S}^c} \hat{X}({\bf p},{\bf q}).
\end{aligned}
\end{equation}
%where the last line follows
since $\hat{X}({\bf p},{\bf q})+\hat{X}({\bf q},{\bf p})=0$ for any ${\bf p},{\bf q}$ and thus $\sum_{{\bf p}\in \hat{S}}\sum_{{\bf q}\in \hat{S}} \hat{X}({\bf p},{\bf q})=0$.
\end{proof}
Using this lemma in \eqref{eq:fluxLastStep}, we obtain
\begin{equation}
\begin{aligned}
0&=-\sum_{{\bf p}\in S}\sum_{\hat{\bf p}\in S^c} X({\bf p},\hat{\bf p}) \\
&=\sum_{{\bf p}^{-m}\in E^{-m}}\sum_{{\bf p}^m\in E^m} {u}^m({\bf r}^{m},{\bf p}^{-m})-{u}^m({\bf p}^{m},{\bf p}^{-m}).
\end{aligned}
\end{equation}
Since ${\bf r}^{m}$ is arbitrary,   it follows that
\begin{equation}
\sum_{{\bf p}^{-m}\in E^{-m}} {u}^m({\bf q}^{m},{\bf p}^{-m})-{u}^m({\bf r}^{m},{\bf p}^{-m})=0.
\end{equation}
for all ${\bf q}^{m}, {\bf r}^{m}\in E^m$.
\end{proof}

\begin{proof}[Proof of Theorem \ref{theo:ZSPH_dim}]
Since ${\cal Z} \oplus {\cal I}= {\cal G}_{{\cal M},E}$, the last column immediately follows from Proposition~\ref{cor:potDimFin} and Theorem~\ref{theo:decompTheoSpace}.
% the dimensions and orthogonality of  $\cal P$, $\cal H$, $\cal N$. 
Below,  we present the dimension results  for each row of the table, and the corresponding entries in the first two columns. 

Throughout the proof we denote by $e$ the $h$ dimensional vector of ones.  Since ${\cal N}=\ker D $, the
two-player games in $\cal N$ take the form 
%$(a e^T, e b^T)$ for some
$(e a^T,  b e^T)$ for some
$a, b\in \mathbb{R}^h$. We shall make use of this fact in the proof.

%In the following, we denote the  payoffs of games in $\cal I$ by$(v,v)$ and payoffs of  the games in $\cal Z$ by $(u,-u)$.

${\cal P}\oplus {\cal H} \oplus{ \cal N}$: Since ${\cal P}\oplus {\cal H} \oplus{ \cal N}={\cal G}_{{\cal M},E}$, it follows that $\dim ({\cal P}\oplus {\cal H} \oplus{ \cal N}) \cap {\cal Z}= \dim {\cal Z} $ and   $\dim ({\cal P}\oplus {\cal H} \oplus{ \cal N}) \cap {\cal I}= \dim {\cal I} $.
Zero-sum games are  games with payoff matrices $(A, -A)$ for some $A\in {\mathbb R}^{h\times h}$.  Thus, the dimension of the zero-sum games is equivalent to the dimension of possible $A$ matrices that define zero-sum games and hence $\dim {\cal Z} = h^2$.
Similarly, identical interest games are games with payoff matrices $(A,A)$ for some $A\in {\mathbb R}^{h\times h}$, and hence $\dim {\cal I} = h^2$.

${\cal P}\oplus {\cal N}$: By Theorem \ref{theo:setOfPot}, it follows that ${\cal P}\oplus {\cal N}$ is equivalent to the set of potential games. Observe that all identical interest games are potential games, where the utility functions of players are equal to the potential function of the game. Thus, it follows that $\dim ({\cal P}\oplus {\cal N}) \cap {\cal I}= \dim {\cal I}= h^2$.

Let $\cal G$ denote a zero-sum game in ${\cal P}\oplus {\cal N}$, with payoff matrices $(A,-A)$, and denote the matrix corresponding to a potential function of $\cal G$ by $\phi$.  Thus, both the game with payoffs  $(A,-A)$  and  $(\phi,\phi)$ belong to ${\cal N}\oplus {\cal P}$, and %have the same potential component. Hence, 
 $(A,-A)$  is different from   $(\phi,\phi)$ by its nonstrategic component.
Hence,  for some $a, b\in \mathbb{R}^h$,
%Since, actual payoffs are different than the potential function, by nonstrategic change in the payoffs it follows that 
 $A= \phi +  e a^T$, $-A=  \phi  +  b e^T$,  for some $a,b \in \mathbb{R}^h$ and
\begin{equation*}
A-A= \phi +  e a^T+  \phi +    b e^T= 2 \phi +  e a^T +  b e^T =0,
\end{equation*}
thus  $-2 \phi_{ij}= a_j+ b_i$  and $A_{ij}= \phi_{ij} + a_j= \frac{a_j - b_i}{2}$ for all $i,j\in \{1 \dots n\}$. 
Hence, $a,b\in \mathbb{R}^h$  characterize the possible payoff matrices $A$, and it can be seen that   the set of these matrices has dimension  $2h-1$. Since these matrices uniquely characterize zero-sum games that are also potential games, it follows that
the dimension of  $({\cal P}\oplus {\cal N}) \cap {\cal Z}$ is equal to $2h-1$.

${\cal H}\oplus {\cal N}$:  The games in this set do not have potential components.
If a game in ${\cal H}\oplus {\cal N}$ is an identical interest game, then it also belongs to ${\cal P}\oplus {\cal N}$. 
Due to the direct sum property of ${\cal P}\oplus{\cal H}\oplus {\cal N}$, it follows that this game can only have nonstrategic component. Therefore, $\dim ({\cal H}\oplus {\cal N}) \cap {\cal I}= \dim  {\cal N} \cap {\cal I} $. 
Let $\cal G$ denote a game in ${\cal N} \cap {\cal I}$.  Since $\cal G$ has only nonstrategic information it follows that its payoffs are given by $( e a^T,  b e^T)$, for some vectors $a$ and $b$.
Then, being an identical interest game implies that
$e a^T=  b e^T$, which requires that all entries of payoff matrices are identical, thus $\dim {\cal N} \cap {\cal I} = 1$.

Consider a zero-sum game in ${\cal H}\oplus {\cal N}$, with payoff matrices $(A,-A)$. Since, both players have equal number of strategies, the harmonic component of this game is also zero-sum and the  payoff matrices in the harmonic component can be denoted by $(A_H,-A_H)$ for some $A_H \in \mathbb{R}^{h\times h}$. 
Because the original game is in ${\cal H}\oplus {\cal N}$, the payoff matrices satisfy
$A=A_H +  e a^T$, $-A=-A_H+ be^T$,  where $(e a^T, b e^T)$ corresponds to the nonstrategic component of the game. It follows that $ e a^T+  b e^T = 0$,  and hence  $e a^T$ and $- b e^T$ are matrices, which have all of their entries identical.
Thus, the nonstrategic component of  the games in $ ({\cal H}\oplus {\cal N}) \cap {\cal Z}$, forms a $1$ dimensional subspace.
Since  the harmonic component is arbitrary, 
it follows that 
%$({\cal H}\oplus {\cal N}) \cap {\cal Z} = {\cal H} \oplus \tilde{\cal N}$. Consequently, $\dim \hat{\cal N}=1$ and
 $\dim ({\cal H}\oplus {\cal N}) \cap {\cal Z} =\dim {\cal H}+ 1= (h-1)^2+ 1= h^2-2h+2$.
\end{proof}

\end{document}